\documentclass[11pt]{article}
\pdfoutput=1
\usepackage{tikz}

\usepackage[T1]{fontenc}
\usepackage{array}
\usepackage{fullpage}
\usepackage{amsmath,amsfonts,amssymb,graphicx,color,aas_macros}
\DeclareMathOperator{\Tr}{tr}
\usepackage{authblk}
\usepackage{float}
\usepackage{amsthm}

\usepackage{listings}
\usepackage{mathrsfs}
\usepackage{stmaryrd}
\usepackage{algorithm}
\usepackage{algpseudocode}
\usepackage{feyn}
\usepackage[unicode=true,bookmarks=true,
bookmarksnumbered=false,bookmarksopen=false,
breaklinks=false,
pdfborder={0 0 1},
backref=false,colorlinks=true]{hyperref}
\usepackage{xcolor}
\usepackage{verbatim}
\usepackage{cite}
\numberwithin{equation}{section}

\hypersetup{pdftitle={Quantum State Reduction: Generalized Bipartitions from Algebras of Observables},
 pdfauthor={Oleg Kabernik, Jason Pollack, and Ashmeet Singh},
 citecolor=black,linkcolor=black,urlcolor=black}

\newcommand{\ket}[1]{\left| #1 \right\rangle}
\newcommand{\bra}[1]{\left\langle #1 \right |}
\def\({\left(}
\def\){\right)}
\def\[{\left[}
\def\]{\right]}

\newcommand{\be}{\begin{equation}}
\newcommand{\ee}{\end{equation}}
\newcommand{\Eq}[1]{Eq.~(\ref{#1})}

\newcommand{\mrm}[1]{\mathrm{#1}}

\newcommand{\Hil}{\mathcal{H}}
\newcommand{\Ham}{H}

\newcommand{\eye}{I}
\newcommand{\lcb}{\left[}
\newcommand{\rcb}{\right]}

\newcommand{\Dim}{\textrm{dim\,}}

\newtheorem{thm}{Theorem}[section]
\newtheorem{prop}[thm]{Proposition}
\newtheorem{lem}[thm]{Lemma}

\newtheorem{defn}[thm]{Definition}

\newtheorem{example}[thm]{Example}

\newcommand{\braket}[2]{\langle #1 \vphantom{#2}|
   #2 \vphantom{#1} \rangle}

\newcommand{\ketbra}[2]{|#1\rangle \langle #2 |}

\begin{document}

\interfootnotelinepenalty=10000

\hfill CALT-TH-2019-36
\hfill

\vspace{2cm}
\thispagestyle{empty}
\begin{center}
{\LARGE \bf
Quantum State Reduction: Generalized Bipartitions from Algebras of Observables
}\\
\bigskip\vspace{1cm}{
{\large Oleg Kabernik,${}^a$ Jason Pollack,${}^a$ and Ashmeet Singh${}^b$}
} \\[7mm]
 {\it ${}^a$Department of Physics and Astronomy\\University of British Columbia, Vancouver, BC, V6T 1Z1, Canada\\[1.5mm]
${}^b$Walter Burke Institute for Theoretical Physics \\
    California Institute of Technology, Pasadena, CA 91125, USA} \let\thefootnote\relax\footnote{\noindent e-mail:\\\url{olegk@phas.ubc.ca}, \url{jpollack@phas.ubc.ca}, \url{ashmeet@caltech.edu}} \\
 \end{center}
\bigskip
\centerline{\large\bf Abstract}

\begin{quote} \small

Reduced density matrices are a powerful tool in the analysis of entanglement structure, approximate or coarse-grained dynamics, decoherence, and the emergence of classicality. It is straightforward to produce a reduced density matrix with the partial-trace map by ``tracing out'' part of the quantum state, but in many natural situations this reduction may not be achievable. We investigate the general problem of identifying how the quantum state reduces given a restriction on the observables. For example, in an experimental setting, the set of observables that can actually be measured is usually modest (compared to the set of all possible observables) and their resolution is limited. In such situations, the appropriate state-reduction map can be defined via a \emph{generalized bipartition}, which is associated with the structure of irreducible representations of the algebra generated by the restricted set of observables. One of our main technical results is a general, not inherently numeric, algorithm for finding irreducible representations of matrix algebras. We demonstrate the viability of this approach with two examples of limited--resolution observables.  
The definition of quantum state reductions can also be extended beyond algebras of observables. To accomplish this task we introduce a more flexible notion of bipartition, the \textit{partial bipartition}, which describes coarse-grainings preserving information about a limited set (not necessarily algebra) of observables. We describe a variational method to choose the coarse-grainings most compatible with a specified Hamiltonian, which exhibit emergent classicality in the reduced state space. We apply this construction to the concrete example of the 1-D Ising model. 
Our results have relevance for quantum information, bulk reconstruction in holography, and quantum gravity.

\end{quote} 

\setcounter{footnote}{0}

\newpage
\tableofcontents
\newpage

\section{Introduction and Motivation}

How do we describe the state of a system about which we have only limited information? 
In the most general form, this is a question for probabilists: the best that can be done in the Bayesian approach, for example, is to make our best guess in the form of a distribution over the possible states of the system compatible with what is already known and update this guess as we learn new information.
In physical applications, however, we typically encounter situations in which  we can only make certain types of measurements on a system.
For example, we might only be able to measure extensive, macroscopic quantities of a gas; or we might be able to probe only classical observables of a quantum system.

In classical statistical mechanics one usually proceeds by enumerating the possible ``microstates'' of the underlying microphysical system (for example, a gas of $N$ point particles in a finite volume with positions, momenta, and possible interactions). Then we partition the microstates into ``macrostates'' by collecting together the states with approximately the same values of some coarse-grained extensive property which probes the average behavior of the particles (for example, temperature, or some hydrodynamic quantity like viscosity). 
In other words, we choose a particular statistical ensemble, write the appropriate partition function, and use it as the generating functional for macroscopic observables. 
When certain assumptions are valid, it is then valid to track the values of the macroscopic quantities without reference to the underlying microscopic physics.
These assumptions have to do with compatibility between the macroscopic observables and the microscopic dynamics of the theory. We want the values of the macroscopic variables to evolve continuously in time, which requires that the time evolution of a macrostate to itself be a macrostate to some approximation; that is, if two microstates are in the same macrostate at one time, there should exist another set of macrostates for the system at a later time such that the time-evolved microstates will usuallly be in the same new macrostate.
Of course, this picture can be generalized in various ways by relaxing some of the assumptions, or by working with probability distributions over microstates instead of partititions \cite{Jarzynski:1996oqb,1998JSP....90.1481C,Bartolotta:2015toa}. 

In quantum mechanics, the story is usually told differently. 
Given a Hilbert space $\Hil$ we can work with pure states $\ket{\psi} \in \Hil$ or mixed states $\rho \in \mathcal{L}(\Hil)$, which can be thought of as classical statistical mixtures of the states $\{\ket{\psi}_i\}$ in the basis $\{\ket{\psi}_i\bra{\psi}_i\}$ in which $\rho$ is diagonal.
When the Hilbert space has a tensor-product structure, $\Hil \cong \Hil_A \otimes \Hil_{\bar A}$, there is a natural state-reduction map, the partial-trace map $\Tr_{\bar A}$, which maps mixed states in $\mathcal{L}(\Hil)$ to mixed states in $\mathcal{L}(\Hil_A)$ via $\rho \mapsto \rho_A \equiv  \Tr_{\bar A} \rho$.
Then the reduced state $\rho_A$ preserves information about operators acting only on $\Hil_A$, in the sense that the expectation value of ${O}_A \otimes  I_{\bar A}$, with $ I_{\bar A}$ the identity operator on $\Hil_{\bar A}$, acting on $\rho$ is the same as the expectation value of ${O}_A$ acting on $\rho_A$, for all states $\rho$ and linear operators ${O}_A$.

So far this picture seems quite different from the classical one summarized above. 
Certainly, if we have a $3N$ or $6N$-dimensional configuration space or phase space, we can consider the reduced spaces generated by projection onto some lower-dimensional subspace.
We can then ask the question of what the reduced dynamics in this subspace look like. 
In particular, we might find that the new dynamics is dissipative, if the particles traced out act as a heat bath for the ones kept in the description, or, in the opposite extreme, that the kept particles only act amongst themselves and can be described without reference to the remainder.
If we didn't actually know which coordinates in the phase space corresponded to the positions or momenta of individual particles, we might hope to identify them by looking for subspace projections with particularly simple reduced dynamics.

The quantum analog of this process is known as the \textit{decoherence program} \cite{Zeh:1970fop,Zurek:1981xq,Griffiths:1984rx,Joos:1984uk,Schlosshauer:2003zy}.
In this program, one is given, or looks for, decompositions of $\Hil$ into a system and environment, $\Hil \cong \Hil_S \otimes \Hil_E$. 
This induces a decomposition of the Hamiltonian $\Ham = \Ham_S \otimes  I_E + \Ham_{\rm int} +  I_S \otimes \Ham_E$.
For certain choices of the Hamiltonian and sets of initial states---for example, interaction-dominated Hamiltonians and initial product states---the action of the environment, to a good approximation, is to take an initial state of the system to a superposition of system states, in some basis, which evolve without interfering. 
When this happens, we say that the initial state has \textit{branched}, and the set of system states whose evolution is preserved by the environment are the \textit{classical states} of the system.
(We will review the decoherence program in more detail below.)

However, it is easy to see that most coarse-grainings cannot be described in the decoherence picture. 
Most observables do not take the simple form of acting on a single tensor factor, even when such a factorization of the Hilbert space exists.
In particular, the sorts of \textit{collective} observables which correspond to the averaged, macroscopic properties featured in statistical mechanics do not take this form.
That is, we do \textit{not} expect, even approximately, a factorization of the form $\Hil \cong \Hil_\mathrm{collective} \otimes \Hil_\mathrm{other}$ for the sorts of collective observables we might measure in a laboratory. 

A similar situation arises in field theories, in which we often wish to construct some notion of a state restricted to a finite spatial region. 
It is well known \cite{Reeh1961,Witten:2018lha} that even in the simplest field theories we cannot simply apply the naive partial-trace map to construct the reduced state as discussed above.
There is, nevertheless, a good notion of algebras of observables restricted to a spatial region, which is provided by modular theory (e.g.\ \cite{Araki:1976zv,Haag:1992hx,Lashkari:2018nsl}), and in many cases we can pass to a (finite-dimensional) latticization, for example a tensor network, in which these issues do not arise.
When the theory has a gauge symmetry, however, the physical Hilbert space is restricted to states which obey global constraints like a Gauss law, and we cannot consistently restrict to subregions in a gauge-invariant way. 
The approach of the edge modes program \cite{Donnelly:2014gva,Donnelly:2014fua,Donnelly:2015hta,Lin:2018bud} is to embed the physical Hilbert space into a larger, ``ungauged'' Hilbert space in which the constraints have been removed and subregions are well-defined.

Given that many natural coarse-grainings of quantum systems cannot be captured by the partial-trace map, it is natural to consider more general state-reduction maps.
It is only when such a map can be constructed from a physically-motivated coarse-graining that we are furnished with a true reduced density matrix to which we can apply the well-developed machinery of decoherence, von Neumann entropy, etc.
The main goal of this paper is to provide such an interpretation for a large class of general quantum coarse-grainings.
 
We will provide an algorithm which takes a (finite) set of observables on a (finite-dimensional) Hilbert space and outputs a decomposition of the Hilbert space into irreducible representations of the algebra generated by the observables. Such a decomposition will be called a \emph{generalized bipartition}. The state-reduction map is then defined by tracing out tensor factors of subspaces that appear in this decomposition which is not equivalent to a partial-trace of any single tensor factor of the original Hilbert space. 
However, like the usual partial-trace map, such state reductions preserve the expectation values of all observables in this algebra. Furthermore,  unitary dynamics on $\Hil$ will induce some (typically) non-unitary dynamics on the reduced state so, as with the usual partial-trace reductions, we can perform a decoherence analysis to determine what observables behave classically.

There are many cases in which a coarse-graining is operationally well-described as having access to all elements in a subalgebra of observables. 
In some cases, however, it is more appropriate to consider only a restricted \emph{set} of observables which need not comprise an algebra. 
Classical experimenters, for example, though they might be able to devise setups to measure the (coarse-grained) position and momentum of some system in a lab, would have trouble implementing arbitary superpositions of these operators.
We are thus motivated to define \emph{partial bipartitions}, which implement more general state-reduction maps.
Partial bipartitions are best-suited to a \emph{variational} approach, in which one scans over possible coarse-grainings with the goal of determing which restricted set of observables is ``most classical'' \cite{Pollack:2018yum,mereology}. 

\subsection{Summary of Results}

Because of the very general nature of our subject we have chosen to make this paper as self-contained as possible, often at the expense of brevity. 
In this subsection, we summarize the explicit results of the paper for the benefit of the busy reader. 

\begin{itemize}
\item A \emph{generalized bipartition} \eqref{eq:def generalized bipartition} is a direct-sum decomposition of a Hilbert space $\Hil$ into a sum of bipartite blocks $\Hil_q \cong \mathcal{H}_{A_{q}}\otimes\mathcal{H}_{B_{q}}$:
\be
\mathcal{H}\cong\bigoplus_{q}\Hil_q\cong\bigoplus_{q}\mathcal{H}_{A_{q}}\otimes\mathcal{H}_{B_{q}},
\ee
where each sector $\Hil_q$ is spanned by a set of basis elements $\{\ket{e^q_{ik}}\}$ and the isometry between $\Hil_q$ and $\mathcal{H}_{A_{q}}\otimes\mathcal{H}_{B_{q}}$ maps the basis element $\ket{e^q_{ik}}$ to the product state $\ket{a^q_{i}}\ket{b^q_k}$, with $\{a^q_{i}\}$ and $\{b^q_{k}\}$ respectively bases for $\mathcal{H}_{A_{q}}$, $\mathcal{H}_{B_{q}}$.
The index structure of the $\ket{e^q_{ik}}$ can be conveniently represented as a block-diagonal table, which we refer to as a \emph{bipartition table} \eqref{eq:generic block diag BPT}: 

\noindent %
\noindent\begin{minipage}[c]{1\columnwidth}%
\begin{tabular}{>{\centering}m{0.9\columnwidth}>{\raggedright}m{0.1\columnwidth}}
 & \tabularnewline
\centering{}%
\begin{tabular}{|c|c|c|ccc}
\cline{1-3} \cline{2-3} \cline{3-3} 
$e_{11}^{1}$ & $e_{12}^{1}$ & $\cdots$ &  &  & \tabularnewline
\cline{1-3} \cline{2-3} \cline{3-3} 
$e_{21}^{1}$ & $e_{22}^{1}$ & $\cdots$ &  &  & \tabularnewline
\cline{1-3} \cline{2-3} \cline{3-3} 
$\vdots$ & $\vdots$ & $\ddots$ &  &  & \tabularnewline
\cline{1-5} \cline{2-5} \cline{3-5} \cline{4-5} \cline{5-5} 
\multicolumn{1}{c}{} & \multicolumn{1}{c}{} &  & \multicolumn{1}{c|}{$e_{11}^{2}$} & \multicolumn{1}{c|}{$\cdots$} & \tabularnewline
\cline{4-5} \cline{5-5} 
\multicolumn{1}{c}{} & \multicolumn{1}{c}{} &  & \multicolumn{1}{c|}{$\vdots$} & \multicolumn{1}{c|}{$\ddots$} & \tabularnewline
\cline{4-5} \cline{5-5} 
\multicolumn{1}{c}{} & \multicolumn{1}{c}{} & \multicolumn{1}{c}{} &  &  & $\ddots$\tabularnewline
\end{tabular} & %
\begin{minipage}[c]{0.04\columnwidth}%
\begin{center}
\begin{equation}
\end{equation}
\par\end{center}
~%
\end{minipage}\tabularnewline
 & \tabularnewline
\end{tabular}%
\end{minipage}
The upper index of $\ket{e^q_{ik}}$ labels the block, and the lower indices label position within the block.

\item Generalized bipartitions are interesting for (at least) two reasons.
First, they provide a natural way of talking about ``the degrees of freedom in B'' and ``the set of measurements which can be peformed on B.''
In particular, consider the \emph{bipartition operators} \eqref{eq:BPO def}:
\begin{equation}
S_{kl}^{q}:=I_{A_{q}}\otimes\ket{b_{k}^{q}}\bra{b_{l}^{q}}=\sum_{i}\ket{e_{ik}^{q}}\bra{e_{il}^{q}}.
\end{equation}
The linear combinations of the $S_{kl}^{q}$ comprise the space of linear operators that act on a Hilbert space (isomorphic to) $\Hil_B := \bigoplus_q \Hil_{B_q}$. 
The bipartition operators can therefore be used to define a state-reduction map $tr_{(A)}$, distinct from the standard partial-trace map $tr_A$, from $\Hil$ to $\Hil_B$ \eqref{eq:tr_(A)  action with S_kl}:
\begin{align}
\rho_B = tr_{\left(A\right)}\left(\rho\right) & :=\sum_q\sum_{k,l}tr\left(S_{kl}^q\rho\right)\left|b^q_{l}\right\rangle \left\langle b^q_{k}\right| = \sum_q tr_{A_q}\(\rho_q\) \in \mathcal L(\Hil_B),\label{eq:intro_state_red}
\end{align}
where $\rho_q$ is the projection of the state $\rho$ onto the sector $\Hil_q$.
(There is an analogous state-reduction map onto $\Hil_A$ produced from the dual generalized bipartition, which represents the isometry $\mathcal{H}\cong\bigoplus_{q}\mathcal{H}_{B_{q}}\otimes\mathcal{H}_{A_{q}}$; its bipartition table is constructed by taking the transpose of each block in the original table.)

\item Second, generalized bipartitions are interesting because they appear in the foundational result of the representation theory of operator algebras, the Wedderburn decomposition theorem (Theorem \ref{thm:Wedderburn decomposition}). 
In our language, subject to technical details which we discuss in the main presentation of the theorem below, the decomposition theorem says that any subalgebra $\mathcal{A}$ of $\mathcal{L}(\Hil)$ induces a generalized bipartition of $\Hil$, such that the subalgebra is identical to the set of operators which are linear combinations of the bipartition operators $S^q_{kl}$ which act on $\Hil_B$ alone.
The generalized bipartition thus provides a decomposition of $\Hil$ into irreducible representations of $\mathcal{A}$.
That is, any subalgebra furnishes a generalized bipartition, and any identification of degrees of freedom given by a bipartition table defines a subalgebra.
We emphasize that the decomposition theorem is not constructive: it says only that given a subalgebra such a decomposition must exist.

\item The main technical accomplishment of the paper is to provide an explicit construction of the generalized bipartition (that is, the irrep decomposition) of the (finitely generated) algebra $\mathcal{A}$. This is accomplished by Algorithm \ref{alg:Irrepd-decomposition}, whose correctness is established in Theorem \ref{full theorem} via a number of intermediate lemmas.
We refer the reader to Section \ref{sec:Irrep-Decompositions-of} for details. 
The main idea of the algorithm is based on the fact that projections whose rank cannot be reduced within the algebra are the fundamental building blocks of the algebra. Such minimal projections can be distilled from the initial spectral projections of the generators by breaking them into projections of smaller rank with an operation we call \textit{scattering}:
\begin{equation}
\begin{array}{c}
\Pi_{1}\\
\\
\Pi_{2}
\end{array}
\Diagram{fdA &  & fuA\\
	& f\\
	fuA &  & fdA
}
\begin{array}{c}
\Pi_{1}^{\left(\lambda_{1}\right)}+\Pi_{1}^{\left(\lambda_{2}\right)}+...+\Pi_{1}^{\left(0\right)}\\
\\
\Pi_{2}^{\left(\lambda_{1}\right)}+\Pi_{2}^{\left(\lambda_{2}\right)}+...+\Pi_{2}^{\left(0\right)}.
\end{array}
\end{equation}
The result on the right-hand side of this operation is given by the
spectral decomposition of the operator $\Pi_{1}\Pi_{2}\Pi_{1}$.
Once all projections have been scattered into minimal projections, we consider a graph, which we call a \textit{reflection network}, that consists of the minimal projections as vertices with edges defined by their orthogonality relations. 
Under certain conditions, such a reflection network naturally corresponds to a bipartition table. We  leverage this correspondence to identify the irrep decomposition with this bipartition table. 

\item The main application of the algorithm that we will focus on is the
idea that operational constraints lead to state reductions. The prototypical
example of that is the system-environment split in the context of
the decoherence program. There, the operational constraints are defined
by the observer's inability to control or measure the environment
which leads to the state-reduction map implemented by tracing out
the environment. In Section \ref{subsec: State Reduction from Operational Constraints}
we formalize the idea that any operational constraints given by some
restricted set of observables, lead to a state-reduction map; this
is what we call \emph{operational approach to decoherence}. The correspondence
between operational constraints and state reductions is obtained by
constructing the generalized bipartition associated with the algebra
of restricted observables.
\item In the context of the operational approach, we will
study two, relatively straightforward, examples of state reductions.
One of the examples is concerned with the operational constraints
of an observer unable to distinguish spin and orbital angular
momentum components; this leads to superselection of the total
angular momentum sectors. This example is interesting not because
of the conclusion -- it can be deduced from the standard formalism
of angular momentum addition -- but because we can reach this conclusion
independently by analytically applying our algorithm. Remarkably,
even the correct Clebsch-Gordan coefficients come out as byproducts
of this construction. The second example finds the state reduction
map corresponding to an observer's inability to resolve a bound pair
of particles on a lattice. This example also results in superselection
but in this case the two sectors are the symmetric and the anti-symmetric
configurations of the pair.

\item The machinery of bipartition tables can be applied more generally than matrix algebras or generalized bipartitions.
In particular, the state-reduction map $tr_{(A)}$ still produces a valid reduced state in $\Hil_B$ if some of the entries in the bipartition table are removed. 
The resulting bipartition table, which defines a \emph{partial bipartition} \eqref{eq:pb}, is still block-diagonal but not all of the blocks are rectangular:
\begin{center}
\begin{tabular}{ccc}
\begin{tabular}{|c|c|c}
\hline 
$e_{1;1,1}$ & $e_{1;1,2}$ & \multicolumn{1}{c|}{$...$}\tabularnewline
\hline 
$e_{1;2,1}$ & $\ddots$ & \tabularnewline
\cline{1-2} 
$\vdots$ & \multicolumn{1}{c}{} & \tabularnewline
\cline{1-1} 
\end{tabular} &  & \tabularnewline
 & %
\begin{tabular}{|c|c|c}
\hline 
$e_{2;1,1}$ & $e_{2;1,2}$ & \multicolumn{1}{c|}{$...$}\tabularnewline
\hline 
$e_{2;2,1}$ & $\ddots$ & \tabularnewline
\cline{1-2} 
$\vdots$ & \multicolumn{1}{c}{} & \tabularnewline
\cline{1-1} 
\end{tabular} & \tabularnewline
 &  & $\ddots$\tabularnewline
\end{tabular}\be\ee\
\par\end{center} 
The bipartition operators still correspond to the spanning set of all linear operators in this reduced space, but, in general they no longer span an algebra.
In particular, the last equality in \eqref{eq:intro_state_red} does not hold for a non-rectangular block.
Hence the state-reduction map is not related to the usual partial-trace, since $\Hil_{B_q}$ need not be a tensor factor of $\Hil_q$; we instead say that $\Hil_{B_q}$ is a \emph{partial subsystem} of $\Hil_q$ and write
\be
\Hil_q \cong \Hil_{A_q} \oslash \Hil_{B_q}.
\ee
The same relation holds\footnote{Formally, we can embed $\Hil$ into the larger Hilbert space 
\be
\Hil_A\otimes\Hil_B:=\(\bigoplus_q \Hil_{A_q}\)\otimes\(\bigoplus_q \Hil_{B_q}\)=\bigoplus_{q,q^\prime}\Hil_{A_q}\otimes\Hil_{B_{q^\prime}},
\ee
so that $\Hil$ comprises the diagonal entries $q=q^\prime$, and then the partial-trace map $tr_A$ on this bipartite Hilbert space does indeed map those states in $\Hil_A\otimes\Hil_B$ supported on $\Hil$ to states on $\Hil_B$. Hence we can obtain the reduced density matrix $\rho_B$ by tracing out degrees of freedom, at the cost of working with a larger, auxilliary Hilbert space. As we will discuss below, this procedure is closely related to passing from the physical to the ``ungauged'' Hilbert space when computing the entropy of subregions of states in theories with gauge symmetries.\label{fn:embedding}} between the collection of all the degrees of freedom in $B$ and the full Hilbert space: $\Hil \cong \Hil_{A} \oslash \Hil_{B}$.

\item{Using the machinery of partial bipartitions, we can capture very general coarse-grainings of Hilbert space, since in most cases, the coarse-grained space which will preserve some relevant information will not correspond to a factor of Hilbert space. For example, it may be specified by a restricted set of observables which do not necessarily form an algebra. A particular interesting case which we consider in detail in this paper is to look for coarse-graining of a collection of $N$ underlying degrees of freedom (such as $N$ particles) based on a \emph{collective} or average feature of these degrees of freedom while tracing out the internal features. We focus on obtaining such a partial bipartition, $\Hil \cong \mathcal{S}_{\mrm{collective}} \oslash \mathcal{S}_{\mrm{internal}}$, where $\mathcal{S}_{\mrm{collective}}$ is the partial subsystem representing the coarse-graining which exhibits classical behavior under evolution by the Hamiltonian. This is a variational approach where we iterate over all possible bipartitions which define the split---that is, rearrangements of the elements inside the blocks of the bipartition table---and preferentially choose the one(s) which is(are) most compatible with the Hamiltonian and demonstrates quasi-classical features. Classicality is marked by the existence of macroscopic pointer states compatible with the Hamiltonian, superposition of which exhibit fast dynamical decoherence.}

\item{To define the coarse-graining $\mathcal{S}_{\mrm{collective}}$, we search for the collective observable $M_{c}$, of the form
\begin{equation}
M_{c} = \sum_{\mu = 1}^{N} M_{\mu} \: ,
\end{equation}
where each $M_{\mu}$ acts only on the $\mu$-th particle, most compatible or stationary with respect to the Hamiltonian, by minimizing the norm of $\lcb \Ham, M_{c} \rcb$ as in Eq.\ (\ref{Mc_compatibility}). Similar to the notion of predictability sieve\cite{Zurek:1994zq} in the decoherence literature, eigenstates of $M_{c}$ will define robust, pointer states of the system since they are most compatible with the Hamiltonian. Given the underlying $N$ degrees of freedom, the eigenstates of $M_{c}$ furnish a factorizable basis for Hilbert space, and eigenstates with distinct eigenvalues will label macroscopically distinct pointer states. These can be used to label and construct different columns of the bipartition table which specify the coarse-graining. Pointer states identified in this manner are special low-entropy states which stay robust to entanglement production under evolution. This is a telltale sign of a classical variable which does not arbitrarily entangle with all other degrees of freedom on short timescales. In this sense, eigenstates of the collective observable chosen by the compatibility condition of Eq.\ (\ref{Mc_compatibility}) are classical, macroscopic pointer states which capture an average, collective property of the underlying degrees of freedom which is as robust under evolution as possible.}

\item{Based on the transition structure of the Hamiltonian written in the factorized $M_c$ basis, we can split our Hilbert space into superselection sectors which never interact and hence  form disjoint blocks of our bipartition table. To fix the remaining freedom within each block of the bipartition table, we need to fix the alignment of the rows for which we return to the question of quasi-classicality. A defining feature of our coarse-graining should be that dynamics in the reduced space constructed from the state-reduction map defined by the bipartition table will reflect features of classicality. After identifying the column structure of the bipartition table based on compatibility of a collective observable $M_c$ with the Hamiltonian, we focus on effective dynamical decoherence by the Hamiltonian. Hence, we expect the row alignment of the bipartition table to be such that Hamiltonian evolution decoheres superpositions of macroscopic pointer states by ``interaction'' with $\mathcal{S}_{\mrm{internal}}$. We quantify the entanglement production of a pure state $\rho(t) = \ket{\psi(t)}\bra{\psi(t)} \in \mathcal{L}(\Hil)$ evolving under evolution by the Hamiltonian using linear entanglement entropy,
\begin{equation}
S_{\mrm{lin}}(t) = 1 - \Tr(\rho^{2}_{c}(t)) \: .
\end{equation}
where
\begin{equation}
\rho_c(t) \: \equiv \: \Tr_{(\mathcal{S}_{\mrm{internal}})}\rho(t) \: ,
\end{equation}
is the reduced state which $\rho(t)$ gets mapped to by the state-reduction map $\Tr_{(\mathcal{S}_{\mrm{internal}})}$. 
We iterate over all finite, discrete permutations of row alignments to select (the class of) bipartition table(s) which maximize entanglement production. This is done for a set of candidate classical states which are taken to be natural extensions of the unentangled, initial ready states in the decoherence literature.} 

\item{Using this algorithm to obtain the classical coarse-graining of an underlying $N$ degrees of freedom based on a collective feature compatible with the Hamiltonian, we analyze the Ising model in 1-D. We see the emergence of different coarse-grainings depending on the whether the nearest neighbor spin interaction or the external magnetic field dominates the Hamiltonian, a phenomenon akin to a phase transition. Depending on the preferentially selected collective compatible observable, either the total spin-z or total spin-x of the Ising chain, the coarse-graining may or may not exhibit superselection sectors. In both cases, the dimension of the coarse-grained space is $\sim \mathcal{O}(N)$ compared to the original Hilbert space, which has dimension $\sim \mathcal{O}(2^{N})$. The classical coarse-grainings  picked out exhibit fast dynamical decoherence between eigenstates of the compatible macroscopic variable and lead to emergent quasi-classicality. We exhibit numerical results for the case of $N = 3$ and $N=4$ spins, where the results are simple. Often a class of such quasi-classical bipartition tables (and hence, coarse-grainings) will get selected which reflects a symmetry between different underlying degrees of freedom from the point of view of the Hamiltonian. This setup can be generalized to other physical systems to study classical coarse-grainings determined by the Hamiltonian itself.}

\end{itemize}

\subsection{Previous Work}

Because of the general nature of our subject there is a vast body of interesting related work.
Here we will only briefly mention some of the previous work directly related to the core problem of state reduction based on observables.

As discussed above, one of our major results is an algorithm for directly computing the irrep decomposition of a Hilbert space with respect to a subalgebra $\mathcal A$. 
We mention two complementary approaches to the same problem.
First, a quite different numerical algorithm for a related matrix-algebra problem was previously given by Murota \emph{et al} \cite{Murota2010} (see \cite{wang2013numerical} for its adaptation in the physics literature). 
Their approach, presented in the context of semidefinite programming, proceeds quite differently.
A key step in their algorithm involves sampling for a random matrix in the algebra, which is inherently numeric and requires the ability to span the operator space of the algebra. Our approach does not require sampling from the algebra and it has no prerequisite of being able to span the algebra. 

Second, in a more physical context, Holbrook \emph{et al}  \cite{holbrook2003noiseless} have proposed an algorithm for computing the noise commutant of an error algebra associated with a noisy channel. Similarly to our approach, they also propose an inherently non-numeric algorithm that relies on minimal projections as the fundamental building blocks of the algebra. However, their algorithm also requires the ability to span the operator space of the algebra, a prerequisite that is not easy to satisfy without numerics. 

Beyond the specific algorithm, we are concerned with the general phenomenon wherein we can assign definite classical dynamics to a set of observables, along the lines of the decoherence program but without a bipartite Hilbert space.
Several previous sets of authors have similarly investigated under what circumstances sets of observables can act as classical systems in their own right. Zanardi \cite{Zanardi:2001zz} identified a notion of ``virtual quantum subsystems'' induced by the operational choice of a set of measurements, and with collaborators \cite{Zanardi:2004zz} argued that in the general case the partition of a Hilbert space into subsystems could be identified from the set of operationally accessible observables (for related recent approaches to the subsystem problem see \cite{Cotler:2017abq,mereology}). Kofler and Brukner \cite{2007PhRvL..99r0403K} considered the emergence of classical physics from the coarse-graining inherent in measurements which can be performed with only finite precision. Duarte \emph{et al} \cite{duarte2017emerging} have constructed a state-reduction map for a blurred and saturated detector and analyzed the reduced dynamics. 

Along the same lines, in a series of papers (e.g.\ \cite{2004SHPMP..35...73C,Castagnino:2008zz,10.1007/978-94-007-2404-4_15,2014BrJPh..44..138F}; see also \cite{2005PhRvA..72a2109S}) Castagnino, Lombardi, and collaborators have developed the self-induced decoherence (SID) program, which conceptualizes decoherence as a dynamical process which identifies the classical variables by inspection of the Hamiltonian, without the need to explicitly identify a set of environment degrees of freedom.
The variational approach we sketch in Section \ref{sec:variational} is similarly concerned with the dynamical selection of a preferred set of observables.

\subsection{Organization of the Paper}

Because this paper is aimed at a broad audience, and mostly uses the tools of fundamental quantum mechanics along with linear algebra and representation theory, we have attempted to keep it self-contained and pedagogical to the extent possible.
In Section \ref{sec:prelim} we accordingly review the technical and conceptual tools we will use in the remainder of the paper.
In particular we review the concept of generalized bipartitions and bipartition tables introduced by one of us in \cite{2018PhRvA..97e2130K}, as well as results from the mathematical literature on representations of matrix algebras.

The remainder of the paper is concerned with the application of these tools to physical situations.
We will mostly be concerned with an operational approach, in which we assume a lab-like setup in which a set of accessible observables has been specified, and investigate the decoherences of the resulting states.
In Section \ref{sec:decoherence} we set up this general operational problem and its relation to the decoherence program, which we review. 
In Section \ref{sec:Irrep-Decompositions-of} we then then present the general algorithm for passing from an operator algebra to a bipartition.
Given this mechanism for producing a reduced state containing the desired coarse-grained information, we can use the tools of the decoherence program to investigate the dynamics and classicality of the reduced states.
Having specified the general algorithm, we specialize in Section \ref{sec:Examples-of-State} to physically relevant examples.
In particular we focus on the common case where the experimentalist only has access to coarse-grained, collective observables, where the generalized bipartition table takes a particularly simple form and superselection sectors are induced by the operator algebra.

In Section \ref{PartialBPT} we return to the general problem of coarse-graining from observables and discuss the state-reduction maps which arise when the set of observables need not form an algebra.
In Section \ref{sec:variational} we use the tools of the previous section and ideas from the decoherence program to initiate a more abstract, variational approach in which the goal is to determine the ``most classical'' set of observables given only a Hilbert space with a specified Hamiltonian.
To build intuition for the general case, we focus in Section \ref{sec:ising} on the Ising Model, where numerical calculations are tractable.
In Section \ref{sec:discussion} we conclude by sketching some of the potential applications of our work for quantum information, holography, and quantum gravity.

\section{Preliminaries}\label{sec:prelim}

\subsection{Setup and Notation}

Unless stated otherwise, all Hilbert spaces will be complex and finite-dimensional and the notions of \emph{linear operator }and \emph{matrix}
will be used interchangeably. We will denote with $\mathcal{L}\left(\mathcal{H}\right)$
the space of linear operators on the Hilbert space $\mathcal{H}$.
Isometric Hilbert spaces will be identified by the relation $\mathcal{H}_{1}\cong\mathcal{H}_{2}$
associated with some isometry $V$ between the spaces (most of the
time, isometric Hilbert spaces will arise when we relabel or reinterpret
the basis elements).

An\emph{ orthogonal projection} $\Pi\in\mathcal{L}\left(\mathcal{H}\right)$
is defined by the property $\Pi=\Pi^{\dagger}=\Pi^{2}$. In the following
we will refer to such an operator simply as a \emph{projection}, implying an
\emph{orthogonal projection }as defined here. This should not be confused
with the notion of \emph{pairwise }orthogonal projections\emph{ }which
refers to a set of projections $\left\{ \Pi_{k}\right\} $ such that
$\Pi_{k}\Pi_{k'}=\delta_{kk'}\Pi_{k}$ (we will sometimes omit \emph{pairwise}
when referring to such sets). The \emph{eigenspace }of a projection
$\Pi$ is the subspace of $\mathcal{H}$ on which $\Pi$ acts as the
identity. Similarly, an \emph{eigenbasis} of $\Pi$ refers to a set of orthonormal
vectors that span the eigenspace of $\Pi$. The\emph{ rank }of a projection\emph{
}is also the dimension of its eigenspace; we will often use this relation
implicitly.

A\emph{ partial isometry }$S\in\mathcal{L}\left(\mathcal{H}\right)$
is defined by the properties $SS^{\dagger}=\Pi_{\mathrm{fin}}$ and $S^{\dagger}S=\Pi_{\mathrm{in}}$
where $\Pi_{\mathrm{in}}$ and $\Pi_{\mathrm{fin}}$ are projections. A partial isometry
$S$ acts as an isometry on the eigenspace of $\Pi_{\mathrm{in}}$, mapping
it to the eigenspace of $\Pi_{\mathrm{fin}}$ (both projections have the same
rank), and it annihilates vectors that are orthogonal to the eigenspace
of $\Pi_{\mathrm{in}}$ (the kernel of $S$ is the kernel of $\Pi_{\mathrm{in}}$).
Every projection $\Pi$ is also a partial isometry ($\Pi_{\mathrm{in}}=\Pi_{\mathrm{fin}}=\Pi$),
so we will say that $S$ is a \emph{proper }partial isometry
if it is a partial isometry but it is not a projection.

A \emph{graph $G:=\left\{ V,E\right\} $ }is defined by a set of vertices
$V:=\left\{ v_{i}\right\} $ and a set of edges $E:=\left\{ \left(v_{i},v_{j}\right)\right\} $.
A \emph{path} $p$ on the graph is an ordered set of vertices $p=\left(v_{i_{1}},v_{i_{2}},...\right)$
such that every consecutive pair is connected by an edge $\left(v_{i_{k}},v_{i_{k+1}}\right)\in E$.
The path $p$ is called \emph{simple }if every vertex appears at most
once in $p$. We will say that a pair of vertices $v_{1},v_{2}\in V$
is \emph{connected by a path} if there is a path $p$ such that $v_{1}$
is its first vertex and $v_{2}$ is its last. A \emph{connected component}
is a subset of vertices $C\subseteq V$ such that every pair $v_{1},v_{2}\in C$
is connected by a path, and every pair $v_{1}\in C$, $v_{2}\in V\backslash C$
is not connected by a path.

\subsection{Generalized Bipartitions and Bipartition Tables\label{sec:Prelim - Bipartition Tables}}

A \emph{bipartite system} is a system that consists of two distinct subsystems
$A$ and $B$. A {bipartition} of a system is an explicit specification of
these subsystems. When the system is bipartite by construction---the system of two qubits, for example---it comes with a natural bipartite
structure $\mathcal{H}\cong\mathcal{H}_{A}\otimes\mathcal{H}_{B}$. The
Hilbert space of the whole system is constructed from the tensor product
of two Hilbert spaces and the bases are naturally constructed from
products of local bases. Such a construction, however, is not necessary
and we can always impose a bipartition after the fact by selecting
a bipartite tensor product structure in any (non-prime dimensional)
Hilbert space. Different bipartitions of the Hilbert space identify
different subsystems that are not necessarily physical in the usual
sense but are associated with distinct degrees of freedom that define
a virtual subsystem \cite{Zanardi:2001zz}.

Formally, given a $d$-dimensional Hilbert space $\mathcal{H}$ such
that $d=d_{A}d_{B}$, we can introduce an auxiliary bipartite Hilbert
space $\mathcal{H}_{A}\otimes\mathcal{H}_{B}$ with dimensions $\Dim \mathcal{H}_{A} = d_{A}$ and $\Dim \Hil_{B} = d_{B}$.
By isometrically mapping the original Hilbert space $\mathcal{H}$
into $\mathcal{H}_{A}\otimes\mathcal{H}_{B}$ we impose a tensor product
structure that might not have been explictly present beforehand. Different choices of the isometry
$V:\text{\ensuremath{\mathcal{H}}}\longrightarrow\mathcal{H}_{A}\otimes\mathcal{H}_{B}$
specify different choices of bipartition, and the isometry $V$ itself
is fully described by some orthonormal basis $\ket{e_{ik}}$ in $\mathcal{H}$
where $i=1...d_{A}$ and $k=1...d_{B}$ such that $V\ket{e_{ik}}=\ket{a_{i}}\ket{b_{k}}$, where the elements $\ket{e_{ik}}$ and $\ket{a_{i}}\ket{b_{k}}$ are pairs
of right and left singular vectors of $V$. The choice of bipartition
is therefore conveniently summarized by choosing the
elements $\ket{e_{ik}}$ and arranging them into a rectangular table such that the $i,k$
indices correspond to the row and column of the element, respectively:

\noindent\begin{minipage}[c]{1\columnwidth}%
\begin{tabular}{>{\centering}m{0.9\columnwidth}>{\raggedright}m{0.1\columnwidth}}
 & \tabularnewline
\centering{}%
\begin{tabular}{c|c|c|c|c|c}
\multicolumn{1}{c}{} & \multicolumn{1}{c}{$b_{1}$} & \multicolumn{1}{c}{$b_{2}$} & \multicolumn{1}{c}{$\cdots$} & \multicolumn{1}{c}{$b_{d_{B}}$} & \tabularnewline
\cline{2-5} \cline{3-5} \cline{4-5} \cline{5-5} 
$a_{1}$ & $e_{11}$ & $e_{12}$ & $\cdots$ & $e_{1d_{B}}$ & \tabularnewline
\cline{2-5} \cline{3-5} \cline{4-5} \cline{5-5} 
$a_{2}$ & $e_{21}$ & $e_{22}$ & $\cdots$ & $e_{2d_{B}}$ & \tabularnewline
\cline{2-5} \cline{3-5} \cline{4-5} \cline{5-5} 
$\vdots$ & $\vdots$ & $\vdots$ & $\ddots$ & $\vdots$ & \tabularnewline
\cline{2-5} \cline{3-5} \cline{4-5} \cline{5-5} 
$a_{d_{a}}$ & $e_{d_{A}1}$ & $e_{d_{A}2}$ & $\cdots$ & $e_{d_{A}d_{B}}$ & \tabularnewline
\cline{2-5} \cline{3-5} \cline{4-5} \cline{5-5} 
\end{tabular} & %
\begin{minipage}[c]{0.04\columnwidth}%
\begin{center}
\begin{equation}
\end{equation}
\par\end{center}
~%
\end{minipage}\tabularnewline
 & \tabularnewline
\end{tabular}%
\end{minipage}

\noindent The rows of this table are associated with the degree of
freedom of subsystem $A$ and the columns are associated with the
degree of freedom of subsystem $B$. We will refer to such tables, which one of us first introduced in \cite{2018PhRvA..97e2130K}, as \emph{bipartition tables} (BPTs).
It should be clear that for each bipartition table there is another, trivially related one derived by swapping the row and column indices, which simply swaps the first and second systems in the bipartition. 

As a simple example, consider a system of two qubits and the product
basis $\{ \ket{00}$,$\ket{01}$,$\ket{10}$,$\ket{11} \}$. The BPT

\noindent\begin{minipage}[c]{1\columnwidth}%
\begin{tabular}{>{\centering}m{0.9\columnwidth}>{\raggedright}m{0.1\columnwidth}}
 & \tabularnewline
\centering{}%
\begin{tabular}{c|c|c|c}
\multicolumn{1}{c}{} & \multicolumn{1}{c}{$0_{B}$} & \multicolumn{1}{c}{$1_{B}$} & \tabularnewline
\cline{2-3} \cline{3-3} 
$0_{A}$ & $00$ & $01$ & \tabularnewline
\cline{2-3} \cline{3-3} 
$1_{A}$ & $10$ & $11$ & \tabularnewline
\cline{2-3} \cline{3-3} 
\end{tabular} & %
\begin{minipage}[c]{0.04\columnwidth}%
\begin{center}
\begin{equation}
\end{equation}
\par\end{center}
~%
\end{minipage}\tabularnewline
 & \tabularnewline
\end{tabular}%
\end{minipage}

\noindent represents the natural tensor product structure given by
construction, with each of the elements placed at the row and column
that corresponds to the values of the qubits. The subsystems $A$
and $B$ in this case are the qubits themselves.

A minor rearrangement of the two qubit BPT

\noindent\begin{minipage}[c]{1\columnwidth}%
\begin{tabular}{>{\centering}m{0.9\columnwidth}>{\raggedright}m{0.1\columnwidth}}
 & \tabularnewline
\centering{}%
\begin{tabular}{c|c|c|c}
\multicolumn{1}{c}{} & \multicolumn{1}{c}{$\mrm{even}_{B}$} & \multicolumn{1}{c}{$\mrm{odd}_{B}$} & \tabularnewline
\cline{2-3} \cline{3-3} 
$0_{A}$ & $00$ & $01$ & \tabularnewline
\cline{2-3} \cline{3-3} 
$1_{A}$ & $11$ & $10$ & \tabularnewline
\cline{2-3} \cline{3-3} 
\end{tabular} & %
\begin{minipage}[c]{0.04\columnwidth}%
\begin{center}
\begin{equation}
\end{equation}
\par\end{center}
~%
\end{minipage}\tabularnewline
 & \tabularnewline
\end{tabular}%
\end{minipage}

\noindent results in a new tensor product structure where we relabeled
the columns to better match their new meaning. Here the value of the
left qubit still varies with rows but what now varies with columns
is the overall parity of the two qubits, so subsystem $A$
is still interpreted as the left qubit but subsystem $B$ is now associated
with the parity degree of freedom. The isometry defined by this BPT
is
\begin{equation}
V=\ket{0_{A}}\ket{\mrm{even}_{B}}\bra{00}+\ket{0_{A}}\ket{\mrm{odd}_{B}}\bra{01}+\ket{1_{A}}\ket{\mrm{odd}_{B}}\bra{10}+\ket{1_{A}}\ket{\mrm{even}_{B}}\bra{11},
\end{equation}
so with respect to this bipartition, the entangled Bell state $\ket{00}+\ket{11}$
maps to $\ket{0_{A}}\ket{\mrm{even}_{B}}+\ket{1_{A}}\ket{\mrm{even}_{B}}$, which
is not entangled. From now on we will not explicitly label the rows
and columns on BPTs, but we will implicitly use the fact that the rows
and columns represent the individual degrees of freedom of the two
subsystems in the bipartition.

The visual representation of BPTs can also be extended to capture
direct-sum decompositions of Hilbert spaces. By arranging basis
elements into a block-diagonal table,

\noindent %
\noindent\begin{minipage}[c]{1\columnwidth}%
\begin{tabular}{>{\centering}m{0.9\columnwidth}>{\raggedright}m{0.1\columnwidth}}
 & \tabularnewline
\centering{}%
\begin{tabular}{|c|c|c|ccc}
\cline{1-3} \cline{2-3} \cline{3-3} 
$e_{11}^{1}$ & $e_{12}^{1}$ & $\cdots$ &  &  & \tabularnewline
\cline{1-3} \cline{2-3} \cline{3-3} 
$e_{21}^{1}$ & $e_{22}^{1}$ & $\cdots$ &  &  & \tabularnewline
\cline{1-3} \cline{2-3} \cline{3-3} 
$\vdots$ & $\vdots$ & $\ddots$ &  &  & \tabularnewline
\cline{1-5} \cline{2-5} \cline{3-5} \cline{4-5} \cline{5-5} 
\multicolumn{1}{c}{} & \multicolumn{1}{c}{} &  & \multicolumn{1}{c|}{$e_{11}^{2}$} & \multicolumn{1}{c|}{$\cdots$} & \tabularnewline
\cline{4-5} \cline{5-5} 
\multicolumn{1}{c}{} & \multicolumn{1}{c}{} &  & \multicolumn{1}{c|}{$\vdots$} & \multicolumn{1}{c|}{$\ddots$} & \tabularnewline
\cline{4-5} \cline{5-5} 
\multicolumn{1}{c}{} & \multicolumn{1}{c}{} & \multicolumn{1}{c}{} &  &  & $\ddots$\tabularnewline
\end{tabular} & %
\begin{minipage}[c]{0.04\columnwidth}%
\begin{center}
\begin{equation}
\label{eq:generic block diag BPT}
\end{equation}
\par\end{center}
~%
\end{minipage}\tabularnewline
 & \tabularnewline
\end{tabular}%
\end{minipage}

\noindent we can specify Hilbert-space decompositions of the form
\begin{equation}
\mathcal{H}\cong\bigoplus_{q}\mathcal{H}_{A_{q}}\otimes\mathcal{H}_{B_{q}},\label{eq:def generalized bipartition}
\end{equation}
where the sector $q$ is spanned by the basis elements $\ket{e_{ik}^{q}}$
of the block $q$ and each sector is further decomposed into a tensor product
of two subsystems according to the arrangement of elements inside
the block. We will refer to decompositions of the form (\ref{eq:def generalized bipartition})
as \emph{generalized bipartitions, }and by BPT we will imply the generalized
form (\ref{eq:generic block diag BPT}). In Sections \ref{PartialBPT}-\ref{sec:ising} we will further generalize this idea to non-rectangular BPTs that capture the notion of \textit{partial bipartitions,} associated with decompositions that cannot be expressed as in Eq.\ (\ref{eq:def generalized bipartition}).

\noindent 

As an example, consider the 3 spin-$\frac{1}{2}$ system decomposed
into total spin sectors:
\begin{equation}
\mathcal{H}=\frac{1}{2}\otimes\frac{1}{2}\otimes\frac{1}{2}\cong\frac{3}{2}\oplus\frac{1}{2}\oplus\frac{1}{2}.
\end{equation}
The bases that correspond to each total spin sector are $\ket{\frac{3}{2},m}$,
$\ket{\frac{1}{2},m,1}$, $\ket{\frac{1}{2},m,2}$ where $m$ varies
from $\frac{3}{2}$ to $-\frac{3}{2}$ in integer steps and $1,2$ label
the two distinct sectors of total spin $\frac{1}{2}$.The BPT

\noindent\begin{minipage}[c]{1\columnwidth}%
\begin{tabular}{>{\centering}m{0.9\columnwidth}>{\raggedright}m{0.1\columnwidth}}
 & \tabularnewline
\centering{}%
\begin{tabular}{cccc|cc|cc}
\cline{1-4} \cline{2-4} \cline{3-4} \cline{4-4} 
\multicolumn{1}{|c|}{$\frac{3}{2},\frac{3}{2}$} & \multicolumn{1}{c|}{$\frac{3}{2},\frac{1}{2}$} & \multicolumn{1}{c|}{$\frac{3}{2},-\frac{1}{2}$} & $\frac{3}{2},-\frac{3}{2}$ &  & \multicolumn{1}{c}{} &  & %
\begin{tabular}{c}
\tabularnewline
\tabularnewline
\end{tabular}\tabularnewline
\cline{1-6} \cline{2-6} \cline{3-6} \cline{4-6} \cline{5-6} \cline{6-6} 
 &  &  &  & \multicolumn{1}{c|}{$\frac{1}{2},\frac{1}{2},1$} & $\frac{1}{2},-\frac{1}{2},1$ &  & %
\begin{tabular}{c}
\tabularnewline
\tabularnewline
\end{tabular}\tabularnewline
\cline{5-8} \cline{6-8} \cline{7-8} \cline{8-8} 
\begin{tabular}{c}
\tabularnewline
\tabularnewline
\end{tabular} &  &  & \multicolumn{1}{c}{} &  &  & \multicolumn{1}{c|}{$\frac{1}{2},\frac{1}{2},2$} & \multicolumn{1}{c|}{$\frac{1}{2},-\frac{1}{2},2$}\tabularnewline
\cline{7-8} \cline{8-8} 
\end{tabular} & %
\begin{minipage}[c]{0.04\columnwidth}%
\begin{center}
\begin{equation}
\end{equation}
\par\end{center}
~%
\end{minipage}\tabularnewline
 & \tabularnewline
\end{tabular}%
\end{minipage}

\noindent represents the direct-sum decomposition of the Hilbert space
into total spin sectors. By stacking the two rows of total spin $\frac{1}{2}$
into a single block,

\noindent %
\noindent\begin{minipage}[c]{1\columnwidth}%
\begin{tabular}{>{\centering}m{0.9\columnwidth}>{\raggedright}m{0.1\columnwidth}}
 & \tabularnewline
\centering{}%
\begin{tabular}{cccc|c|c|}
\cline{1-4} \cline{2-4} \cline{3-4} \cline{4-4} 
\multicolumn{1}{|c|}{$\frac{3}{2},\frac{3}{2}$} & \multicolumn{1}{c|}{$\frac{3}{2},\frac{1}{2}$} & \multicolumn{1}{c|}{$\frac{3}{2},-\frac{1}{2}$} & $\frac{3}{2},-\frac{3}{2}$ & \multicolumn{1}{c}{%
\begin{tabular}{c}
\tabularnewline
\tabularnewline
\end{tabular}} & \multicolumn{1}{c}{}\tabularnewline
\hline 
\begin{tabular}{c}
\tabularnewline
\tabularnewline
\end{tabular} &  &  &  & $\frac{1}{2},\frac{1}{2},1$ & $\frac{1}{2},-\frac{1}{2},1$\tabularnewline
\cline{5-6} \cline{6-6} 
\begin{tabular}{c}
\tabularnewline
\tabularnewline
\end{tabular} &  &  &  & $\frac{1}{2},\frac{1}{2},2$ & $\frac{1}{2},-\frac{1}{2},2$\tabularnewline
\cline{5-6} \cline{6-6} 
\end{tabular} & %
\begin{minipage}[c]{0.04\columnwidth}%
\begin{center}
\begin{equation}
\end{equation}
\par\end{center}
~%
\end{minipage}\tabularnewline
 & \tabularnewline
\end{tabular}%
\end{minipage}

\noindent we specify a different, more subtle decomposition of the Hilbert space.
We now have two sectors, one associated with total spin $\frac{3}{2}$
and the other with total spin $\frac{1}{2}$ where the $\frac{1}{2}$ sector is further decomposed
into a tensor product
\begin{equation}
\mathcal{H}\cong\frac{3}{2}\oplus\left(\mathcal{N}_{\frac{1}{2}}\otimes\frac{1}{2}\right).\label{eq:3 spin H space decomp}
\end{equation}
The virtual subsystem $\mathcal{N}_{\frac{1}{2}}$ is usually referred
to as the \emph{multiplicity} subsystem while $\frac{1}{2}$ still
represents the total spin-$\frac{1}{2}$ magnetization degree of freedom. The multiplicity subsystem $\mathcal{N}_{\frac{1}{2}}$ is also well known as the prototypical example of a \textit{noiseless subsystem} \cite{lidar2014review}, which encodes information in the relational degrees of freedom that are invariant under collective rotations. 
In general, such bipartitions naturally arise from the structure of irreducible representations of symmetry groups, as we will see below.

\subsection{Matrix Algebras and Their Representation\label{subsec:Matrix-Algebras-and their representation}}

We will now summarize the relevant results of the representation theory
of finite-dimensional operator algebras and relate them to the BPT
picture of the previous subsection. Our exposition will emphasize the
structural details of the representation theory at the expense of
mathematical rigor. The mathematically inclined reader is referred to
\cite{beny2015algebraic} or \cite{farenick2012algebras}.

Let us first define what we mean by a matrix algebra\footnote{In the literature, matrix algebras are often referred to as von Neumann
algebras or $C^{*}$-algebras, even when only finite-dimensional spaces
are involved. We prefer the term ``matrix algebra'' to emphasize
the fact that we are dealing with a simpler, finite-dimensional, case
where we need not be concerned with the subtleties of infinite-dimensional
spaces.}
\begin{defn}
A \emph{matrix algebra }is a subset $\mathcal{A}\subseteq\mathcal{L}\left(\mathcal{H}\right)$
such that for any $M_{1},M_{2}\in\mathcal{A}$ and $c\in\mathbb{C}$:

(1) $M_{1}+M_{2}\in\mathcal{A}$

(2) $M_{1}M_{2}\in\mathcal{A}$

(3) $cM_{1}\in\mathcal{A}$

(4) $M_{1}^{\dagger}\in\mathcal{A}$
\end{defn}

\noindent For example, the set $\mathcal{L}\left(\mathcal{H}\right)$
is a \emph{full }matrix algebra on $\mathcal{H}$. From here on we
will use the term \emph{algebra }to mean \emph{matrix algebra }as
defined above.

Any finite (or infinite) set of matrices $\mathcal{M}:=\left\{ M_{1},M_{2},...M_{n}\right\} $
can generate the algebra $\mathcal{A}:=\left\langle M_{1},M_{2},...M_{n}\right\rangle $
(which the angled brackets denote) by taking the closure of $\mathcal{M}$
with respect to operations in the above definition. It should be clear
then that the algebra $\left\langle M_{1},M_{2},...M_{n}\right\rangle $
is spanned by linear combinations of products of elements $\left\{ M_{1},M_{2},...M_{n}\right\} \cup\left\{ M_{1}^{\dagger},M_{2}^{\dagger},...M_{n}^{\dagger}\right\} $.

The central result of representation theory of matrix algebras is
known as Wedderburn decomposition \cite{wedderburn1934lectures}, and it can be stated
in the following way:
\begin{thm}
\label{thm:Wedderburn decomposition} (Wedderburn Decomposition) For
every algebra $\mathcal{A}\subseteq\mathcal{L}\left(\mathcal{H}\right)$,
the Hilbert space $\mathcal{H}$ decomposes into 
\begin{equation}
\mathcal{H}\cong\left[\bigoplus_{q}\mathcal{H}_{A_{q}}\otimes\mathcal{H}_{B_{q}}\right]\oplus\mathcal{H}_{0}\label{eq:Wedder. decomposition H space}
\end{equation}
 such that every element $M\in\mathcal{A}$ is of the form
\begin{equation}
M=\left[\bigoplus_{q}I_{A_{q}}\otimes M_{B_{q}}\right]\oplus0,\label{eq:Wedder decomposition matrix}
\end{equation}
where $I_{A_{q}}$ is the identity on $\mathcal{H}_{A_{q}}$ and $M_{B_{q}}$
is any matrix on $\mathcal{H}_{B_{q}}$, and all matrices of this
form are elements of $\mathcal{A}$.

~

\noindent (for a contemporary exposition of the proof see Section
2.7 in \cite{beny2015algebraic} or Appendix A of \cite{Harlow:2016vwg})
\end{thm}

In the language of representation theory, Eq.\ (\ref{eq:Wedder. decomposition H space})
is the decomposition of $\mathcal{H}$ into irreducible representations
(irreps) of the algebra $\mathcal{A}$. The tensor factors $\mathcal{H}_{B_{q}}$
in the bipartition are associated with distinct irreps of $\mathcal{A}$
while the tensor factors $\mathcal{H}_{A_{q}}$ are associated with
the multiplicity of distinct irreps. It is important to note the significance
of the fact that not only all $M\in\mathcal{A}$ are of the form (\ref{eq:Wedder decomposition matrix}),
but that \emph{any} matrix of this form is necessarily an element
of $\mathcal{A}$. Therefore, the decomposition (\ref{eq:Wedder. decomposition H space})
is the defining structure of an algebra that selects the elements
of the algebra to be \emph{all} the matrices that act nontrivially
only on the tensor factors $\mathcal{H}_{B_{q}}$ in the decomposition.
The null space $\mathcal{H}_{0}$ is the space where the algebra is
not supported and its elements act on $\mathcal{H}_{0}$ as the null
matrix. From now on we will ignore the null space in the decomposition
and assume the Hilbert space $\mathcal{H}$ to exclude $\mathcal{H}_{0}$.
\footnote{In the cases that we will consider, $\mathcal{H}_{0}$ does not appear
in the decomposition. Even when $\mathcal{H}_{0}$ does appear, it
simply means that that part of the Hilbert space is irrelevant for
operators of the algebra.}

As was discussed in Sec.\ \ref{sec:Prelim - Bipartition Tables}, decompositions
such as (\ref{eq:Wedder. decomposition H space}) are generalized
bipartitions that correspond to a BPT of the form (\ref{eq:generic block diag BPT}).
This correspondence and the result of Theorem \ref{thm:Wedderburn decomposition}
suggest that the defining structure of an algebra is explicitly captured
by a BPT. We can therefore explicitly specify algebras with BPTs and
vice versa via this correspondence.

In order to see what the BPT tells us about the structure of an algebra
we consider the basis $\{\ket{e_{ik}^{q}}\}$ that corresponds to the decomposition
(\ref{eq:Wedder. decomposition H space}) in the sense that for every
sector $q$ there are product bases $\{\ket{a_{i}^{q}}\ket{b_{k}^{q}}\}$
of $\mathcal{H}_{A_{q}}\otimes\mathcal{H}_{B_{q}}$ such that $\ket{e_{ik}^{q}}=\ket{a_{i}^{q}}\ket{b_{k}^{q}}$
(note that this definition is not unique and any choice of local basis
$\ket{a_{i}^{q}}$ and $\ket{b_{k}^{q}}$ can work). According to
Eq.\ (\ref{eq:Wedder decomposition matrix}), all matrices in the algebra
can be constructed from linear combinations of the operators 
\begin{equation}
S_{kl}^{q}:=I_{A_{q}}\otimes\ket{b_{k}^{q}}\bra{b_{l}^{q}}=\sum_{i}\ket{e_{ik}^{q}}\bra{e_{il}^{q}}.\label{eq:BPO def}
\end{equation}
These operators, which we will call \emph{bipartition operators} (BPOs),
are partial isometries, and they form an (unnormalized) operator basis
for the algebra.

Now consider the BPT constructed with the basis $\ket{e_{ik}^{q}}$:

\noindent\begin{minipage}[c]{1\columnwidth}%
\begin{tabular}{>{\centering}m{0.9\columnwidth}>{\raggedright}m{0.1\columnwidth}}
 & \tabularnewline
\centering{}%
\begin{tabular}{|c|c|c|ccc}
\cline{1-3} \cline{2-3} \cline{3-3} 
$e_{11}^{1}$ & $e_{12}^{1}$ & $\cdots$ &  &  & \tabularnewline
\cline{1-3} \cline{2-3} \cline{3-3} 
$e_{21}^{1}$ & $e_{22}^{1}$ & $\cdots$ &  &  & \tabularnewline
\cline{1-3} \cline{2-3} \cline{3-3} 
$\vdots$ & $\vdots$ & $\ddots$ &  &  & \tabularnewline
\cline{1-5} \cline{2-5} \cline{3-5} \cline{4-5} \cline{5-5} 
\multicolumn{1}{c}{} & \multicolumn{1}{c}{} &  & \multicolumn{1}{c|}{$e_{11}^{2}$} & \multicolumn{1}{c|}{$\cdots$} & \tabularnewline
\cline{4-5} \cline{5-5} 
\multicolumn{1}{c}{} & \multicolumn{1}{c}{} &  & \multicolumn{1}{c|}{$\vdots$} & \multicolumn{1}{c|}{$\ddots$} & \tabularnewline
\cline{4-5} \cline{5-5} 
\multicolumn{1}{c}{} & \multicolumn{1}{c}{} & \multicolumn{1}{c}{} &  &  & $\ddots$\tabularnewline
\end{tabular} & %
\begin{minipage}[c]{0.04\columnwidth}%
\begin{center}
\begin{equation}
\end{equation}
\par\end{center}
~%
\end{minipage}\tabularnewline
 & \tabularnewline
\end{tabular}%
\end{minipage}

\noindent and the subspaces selected by the basis elements of the distinct
rows and columns. The BPO $S_{kl}^{q}$ acts by mapping the basis element
in column $l$ of block $q$ to the parallel element in column $k$
of the same block; this is a partial isometry between subspaces of
the columns. Since the basis elements inside each row are mapped to
themselves by the BPOs, and since the BPOs span the algebra, distinct rows of
the BPT define invariant subspaces of the algebra. The row subspaces
are \emph{minimal} invariant subspaces (they do not contain smaller
invariant subspaces) because BPOs act on these subspaces as the full
matrix algebra which is irreducible \cite{farenick2012algebras}.

Column subspaces are also a meaningful part of the matrix algebra
structure. The projection operator on the subspace of column $k$
in block $q$ is just a special case of a BPO (projections are the
trivial partial isometries from subspaces to themselves):
\begin{equation}
S_{kk}^{q}=\sum_{i}\ket{e_{ik}^{q}}\bra{e_{ik}^{q}}.
\end{equation}
The adjoint action of the projection $S_{kk}^{q}$ on any other BPO
results in 
\begin{equation}
S_{kk}^{q}S_{k'l'}^{q'}S_{kk}^{q}=\delta_{qq'}\delta_{kk'}\delta_{kl'}S_{kk}^{q}.
\end{equation}
 Since every element of the algebra is a linear combination of BPOs,
the adjoint action of $S_{kk}^{q}$ on any $M\in\mathcal{A}$ must
result in 
\begin{equation}
S_{kk}^{q}MS_{kk}^{q}\propto S_{kk}^{q}.\label{eq:min proj def}
\end{equation}
Projections in the algebra for which Eq.\ (\ref{eq:min proj def})
holds for all elements $M\in\mathcal{A}$ are the key building blocks
of the algebra:
\begin{defn}
	\label{def:minimal projection}A projection $\Pi\in\mathcal{A}$ is
	called a \emph{minimal projection} if for every $M\in\mathcal{A}$ we
	have $\Pi M\Pi\propto\Pi$. \footnote{This property is equivalent to a different, more common, defining
		property: $\Pi_{\mathrm{min}}$ is minimal if for all projections $\Pi\in\mathcal{A}$
		such that $\Pi \Pi_{\mathrm{min}}=\Pi$ it implies that either $\Pi=0$ or $\Pi=\Pi_{\mathrm{min}}$.
		We prefer to define it the other way because this is the only property
		of minimal projections that we will use.}
\end{defn}

\noindent Not only are all $S_{kk}^{q}$'s minimal projections, they
are also the maximal set of such projections.
\begin{defn}
\label{def:MSMP}A set of projections $\left\{ \Pi_{k}\right\} \subseteq\mathcal{A}$
is called a \emph{maximal set of} \emph{minimal projections (MSMP)
}if every $\Pi_{k}$ is minimal and all $\Pi_{k}$ are pairwise orthogonal
and sum to the identity element $I_{\mathcal{A}}:=\sum_{k}\Pi_{k}$
of the algebra.
\end{defn}

\noindent The columns of a BPT are therefore a concise summary of
a particular choice of MSMP given by the BPOs $\left\{ S_{kk}^{q}\right\} $
(the non-uniqueness of this choice traces back to the freedom to choose
the local basis $\ket{b_{k}^{q}}$ ).

The commutant $\mathcal{A}'$ of an algebra $\mathcal{A}$ is the
set of all matrices that commute with every element of $\mathcal{A}$
\begin{equation}
\mathcal{A}':=\left\{ M'\in\mathcal{L}\left(\mathcal{H}\right)|\left[M',M\right]=0,\,\forall M\in\mathcal{A}\right\} ,
\end{equation}
and is itself also an algebra. The irrep decomposition for $\mathcal{A}'$
is essentially the same as for $\mathcal{A}$ with the roles of the
tensor factors $\mathcal{H}_{A_{q}}$ and $\mathcal{H}_{B_{q}}$ reversed.
That is, if
\begin{equation}
\mathcal{H}\cong\bigoplus_{q}\mathcal{H}_{A_{q}}\otimes\mathcal{H}_{B_{q}}
\end{equation}
is the irrep decomposition for $\mathcal{A}$, then all $M'\in\mathcal{A}'$
are of the form
\begin{equation}
M'=\bigoplus_{q}M'_{A_{q}}\otimes I_{B_{q}}.\label{eq:Wedder decomposition commutant matrix}
\end{equation}
For the BPTs this implies a reversal of roles between rows and columns.
Given the BPT of $\mathcal{A}$ we can get the BPT of $\mathcal{A}'$
by rotating rows into columns; we will call this transformation a
\emph{transpose}. Consequently, BPOs constructed from a transposed
BPT span the commutant of the algebra.

A simple example of an algebra is the full matrix algebra $\mathcal{L}\left(\mathcal{H}\right)$.
The BPT of this algebra is just a single row of all basis elements
$\ket{e_{k}}$ (the choice of basis is arbitrary)

\noindent\begin{minipage}[c]{1\columnwidth}%
\begin{tabular}{>{\centering}m{0.9\columnwidth}>{\raggedright}m{0.1\columnwidth}}
\centering{}%
\begin{tabular}{|c|c|c|c|}
\hline 
$e_{1}$ & $e_{2}$ & $\cdots$ & $e_{d}$\tabularnewline
\hline 
\end{tabular} & %
\begin{minipage}[c]{0.04\columnwidth}%
\begin{center}
\begin{equation}
\end{equation}
\par\end{center}
~%
\end{minipage}\tabularnewline
\end{tabular}%
\end{minipage}

\noindent The BPOs defined by this table are just the matrix units
\begin{equation}
S_{kl}=\ket{e_{k}}\bra{e_{l}}
\end{equation}
that span all the matrices in the algebra. The transpose of this BPT
results in a single column that corresponds to a single BPO that is
the identity matrix $I$. This means that the commutant of the full
matrix algebra $\mathcal{L}\left(\mathcal{H}\right)$ consists of
the span of $I$, as expected.

Another important example of an algebra is the algebra $\left\langle M\right\rangle $
generated by a single self-adjoint matrix $M$. By definition, $\left\langle M\right\rangle $
is the set of all matrices spanned by $M^{n}$ for all natural $n$.
The key fact about this algebra is that it contains, and therefore
can be spanned by, the spectral projections of $M$:
\begin{prop}
\label{prop:MA of single operator }Let $M$ be a self-adjoint matrix
with the spectral decomposition 
\begin{equation}
M=\sum_{k}\lambda_{k}\Pi_{k}
\end{equation}
where $\lambda_{k}$ are distinct (non-zero) eigenvalues and $\Pi_{k}$
are projections on eigenspaces. Then 
\begin{equation}
\left\langle M\right\rangle =\mathsf{span}\left\{ \Pi_{k}\right\} .
\end{equation}
\end{prop}

This fact can be shown by first identifying the identity element $I_{\left\langle M\right\rangle }$
in this algebra (it does not have to be the full identity matrix).
The identity element is constructed using the minimal polynomial $p\left(x\right)$
of $M$ (that is, the smallest degree polynomial for which $p\left(M\right)=0$)
and the fact that for self-adjoint matrices the minimal polynomial
is of the form $p\left(x\right)=f\left(x\right)$ or $p\left(x\right)=xf\left(x\right)$
where $f$ is such that $f\left(0\right)\neq0$. Then
\begin{equation}
I_{\left\langle M\right\rangle }:=\frac{f\left(M\right)-I f\left(0\right)}{-f\left(0\right)}\in\left\langle M\right\rangle 
\end{equation}
acts as the identity on $M$, and uniqueness of the identity implies
that 
\begin{equation}
I_{\left\langle M\right\rangle }=\sum_{k}\Pi_{k}.
\end{equation}
With the identity, we can re-express the spectral projections as 
\begin{equation}
\Pi_{k}=\prod_{l\neq k}\frac{M-\lambda_{l}I_{\left\langle M\right\rangle }}{\lambda_{k}-\lambda_{l}}\in\left\langle M\right\rangle .
\end{equation}
Since every natural power of $M$ is in the span of spectral projections,
$\Pi_{k}$'s span the whole algebra $\left\langle M\right\rangle $.

The projections $\Pi_{k}$ are in fact the MSMP of $\left\langle M\right\rangle$,
since for all powers $n$ we have 
\begin{equation}
\Pi_{k}M^{n}\Pi_{k}=\left(\lambda_{k}\right)^{n}\Pi_{k}
\end{equation}
and clearly they are pairwise orthogonal and sum to the identity.
Since the MSMP $\left\{ \Pi_{k}\right\} $ spans $\left\langle M\right\rangle$,
these are the only BPOs in this algebra. From a complete set of BPOs
it is easy to build a BPT. In general we have seen that each minimal projection
defines a column and columns in the same block are related to each
other by a proper partial isometry. In this case there are no proper
partial isometries so each column is its own block.

\noindent\begin{minipage}[c]{1\columnwidth}%
\begin{tabular}{>{\centering}m{0.9\columnwidth}>{\raggedright}m{0.1\columnwidth}}
 & \tabularnewline
\centering{}%
\begin{tabular}{c|ccc}
\cline{1-1} 
\multicolumn{1}{|c|}{} &  &  & \tabularnewline
\multicolumn{1}{|c|}{$\Pi_{1}$} &  &  & \tabularnewline
\multicolumn{1}{|c|}{} &  &  & \tabularnewline
\cline{1-2} \cline{2-2} 
 & \multicolumn{1}{c|}{} &  & \tabularnewline
 & \multicolumn{1}{c|}{$\Pi_{2}$} &  & \tabularnewline
 & \multicolumn{1}{c|}{} &  & \tabularnewline
\cline{2-2} 
\multicolumn{1}{c}{} &  & $\ddots$ & \tabularnewline
\end{tabular} & %
\begin{minipage}[c]{0.04\columnwidth}%
\begin{center}
\begin{equation}
\end{equation}
\par\end{center}
~%
\end{minipage}\tabularnewline
 & \tabularnewline
\end{tabular}%
\end{minipage}

\noindent The height of each column is the rank of the projection
and the arrangement of basis elements inside the columns is not important in
this case. The irrep decomposition implied by this BPT decomposes
the Hilbert space into sectors of distinct eigenspaces of $M$ 
\begin{equation}
\mathcal{H}\cong\bigoplus_{k}\mathcal{H}_{A_{k}}\otimes h_{B_{k}}\label{eq:Wedderburn decomposition for <M>}
\end{equation}
where the tensor factors $B$ (associated with the columns in each
block) are one dimensional and the tensor factors $A$ (associated
with the rows in each block) are of dimension equal to the rank of
$\Pi_{k}$. Eq.\ (\ref{eq:Wedder decomposition matrix}) is then the
statement that all elements of $\left\langle M\right\rangle $ are
given by the span of $\Pi_{k}$. Under transpose, each block of the
BPT becomes a row specifying the full matrix algebra on that eigenspace
of $M$. The commutant is then the direct sum of full matrix algebras
on the eigenspaces of $M$, which is also what Eq.\ (\ref{eq:Wedder decomposition commutant matrix})
implies.

As we have seen, the structure of the algebra generated by a single
self-adjoint matrix $M$ is fully characterized by the spectral decomposition
of $M$. Our derivation of the irrep decomposition by constructing a BPT
from BPOs ended up being a roundabout way of decomposing the Hilbert
space into eigenspaces of $M$. We will see in Section \ref{sec:Irrep-Decompositions-of}
that this approach generalizes to algebras generated by multiple elements
$\left\langle M_{1},M_{2}...\right\rangle $. In that case, spectral
projections of generators are not sufficient to characterize the structure
of the algebra, but they can be used to produce a complete set of BPOs
that will specify a BPT and so the irrep decomposition.

The last special case of an algebra that is very useful is the \emph{group
algebra}. A group algebra is an algebra generated by matrices that form
a group. The same matrices that generate the group generate the group
algebra, however, the term ``generate'' in the context of matrix
algebras means that we also include linear combinations of the group
elements. That is, if $\mathcal{G}$ is a (finite or Lie) group generated
by $L_{1},L_{2}...$ then the group algebra $\mathbb{C}\mathcal{G}$
is the span of elements of $\mathcal{G}$
\begin{equation}
\mathbb{C}\mathcal{G}:=\left\langle L_{1},L_{2}...\right\rangle =\mathsf{span}\left\{ \mathcal{G}\right\} .
\end{equation}
An important fact about group algebras is that their irrep decomposition
is the same as the irrep decomposition for the group.
\begin{prop}
\label{prop:group algebra irreps}Let $\mathcal{G}$ be a finite or
Lie unitary group generated by $L_{1},L_{2}...$ acting on the Hilbert
space $\mathcal{H}$. If
\begin{equation}
\mathcal{H}\cong\bigoplus_{q}\mathcal{H}_{A_{q}}\otimes\mathcal{H}_{B_{q}}\label{eq:H irrep dec for group}
\end{equation}
is the irrep decomposition of $\mathcal{H}$ such that all elements
$U\left(g\right)\in\mathcal{G}$ are of the form
\begin{equation}
U\left(g\right)=\bigoplus_{q}I_{A_{q}}\otimes U_{q}\left(g\right)
\end{equation}
where $U_{q}\left(g\right)$ are irreducible, then (\ref{eq:H irrep dec for group})
is the irrep decomposition for the group algebra
\begin{equation}
\mathbb{C}\mathcal{G}=\left\langle L_{1},L_{2}...\right\rangle .
\end{equation}
\end{prop}

This fact follows from the observation that if a subspace is invariant
under the action of the group then it is invariant under the action
of the group algebra, since linear combinations of group elements
preserve the same subspaces as the elements themselves. The same reasoning
establishes that invariant subspaces that are equivalent representations
for group elements are also equivalent for linear combinations of
group elements. This leads to the conclusion that groups and their
algebras have the same minimal invariant subspaces with the same equivalences,
which means that they have the same irrep structure, hence the same
irrep decomposition.

Proposition \ref{prop:group algebra irreps} will allow us to construct
the irrep decomposition for group algebras using the known irreps
of groups. For example, going back to the 3 spin-$\frac{1}{2}$ case,
Eq.\ (\ref{eq:3 spin H space decomp}) is the irrep decomposition associated
with the $SU\left(2\right)$ group of collective rotations on the
spins. It is constructed by recognizing the total spin basis (via
the Clebsch-Gordan coefficients) that identify the minimal invariant
subspaces of total spin $\frac{3}{2}$ and $\frac{1}{2}$ that decompose
the Hilbert space into irreps of $SU\left(2\right)$. Since the group
of total rotations is generated by the total spin operators $J_{x}$,
$J_{y}$ ,$J_{z}$, we can conclude that the irrep decomposition of
the algebra generated by $J_{x},J_{y},J_{z}$ is given by the irrep
decomposition (\ref{eq:3 spin H space decomp}).

\section{Operational Approach to Decoherence}\label{sec:decoherence}

\subsection{The Decoherence Program}

In this subsection we review some basic aspects of the decoherence program, which we will apply below to the reduced states produced by our generalized state-reduction maps. 
The decoherence program is a well-established field with an extensive literature and our treatment here will be terse. 
The reader already familiar with its details is invited to proceed to the next subsection.
Conversely, more details can be found, for example, in the review \cite{Schlosshauer:2003zy} or the textbook \cite{2007dqct.book.....S}.
Several formulations of the decoherence program exist; here we discuss only the ``Zurekian'' framework.

The (Zurekian) decoherence program is a formalism for describing the circumstances under which a system can be classically measured.
Recall that the Born rule states that the possible results of measuring an observable ${\mathcal{O}}$ in a state $\ket{\psi}$ of a system represented by the Hilbert space $\Hil_S$ are the eigenstates ${\ket{o_i}}$ of the observables, with a probability $|\braket{o_i}{\psi}|^2$ of obtaining each individual outcome.
From the point of view of the system alone, the (projective) 
measurement process is non-unitary; for example, if the ${\ket{o_i}}$ are not eigenstates of the system's Hamiltonian so that ${\braket{o_i(t)}{o_j(t)}\ne \delta_{ij}}$, time evolution will act differently on the initial state and the post-measurement state.
In particular, interference terms will be suppressed in the post-measurement state, which no longer evolves coherently.

The Zurekian decoherence program implements this ``de-coherence'' process, which is effectively non-unitary for the system alone, as a unitary process on a larger Hilbert space consisting of the tensor product\footnote{Although some careful treatments require a tripartite system-apparatus-environment split (e.g.\ \cite{Elby:1994,Schlosshauer:2003zy,Boddy:2014eba}), here we will only split out the system from the environment; when such distinctions are important we have in mind that the system is small and quantum so that a (large, classical) apparatus is a subsystem of the environment.} $\Hil \simeq \Hil_S \otimes \Hil_E$ of the original system and an environment $\Hil_E$. 
If the Hamiltonian contains interaction terms between the system and environment degrees of freedom, then an initial product state can evolve into an entangled state of the system and environment:
\be
\ket{\Psi\(t=0\)}=\ket{\psi}_S \ket{e_0}_E \rightarrow \ket{\Psi\(t\)}=\sum_i c_i\(t\) \ket{s_i\(t\)}_S \ket{e_i\(t\)}_E.
\ee
In fact for any choice of initial state and time evolution this decomposition can be performed exactly at any moment in time for a particular choice of orthonormal bases for the system and environment (the Schmidt decomposition). 
However, as the time dependence indicates, decompositions at different times are generically unrelated; in particular, the state $\ket{s_i\(\tau\)}_S \ket{e_i\(\tau\)}_E$ is not the Hamiltonian evolution of $\ket{s_i\(t\)}_S \ket{e_i\(t\)}_E$.

For the particular states and interactions that admit decoherence, however, there exists, at least approximately, a decomposition of the entangled state into ``branches'' which evolve independently of each other; that is, a choice of bases in which ${U}_{\tau - t} \ket{s_i\(t\)}_S \ket{e_i\(t\)}_E \approx  \ket{s_i\(\tau\)}_S \ket{e_i\(\tau\)}_E \:\forall i$, so that the $c_i$ are constant\footnote{If the Hamiltonian is time-dependent, as it is, for example, if the interaction only occurs in a specific period of time, then we should interpret this condition as holding in some finite time interval. Intuitively this condition says that, after the measurement-causing interaction occurs, the environment should, at least temporarily, record the state of the system\cite{Halliwell:1999xh,2012NJPh...14h3010J}.}:
\be
\ket{\Psi\(t\)} \approx \sum_i c_i \ket{s_i\(t\)}_S \ket{e_i\(t\)}_E. \label{eq:branching}
\ee
Hence there is a one-to-one association of system states $\ket{s_i}_S$ and environment states $\ket{e_i}_E$ in the ``pointer basis'' given by the decomposition. 
It immediately follows\footnote{Actually \Eq{eq:branching} is more general: it describes a situation in which the system-environment product kets are orthonormal but in which the system or environment kets need not individually be orthonormal. In particular it is easy to imagine situations (for example, a measurement apparatus that can record the state of a spin in multiple different bases) in which the system states $\ket{s_i}_S$ need not be orthonormal, or even where the sum is over a larger number of terms than the dimensionality of the system. In this case we should not expect the reduced density matrix to be a good record of the actual branches. See Sec.\ 2 of \cite{Boddy:2014eba} for further discussion of this point. In practice we expect that we can deal with such cases by moving degrees of freedom from the environment into the system (in the above example, the choice of which basis to measure in) until the system states are themselves orthogonal.} from \Eq{eq:branching} that the reduced density matrix describing the state of the system is
\be
\rho_S = \Tr_E \ketbra{\Psi}{\Psi} \approx |c_i|^2 \ketbra{s_i}{s_i}, \label{eq:decohered}
\ee
so the system can be described to good approximation as a statistical mixture of the states $\ket{s_i}_S$, in agreement with the action of the Born rule on the system alone.
If we only have access to the information in the system, we can check for the presence of decoherence by looking for a choice of basis in which the reduced density matrix becomes, and remains, approximately diagonal.

We comment briefly on the physical significance of the conditions (\ref{eq:branching}),(\ref{eq:decohered}). 
The branch label $i$ picks out a distinct state of the system and, crucially, a state of the environment $\ket{e_i}_E$ which is one-one correlated with the state of the system.
Because each of the environment states $\ket{e_i}_E$ has zero (or very small) overlap with the environment state associated with other system states, we say that the environment is \textit{monitoring} the state of the system, or keeping a record of it.
Again, this is a dynamical process: the environment starts in a particular initial ready state which is not entangled with the initial state of the system, but interactions between the system and environment cause the environment to \textit{record} the state of the system.
It is often convenient to decompose the Hamiltonian generating time evolution on the total Hilbert space into pieces denoting evolution in the system and environment alone, as well as an interaction Hamiltonian connecting the two factors:
\be
\Ham = \Ham_S + \Ham_E + \Ham_{\mathrm{int}}.
\ee
In general a decoherence analysis requires conditions on all of these components, but when decoherence occurs in the limit that the interaction strength is much larger than the other two terms (for a suitable choice of norm), the branches are simply given by the eigenstates of the interaction Hamiltonian.

Because the overall state starts as a product state but ends as an entangled superposition of branches, we see that decoherence is associated with entropy production, visible as the Shannon entropy of the classical probability distribution $|c_i|^2$ over system states.
In fact the connection between decoherence, entropy growth, and the production of records in the environment can be made more precise \cite{2012NJPh...14h3010J}.
In laboratory settings, for example when the environment includes photons and air molecules bouncing off an experimental apparatus, we expect that the environment in fact contains very many highly redundant records of the system state \cite{Halliwell:1999xh}.

We emphasize that in most setups the situations which lead to decoherence are non-generic.
The decoherence program requires, in particular, an initial (low-entropy) product state between the system and the environment, a special initial ``ready'' state of the environment which will subsequently be able to record the state of the system, and dynamics which allow the system to interact with the environment while still admitting effective non-dissipative evolution in the system alone after branching has occurred.
If, instead of analyzing a particular measurement apparatus, we want to use the decoherence formalism to determine which states are classical, we need to vary over some of these initial specifications.
In particular, if we don't start a preferred identification of the system, but instead, like in cosmology, wish to pick out the natural classical degrees of freedom, we need to vary over possible system-environment decompositions \cite{mereology}.

\subsection{State Reduction from Operational Constraints\label{subsec: State Reduction from Operational Constraints}}

In the study of decoherence we usually start by postulating the system-environment
split $\mathcal{H}=\mathcal{H}_{E}\otimes\mathcal{H}_{S}$. The state-reduction map $tr_{E}:\mathcal{H}\longrightarrow\mathcal{H}_{S}$
is then characterized by demanding 
\begin{equation}
tr\left(I_{E}\otimes O_{S}\,\rho\right)=tr\left(O_{S}\,tr_{E}\left(\rho\right)\right)\label{eq:partial trace map definition}
\end{equation}
for all $\rho$ and $O_{S}$, which leads to the definition of the
partial-trace map. The reduced state $tr_{E}\left(\rho\right)$ is
understood as the state of the subsystem $\mathcal{H}_{S}$ and unitary
evolution of $\rho$ (usually) results in a loss of coherence for $tr_{E}\left(\rho\right)$.

The operational justification for the system-environment split $\mathcal{H}_{E}\otimes\mathcal{H}_{S}$
comes from an assertion that only measurements of the form $I_{E}\otimes O_{S}$
are allowed. In the language of matrix algebras (see Section \ref{subsec:Matrix-Algebras-and their representation})
we can say that the allowed measurements $I_{E}\otimes O_{S}$ form
an algebra and the system-environment split $\mathcal{H}_{E}\otimes\mathcal{H}_{S}$
comes from the irrep decomposition of this algebra. By taking this
perspective we do not have to postulate the system-environment split;
instead we derive it as the irrep decomposition of the algebra of
allowed observables. This suggests a strictly operational approach
to decoherence where the algebra of allowed observables is the primary
object from which the Hilbert-space bipartition and the state-reduction
map are derived.

In this operational approach, we start with a Hilbert space $\mathcal{H}$
and an algebra $\mathcal{A}\subseteq\mathcal{L}\left(\mathcal{H}\right)$
that reflects our operational constraints. The assumption is that
in principle, all observables $O\in\mathcal{A}$ can be measured,
but nothing else. This is the generalization of the earlier assumption that only
observables of the form $I_{E}\otimes O_{S}$ are allowed. This of
course may be an overstatement of the practical reality, in which not all
$O\in\mathcal{A}$ are in fact measurable, but it is still a useful
assumption that outlines what definitely cannot be measured. (Similarly, when
we make the usual system-environment split we do not actually consider
all $I_{E}\otimes O_{S}$ to be measurable, but it is still a useful
assumption that outlines the boundary of the inaccessible environment). In Sections \ref{PartialBPT}-\ref{sec:ising} below we will introduce a more flexible notion of bipartition that captures restrictions to observables that do not have to form an algebra.

With the algebra $\mathcal{A}\subseteq\mathcal{L}\left(\mathcal{H}\right)$
the Hilbert space decomposes into the generalized bipartition (see
Theorem \ref{thm:Wedderburn decomposition})
\begin{equation}
\mathcal{H}\cong\bigoplus_{q}\mathcal{H}_{E_{q}}\otimes\mathcal{H}_{S_{q}}
\end{equation}
where only subsystems $\mathcal{H}_{S_{q}}$ are accessible with observables
restricted to $O\in\mathcal{A}$. This decomposition generalizes the
usual system-environment split in that it can identify multiple superselection
sectors each of which is split into system and environment. The superselection
sectors are manifestations of the fact that superpositions between state
vectors in different sectors are unobservable and unpreparable with
the given operational constraints. The reduced Hilbert space is therefore
given by
\begin{equation}
\mathcal{H}_{\left\{ S_{q}\right\} }:=\bigoplus_{q}\mathcal{H}_{S_{q}}
\end{equation}
where the observables $\bigoplus_{q}I_{E_{q}}\otimes O_{S_{q}}\in\mathcal{A}$
reduce to $\bigoplus_{q}O_{S_{q}}$. Now the state-reduction map $tr_{\left\{ E_{q}\right\} }$
can be defined in two steps: first, impose the superselection rules; second, discard
the environments:
\begin{equation}
tr_{\left\{ E_{q}\right\} }:\rho\longmapsto\bigoplus_{q}\Pi_{q}\rho\Pi_{q}\longmapsto\bigoplus_{q}tr_{E_{q}}\left(\Pi_{q}\rho\Pi_{q}\right)\label{eq:state reduction for MAs}
\end{equation}
where $\Pi_{q}$ are projections on the superselection sectors. Finally,
the analog of Eq.\ (\ref{eq:partial trace map definition}) 
\begin{equation}
tr\left(\left(\bigoplus_{q}I_{E_{q}}\otimes O_{S_{q}}\right)\,\rho\right)=tr\left(\left(\bigoplus_{q}O_{S_{q}}\right)\,tr_{\left\{ E_{q}\right\} }\left(\rho\right)\right)
\end{equation}
can be shown to hold by considering the trace on each sector $q$
separately and applying Eq.\ (\ref{eq:partial trace map definition}).

We can now see that restriction of observables to an algebra manifests
itself in two ways: superselection and system-environment split. Superselection
is responsible for eliminating some of the reduced state's coherence
terms by fiat; since no observable that could detect such
coherences is measurable in principle. The system-environment split,
on the other hand, is responsible for eliminating the coherence terms
dynamically. That is, even if some superpositions could be detected
in principle, they become entangled with the environment so rapidly that
we cannot actually see them; this is the idea of environment-induced
superselection\emph{ }or\emph{ einselection} \cite{zurek2003decoherence}. In general,
both superselection and einselection can play a role in the appearance
of classical reality.

A very simple case of classicality from superselection comes up when
we restrict the measurements to a single observable $O$. The algebra
generated by $O$ is spanned by the spectral projections $\Pi_{k}$
(see Proposition \ref{prop:MA of single operator }) associated with
the distinct measurement outcomes. The irrep decomposition is then
the decomposition of $\mathcal{H}$ into the eigenspaces of $O$
\begin{equation}
\mathcal{H}\cong\bigoplus_{k}\mathcal{H}_{E_{k}}\otimes h_{S_{k}}\cong\bigoplus_{k}\mathcal{H}_{E_{k}}
\end{equation}
where the system parts $h_{S_{k}}$ are one-dimensional and can be
absorbed into $\mathcal{H}_{E_{k}}$. The state-reduction map (\ref{eq:state reduction for MAs})
then becomes 
\begin{equation}
\rho\longmapsto\bigoplus_{k}tr\left(\Pi_{k}\rho\right)
\end{equation}
which is the reduction of $\rho$ into a classical probability distribution
over the outcomes $k$. Therefore, when only one observable can be
measured all quantum states are operationally equivalent to classical
probability distributions, and no coherence effects
can be observed.

The more interesting cases involve more than one observable. For example,
in a laboratory settings it is common to have a single readout (measurement)
operation $O$ supplemented by a set of control operations $\left\{ U_{\alpha}\right\} $.
Then the allowed measurements consist of the set $\left\{ O_{\alpha}:=U_{\alpha}^{\dagger}OU_{\alpha}\right\} $
for all $\alpha$. Such sets can be as simple as position and momentum
$\left\{ X,P\right\} $ or the angular momentum operators $\left\{ J_{x},J_{y},J_{z}\right\}$. When the underlying system consists of many particles for
which we can only measure the collective version of these observables,
or when there is a single particle but the observables have limited
resolution, we can expect non-trivial manifestations of superselection
and einselection effects.

This leads us to the main technical difficulty of the operational
approach: finding the irrep decomposition of algebras generated by
$\left\{ O_{\alpha}\right\} $. In the cases where $\left\{ O_{\alpha}\right\} $
forms a group with a known representation structure, the irrep decomposition
is given by the group's irreps (see Proposition \ref{prop:group algebra irreps}).
In other cases, however, we need a systematic way of constructing
the irrep decomposition from the generating set of observables $\left\{ O_{\alpha}\right\} $.
The solution of this problem is the subject of the next section, and one of the main technical results of this paper.

\section{Irrep Decomposition of Matrix Algebras by Scattering of Projections\label{sec:Irrep-Decompositions-of}}

\noindent The problem that we will address here is the following:
\begin{quote}
	Given a finite set of self-adjoint matrices $\left\{ M_{1},M_{2},...,M_{n}\right\} $
	that generate the algebra $\mathcal{A}$, find the irrep decomposition
	of $\mathcal{A}$ as in Theorem \ref{thm:Wedderburn decomposition}.
\end{quote}
As was discussed in Section \ref{subsec:Matrix-Algebras-and their representation},
the explicit specification of an irrep structure can be given by a
choice of basis arranged into a bipartition table (BPT), with the columns
specifying a maximal set of minimal projections (MSMP) and the alignment
of rows specifying the partial isometries that map between the columns.
Conversely, given an MSMP and partial isometries that map between them,
we can construct a BPT by following our definitions of the rows and columns.
This suggests that in order to find the irrep structure of an algebra,
we need to find an MSMP and the partial isometries that map between
them.

According to its definition (Definition \ref{def:MSMP}) an MSMP
is called maximal because it resolves the identity element of the algebra, but this
does not mean that it alone can generate the whole algebra. In the
BPT picture, the elements of the MSMP determine the columns
but are oblivious to how the columns are aligned with each other.
In order to construct the BPT, we will only need to supplement
the MSMP with additional minimal projections that will allow it
to generate the algebra. These additional projections define the partial
isometries that map between the elements of the MSMP, which determines
the alignment of columns in the BPT. The main task of the irrep decomposition
algorithm is then to find a set of minimal projections that generates
the algebra and contains an MSMP.

Before we go into specifics, let us outline the 4 main steps of the
algorithm that we develop in this section:
\begin{enumerate}
	\item Construct the initial set of projections from the spectral projections
	of the generators $\left\{ M_{1},M_{2},...,M_{n}\right\} $.
	\item Keep applying the rank-reducing operation called \emph{scattering
	}on the set of projections until no further reduction is possible;
	this produces the final set of projections.
	\item Verify that the final set of projections (which generates the algebra
	by construction) consists of minimal projections and contains an MSMP.
	\item Use the final set of projections to construct the BPT.
\end{enumerate}
Step 1 is a conversion of the input from self-adjoint operators to
their spectral projections. Step 2, the heart of the algorithm,
uses the scattering operation that we will define in Section \ref{subsec:Scattering-of-Projections}.
Step 3 is necessary because Step 2 is not guaranteed to produce minimal
projections (although this is what happens in practice); we will
explain how to deal with this in Section \ref{subsec:Minimality-and-Completnes}.
Step 4 is the construction of the basis elements that populate the
rows and columns of the BPT, which we will define in Section \ref{subsec:Construction-of-Bipartition}.
The formal definition of the algorithm and the proof of its correctness
are deferred to Section \ref{subsec:The-Algorithm}.

In the following, it will be beneficial to have a concrete example
to consider as we go over the details of the algorithm. For this purpose,
we will now introduce a toy example that will be used throughout this
section to illustrate the steps of the algorithm.

\subsubsection*{Toy Example}

The example that we will consider here is a quantum system described
by an eight-dimensional Hilbert space. The system itself and the measurements
that we will consider are not motivated by physical considerations,
but by their simplicity and ability to illustrate the key aspects
of the algorithm. More physically-motivated examples will be considered
in Section \ref{sec:Examples-of-State} below.

The toy example consists of the Hilbert space $\mathcal{H}$, spanned
by the eight basis elements $\left\{ \ket i\right\} _{i=1,...,8}$,
and two incompatible projective measurements given by the self-adjoint
operators $Z$ and $X$. (This choice of names is only meant to be suggestive of
their non-commutativity; we will remain agnostic to the physical nature
of this system.) The problem is to find the irreps of the algebra
$\left\langle Z,X\right\rangle $ which will allow us to simultaneously
block-diagonalize the two non-commuting observables. Once we have
this structure, it will be apparent what information encoded in the
quantum states is accessible with the measurements $Z,X$, and what
is not.

The observables $Z$ and $X$ have two outcomes associated with the
spectral projections $\left\{ \Pi_{Z;1},\Pi_{Z;2}\right\} $ and $\left\{ \Pi_{X;1},\Pi_{X;2}\right\} $\footnote{We do not need to know their eigenvalues but we will assume that they
	are nonzero. We can always shift all eigenvalues of the observable,
	without changing any physical predictions, so none of them are zero.} that sum to the identity. The spectral projections are defined as
follows:
\begin{align}
\Pi_{Z;1} & :=\ket 1\bra 1+\ket 2\bra 2+\ket 3\bra 3+\ket 4\bra 4\\
\Pi_{X;1} & :=\ket{_{-}^{+37}}\bra{_{-}^{+37}}+\ket{_{-}^{+1256}}\bra{_{-}^{+1256}},
\end{align}
where we have used the shorthand notation 
\begin{equation}
\ket{_{-j_{1},j_{2},...}^{+i_{1},i_{2},...}}:=\frac{1}{\sqrt{N}}\left(\ket{i_{1}}+\ket{i_{2}}+...-\ket{j_{1}}-\ket{j_{2}}-...\right)
\end{equation}
($\sqrt{N}$ is the normalization) so
\begin{align}
\ket{_{-}^{+37}} & :=\frac{1}{\sqrt{2}}\left(\ket 3+\ket 7\right)\\
\ket{_{-}^{+1256}} & :=\frac{1}{2}\left(\ket 1+\ket 2+\ket 5+\ket 6\right).
\end{align}
Their complementary projections are given by $\Pi_{Z;2}:=I-\Pi_{Z;1}$
, $\Pi_{X;2}:=I-\Pi_{X;1}$.

As was discussed in Proposition \ref{prop:MA of single operator },
the spectral projections of each self-adjoint operator are part of
the algebra that it generates, and the algebra $\left\langle Z,X\right\rangle $
is also generated by $\left\langle \Pi_{Z;1},\Pi_{Z;2},\Pi_{X;1},\Pi_{X;2}\right\rangle $.
This replacement of generators from self-adjoint matrices to their
spectral projections is Step 1 of the algorithm. We will continue
this example after we define and prove some facts about the scattering algorithm.

\subsection{Scattering of Projections \label{subsec:Scattering-of-Projections}}

\emph{Scattering} is the basic operation that we will use to break down the
spectral projections of the generators into smaller rank projections.
\begin{defn}
	\emph{\label{def:Scattering }Scattering} is an operation on a pair
	of projections $\Pi_{1}$, $\Pi_{2}$ that produces a pair of sets
	of projections $\left\{ \Pi_{1}^{\left(\lambda\right)}\right\} $,
	$\left\{ \Pi_{2}^{\left(\lambda\right)}\right\} $. The elements in
	each set come from the spectral decompositions
	\begin{align}
	\Pi_{1}\Pi_{2}\Pi_{1} & =\sum_{\lambda\neq0}\lambda\Pi_{1}^{\left(\lambda\right)}\label{eq:def of scattering 1}\\
	\Pi_{2}\Pi_{1}\Pi_{2} & =\sum_{\lambda\neq0}\lambda\Pi_{2}^{\left(\lambda\right)}\label{eq:def of scattering 2}
	\end{align}
	 (the sums are over unique non-zero eigenvalues $\lambda$) with the addition of null projections defined by\footnote{Null projections should not be confused with projections on the kernel
		of $\Pi_{i}\Pi_{j}\Pi_{i}$. The kernel projections are given by $I-\sum_{\lambda}\Pi_{i}^{\left(\lambda\right)}$
		which is not the same as Eq.\ (\ref{eq:def of null projections in scattering}).} 
	\begin{equation}
	\Pi_{i=1,2}^{\left(0\right)}:=\Pi_{i}-\sum_{\lambda\neq0}\Pi_{i}^{\left(\lambda\right)}.\label{eq:def of null projections in scattering}
	\end{equation}
\end{defn}

\noindent It will be very convenient to consider the null projections
$\Pi_{i}^{\left(0\right)}$ as just the $\lambda=0$ elements of the
set of spectral projections $\left\{ \Pi_{i}^{\left(\lambda\right)}\right\} $,
even when $\Pi_{i}^{\left(0\right)}=0$ in Eq.\ (\ref{eq:def of null projections in scattering}).
Also note that, although the definition does not say so explicitly,
the spectrum $\lambda$ in both Eq.\ (\ref{eq:def of scattering 1})
and (\ref{eq:def of scattering 2}) is the same (we will prove this in Lemma \ref{lem: scattering of projections } below).

From this definition, we see that all the projections in the set $\left\{ \Pi_{i}^{\left(\lambda\right)}\right\} $
are pairwise orthogonal and sum to their predecessor 
\begin{equation}
\Pi_{i}=\Pi_{i}^{\left(\lambda_{1}\right)}+\Pi_{i}^{\left(\lambda_{2}\right)}+...+\Pi_{i}^{\left(0\right)},\label{eq:projection sum from scattered projections}
\end{equation}
so they are of lower rank than their predecessor $\Pi_{i}$. Thus,
in analogy with the scattering of particles, the scattering of projections
``breaks'' them into smaller constituents (the ``interaction'' in this
analogy is the adjoint action of Eq.\ (\ref{eq:def of scattering 1}),(\ref{eq:def of scattering 2}))

\begin{equation}
\begin{array}{c}
\Pi_{1}\\
\\
\Pi_{2}
\end{array}
\Diagram{fdA &  & fuA\\
	& f\\
	fuA &  & fdA
}
\begin{array}{c}
\Pi_{1}^{\left(\lambda_{1}\right)}+\Pi_{1}^{\left(\lambda_{2}\right)}+...+\Pi_{1}^{\left(0\right)}\\
\\
\Pi_{2}^{\left(\lambda_{1}\right)}+\Pi_{2}^{\left(\lambda_{2}\right)}+...+\Pi_{2}^{\left(0\right)}
\end{array}
\end{equation}

This defines scattering in the general case. There is also a special
case that is important enough to have its own definition:
\begin{defn}
	A pair of projections $\Pi_{1}$, $\Pi_{2}$ is called \emph{reflecting
	}if both projections remain unbroken by scattering, that is 
	\begin{align}
	\Pi_{1}\Pi_{2}\Pi_{1} & =\lambda\Pi_{1}\label{eq:def of reflecting 1}\\
	\Pi_{2}\Pi_{1}\Pi_{2} & =\lambda\Pi_{2}\label{eq:def of reflecting 2}
	\end{align}
	where the coefficient $\lambda$ is called the \emph{reflection} \emph{coefficient}.
	We will say that $\Pi_{1}$, $\Pi_{2}$ are \emph{properly} \emph{reflecting
	}if the reflection coefficient is not $0$ (i.e. they are not orthogonal,
	$\Pi_{1}\Pi_{2}\neq0$)
\end{defn}

It should be clear that rank 1 projections are always reflecting (however, reflecting projections can be of any rank). 
Another couple of useful facts about reflecting projections are given
by the following proposition:
\begin{prop}
	\label{prop:reflecting proj}Let $\Pi_{1}$, $\Pi_{2}$ be a pair
	of properly reflecting projections with the reflection coefficient
	$\lambda\neq0$, then:
	
	(1) $\Pi_{1}$ and $\Pi_{2}$ have the same rank.
	
	(2) $\Pi_{1}=\Pi_{2}$ if $\lambda=1$.
\end{prop}

\begin{proof}
	We take the trace on both sides of Eq.\ (\ref{eq:def of reflecting 1}),
	(\ref{eq:def of reflecting 2}) and use the cyclic property of the
	trace to get 
	\begin{align}
	tr\left(\Pi_{1}\Pi_{2}\right)= & \lambda tr\left(\Pi_{1}\right)\\
	tr\left(\Pi_{1}\Pi_{2}\right)= & \lambda tr\left(\Pi_{2}\right)
	\end{align}
	Since the $\lambda$'s are the same (this will be proven in general in
	Lemma \ref{lem: scattering of projections }) then $tr\left(\Pi_{1}\right)=\frac{tr\left(\Pi_{1}\Pi_{2}\right)}{\lambda}=tr\left(\Pi_{2}\right)$
	so they must have the same rank.
	
	If, in addition, $\lambda=1$ then
	\begin{align}
	0= & \Pi_{1}-\Pi_{1}\Pi_{2}\Pi_{1}=\left(\Pi_{1}-\Pi_{1}\Pi_{2}\right)\left(\Pi_{1}-\Pi_{1}\Pi_{2}\right)^{\dagger}\\
	0= & \Pi_{2}-\Pi_{2}\Pi_{1}\Pi_{2}=\left(\Pi_{2}-\Pi_{2}\Pi_{1}\right)\left(\Pi_{2}-\Pi_{2}\Pi_{1}\right)^{\dagger}
	\end{align}
	so 
	\begin{align}
	0 & =\Pi_{1}-\Pi_{1}\Pi_{2}\\
	0 & =\left(\Pi_{2}-\Pi_{2}\Pi_{1}\right)^{\dagger}.
	\end{align}
	Thus, $\Pi_{1}=\Pi_{1}\Pi_{2}=\Pi_{2}$.
\end{proof}
The importance of reflecting projections is that they do not break
under scattering (this choice of terminology is a continuation of our
commitment to the analogy with particles). In Step 2 of the algorithm,
we will apply the scattering operation on pairs of projections until
no further reduction is possible. The impossibility of reduction is
then the case of all projections being pairwise reflecting. We are
guaranteed to reach entirely reflecting projections because scattering produces
projections of smaller rank (unless it reflects) and projections
of rank 1 are always reflecting.

The most important fact about scattering is that regardless of what
the initial projections $\Pi_{1}$, $\Pi_{2}$ are, the resulting
projections are a series of reflecting pairs $\left\{ \Pi_{1}^{\left(\lambda\right)},\Pi_{2}^{\left(\lambda\right)}\right\} $
with reflection coefficients $\lambda$, and every pair $\left\{ \Pi_{1}^{\left(\lambda\right)},\Pi_{2}^{\left(\lambda\right)}\right\} $ is orthogonal
to any other pair $\left\{ \Pi_{1}^{\left(\lambda^\prime\right)},\Pi_{2}^{\left(\lambda^\prime\right)}\right\} $. 
\begin{lem}
	\label{lem: scattering of projections }Let $\Pi_{1}$, $\Pi_{2}$
	be the initial projections and $\left\{ \Pi_{1}^{\left(\lambda\right)}\right\} $,
	$\left\{ \Pi_{2}^{\left(\lambda\right)}\right\} $ be the sets of
	post scattering projections given by Definition \ref{def:Scattering },
	then:
	
	\noindent (1) The spectrum of eigenvalues $\lambda$ is the same in
	both sets.
	
	\noindent (2) For all $\lambda\neq\lambda'$ the pairs of projections
	$\Pi_{1}^{\left(\lambda\right)}$, $\Pi_{2}^{\left(\lambda'\right)}$
	are orthogonal.
	
	\noindent (3) For all $\lambda$ the pairs of projections $\Pi_{1}^{\left(\lambda\right)}$,
	$\Pi_{2}^{\left(\lambda\right)}$ are reflecting with reflection coefficient
	$\lambda$.
\end{lem}

\begin{proof}
	We begin by taking $\lambda$ and $\Pi_{1}^{\left(\lambda\right)}$
	to be the eigenvalues and the spectral projections in the decomposition
	of $\Pi_{1}\Pi_{2}\Pi_{1}$ and assume nothing about the spectral
	decomposition of $\Pi_{2}\Pi_{1}\Pi_{2}$.
	
	First, note that $\Pi_{1}^{\left(\lambda\right)}\Pi_{1}=\Pi_{1}\Pi_{1}^{\left(\lambda\right)}=\Pi_{1}^{\left(\lambda\right)}$
	for all $\lambda$ as can be seen from Eq.\ (\ref{eq:projection sum from scattered projections})
	and the fact that all spectral projections (including $\Pi_{1}^{\left(0\right)}$)
	are pairwise orthogonal. Then, if we act on both sides of Eq.\ (\ref{eq:def of scattering 1})
	by the adjoint with $\Pi_{1}^{\left(\lambda\right)}$ and $\Pi_{1}^{\left(\lambda'\right)}$
	we get
	\begin{equation}
	\Pi_{1}^{\left(\lambda\right)}\Pi_{2}\Pi_{1}^{\left(\lambda'\right)}=\delta_{\lambda\lambda'}\lambda\Pi_{1}^{\left(\lambda\right)}.\label{eq:scat of proj proof identity 1}
	\end{equation}
	This equation holds for all $\lambda$ including $\lambda=0$, and
	it does not matter whether $\Pi_{1}^{\left(0\right)}$ vanishes ($\Pi_{1}^{\left(0\right)}=0$)
	or not. In particular $\Pi_{1}^{\left(0\right)}\Pi_{2}\Pi_{1}^{\left(0\right)}=0$
	so $\Pi_{1}^{\left(0\right)}\Pi_{2}=0$ because otherwise we would reach a contradiction,
	\begin{equation}
	0\neq\left(\Pi_{1}^{\left(0\right)}\Pi_{2}\right)\left(\Pi_{1}^{\left(0\right)}\Pi_{2}\right)^{\dagger}=\left(\Pi_{1}^{\left(0\right)}\Pi_{2}\right)\left(\Pi_{2}\Pi_{1}^{\left(0\right)}\right)=\Pi_{1}^{\left(0\right)}\Pi_{2}\Pi_{1}^{\left(0\right)}=0.
	\end{equation}
	Therefore, $\Pi_{1}^{\left(0\right)}\Pi_{2}=\Pi_{2}\Pi_{1}^{\left(0\right)}=0$.
	This allows us to write 
	\begin{equation}
	\Pi_{2}\Pi_{1}\Pi_{2}=\Pi_{2}\left(\Pi_{1}-\Pi_{1}^{\left(0\right)}\right)\Pi_{2}=\Pi_{2}\left(\sum_{\lambda\neq0}\Pi_{1}^{\left(\lambda\right)}\right)\Pi_{2}=\sum_{\lambda\neq0}\lambda\left(\frac{1}{\lambda}\Pi_{2}\Pi_{1}^{\left(\lambda\right)}\Pi_{2}\right).\label{eq: scat of proj proof identity 2}
	\end{equation}
	The last step suggests the definition
	\begin{equation}
	\tilde{\Pi}_{2}^{\left(\lambda\right)}:=\frac{1}{\lambda}\Pi_{2}\Pi_{1}^{\left(\lambda\right)}\Pi_{2}.\label{eq:scat of proj proof identity 3}
	\end{equation}
	These operators are clearly self-adjoint and, using Eq.\ (\ref{eq:scat of proj proof identity 1}),
	we can see that 
	\begin{equation}
	\tilde{\Pi}_{2}^{\left(\lambda\right)}\tilde{\Pi}_{2}^{\left(\lambda'\right)}=\frac{1}{\lambda\lambda'}\Pi_{2}\Pi_{1}^{\left(\lambda\right)}\Pi_{2}\Pi_{1}^{\left(\lambda'\right)}\Pi_{2}=\delta_{\lambda\lambda'}\frac{1}{\lambda}\Pi_{2}\Pi_{1}^{\left(\lambda\right)}\Pi_{2}=\delta_{\lambda\lambda'}\tilde{\Pi}_{2}^{\left(\lambda\right)},
	\end{equation}
	so they are pairwise orthogonal projections. Since the $\lambda$'s are
	distinct and $\tilde{\Pi}_{2}^{\left(\lambda\right)}$ are pairwise
	orthogonal projections, Eq.\ (\ref{eq: scat of proj proof identity 2})
	must be the spectral decomposition of $\Pi_{2}\Pi_{1}\Pi_{2}$. Thus,
	$\tilde{\Pi}_{2}^{\left(\lambda\right)}=\Pi_{2}^{\left(\lambda\right)}$,
	and the spectrum is the same for both $\Pi_{1}\Pi_{2}\Pi_{1}$ and
	$\Pi_{2}\Pi_{1}\Pi_{2}$. This proves claim 1.
	
	Now, if we use Eq.\ (\ref{eq:scat of proj proof identity 3}) as the
	definition of $\Pi_{2}^{\left(\lambda'\right)}$ and simplify with
	Eq.\ (\ref{eq:scat of proj proof identity 1}) we get the identity
	\begin{equation}
	\Pi_{1}^{\left(\lambda\right)}\Pi_{2}^{\left(\lambda'\right)}=\Pi_{1}^{\left(\lambda\right)}\Pi_{2}\Pi_{1}^{\left(\lambda'\right)}\Pi_{2}\frac{1}{\lambda'}=\delta_{\lambda\lambda'}\Pi_{1}^{\left(\lambda\right)}\Pi_{2}.
	\end{equation}
	This proves claim 2. In particular, for $\lambda=\lambda'$, if we
	multiply this identity with its own adjoint on both sides and again
	use Eq.\ (\ref{eq:scat of proj proof identity 1}) and (\ref{eq:scat of proj proof identity 3}),
	we get
	\begin{align}
	\Pi_{1}^{\left(\lambda\right)}\Pi_{2}^{\left(\lambda\right)}\Pi_{1}^{\left(\lambda\right)} & =\Pi_{1}^{\left(\lambda\right)}\Pi_{2}\Pi_{1}^{\left(\lambda\right)}=\lambda\Pi_{1}^{\left(\lambda\right)}\\
	\Pi_{2}^{\left(\lambda\right)}\Pi_{1}^{\left(\lambda\right)}\Pi_{2}^{\left(\lambda\right)} & =\Pi_{2}\Pi_{1}^{\left(\lambda\right)}\Pi_{2}=\lambda\Pi_{2}^{\left(\lambda\right)}
	\end{align}
	which proves claim 3.
\end{proof}
Lemma \ref{lem: scattering of projections } tells us that almost
all projections that come out of scattering are pairwise orthogonal.
In particular, each of the null projections $\Pi_{1}^{\left(0\right)}$,
$\Pi_{2}^{\left(0\right)}$ is orthogonal to all other projections
and only the pairs $\Pi_{1}^{\left(\lambda\right)}$, $\Pi_{2}^{\left(\lambda\right)}$
for $\lambda\neq0$ are not orthogonal but properly reflecting. It
is also interesting to note that if there is $\lambda=1$ in the spectrum
then $\Pi_{1}^{\left(1\right)}=\Pi_{2}^{\left(1\right)}$ (see Proposition
\ref{prop:reflecting proj}), which occurs if the initial projections
project onto intersecting subspaces so $\Pi_{i}^{\left(1\right)}$ is
the projection on their intersection. We can avoid scattering these projections twice in future iterations of the algorithm by eliminating such duplicates. Lastly, note that Eq.\ (\ref{eq:scat of proj proof identity 3}) tells
us how to get the post-scattering projections $\Pi_{2}^{\left(\lambda\neq0\right)}$
from the post-scattering projections $\Pi_{1}^{\left(\lambda\neq0\right)}$
(for $\Pi_{i}^{\left(\lambda=0\right)}$ we use Eq.\ (\ref{eq:def of null projections in scattering}))
so we only need to calculate the spectral decomposition once for $\Pi_{1}\Pi_{2}\Pi_{1}$.

We now define a graph structure for a set of projections:
\begin{defn}
	A (proper)\emph{ reflection network }associated with the set of reflecting
	projections $\left\{ \Pi_{v}\right\} $ is the graph $G=\left\{ V,E\right\} $
	where the vertices are the projections $V:=\left\{ \Pi_{v}\right\} $
	and every properly reflecting pair is connected by an edge $E:=\left\{ \left(\Pi_{v},\Pi_{u}\right)\,|\,\Pi_{v}\Pi_{u}\neq0\right\} $
	(only orthogonal reflecting projections do not share an edge). An
	\emph{improper} \emph{reflection network }is the generalization of
	the above where not all projections are known to be reflecting. In
	that case, there are two kinds of edges: one kind for properly reflecting
	pairs (black solid edge) and one for unknowns (red dashed edge).
\end{defn}

\noindent In general, reflection networks may have multiple connected
components formed by subsets of projections that are orthogonal to
every projection outside the subset. It does not mean, however, that
projections in the same connected component cannot be orthogonal;
as long as there is a sequence of proper reflection (or unknown) relations
connecting the projections, they will be in the same component. Also
note that, according to Proposition \ref{prop:reflecting proj}, all
projections in the same connected component of a proper reflection
network must be of the same rank.

We will now consider how the scattering operation affects the reflection
network by focusing on a pair of projections in the network. According
to Lemma \ref{lem: scattering of projections }, in general a pair
of projections $\left\{ \Pi_{1},\Pi_{2}\right\} $ with unknown relations
(red edge) scatters into a series of pairs of reflecting projections $\{\Pi_{1}^{\left(\lambda\right)}$, $\Pi_{2}^{\left(\lambda\right)}\}$
(black edges unless $\lambda=0$ then no edge), and each pair in the
series is orthogonal to all other pairs (no edges); see Fig \ref{fig:scattering update rule}(a).
Fig \ref{fig:scattering update rule}(b) illustrates the special case
where $\Pi_{1}$ did not break under scattering so $\Pi_{1}\equiv\Pi_{1}^{\left(\lambda_{1}\right)}$.
The case where both $\left\{ \Pi_{1},\Pi_{2}\right\} $ do not break
(not shown) implies that they are reflecting and the red edge between
them is set to black or omitted, depending on whether $\lambda=0$.
\begin{figure}[t]
	\fbox{\begin{minipage}[b]{0.56\columnwidth}%
			\begin{flushleft}
				(a)
				\par\end{flushleft}
			\begin{center}
				\includegraphics[height=0.11\paperheight]{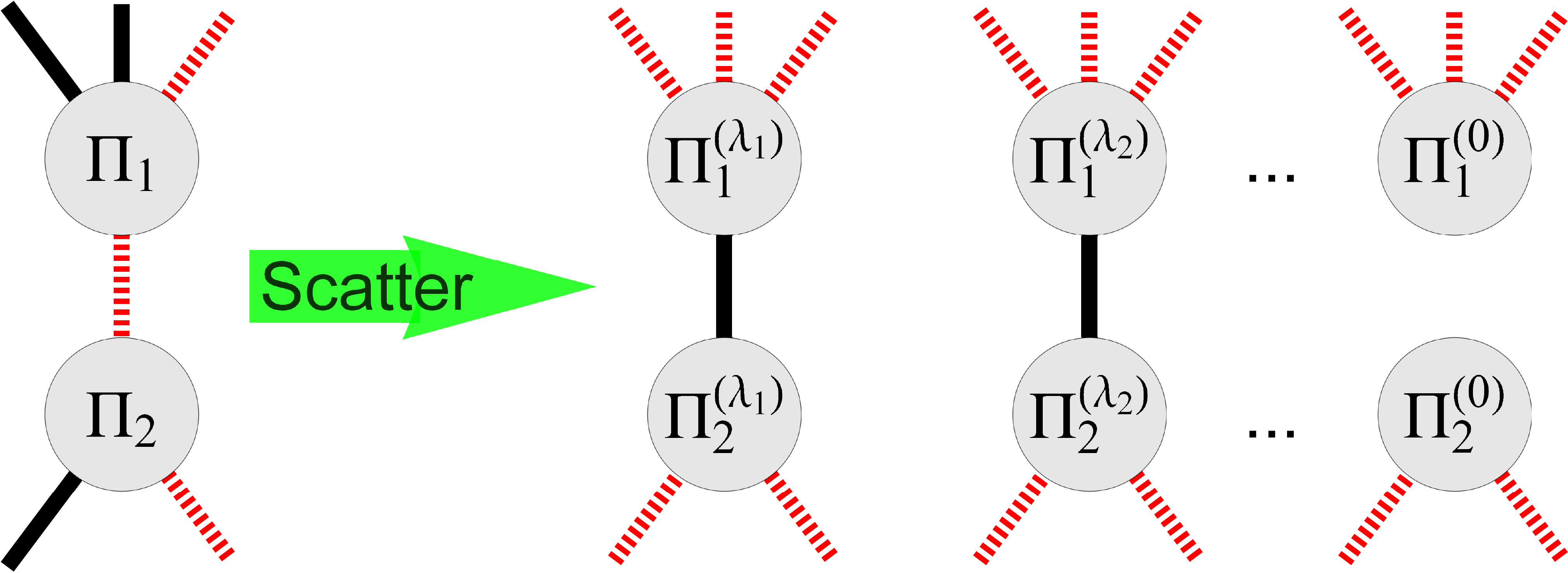}
				\par\end{center}%
	\end{minipage}} %
	\fbox{\begin{minipage}[b]{0.4\columnwidth}%
			\begin{flushleft}
				(b)
				\par\end{flushleft}
			\begin{center}
				\includegraphics[height=0.11\paperheight]{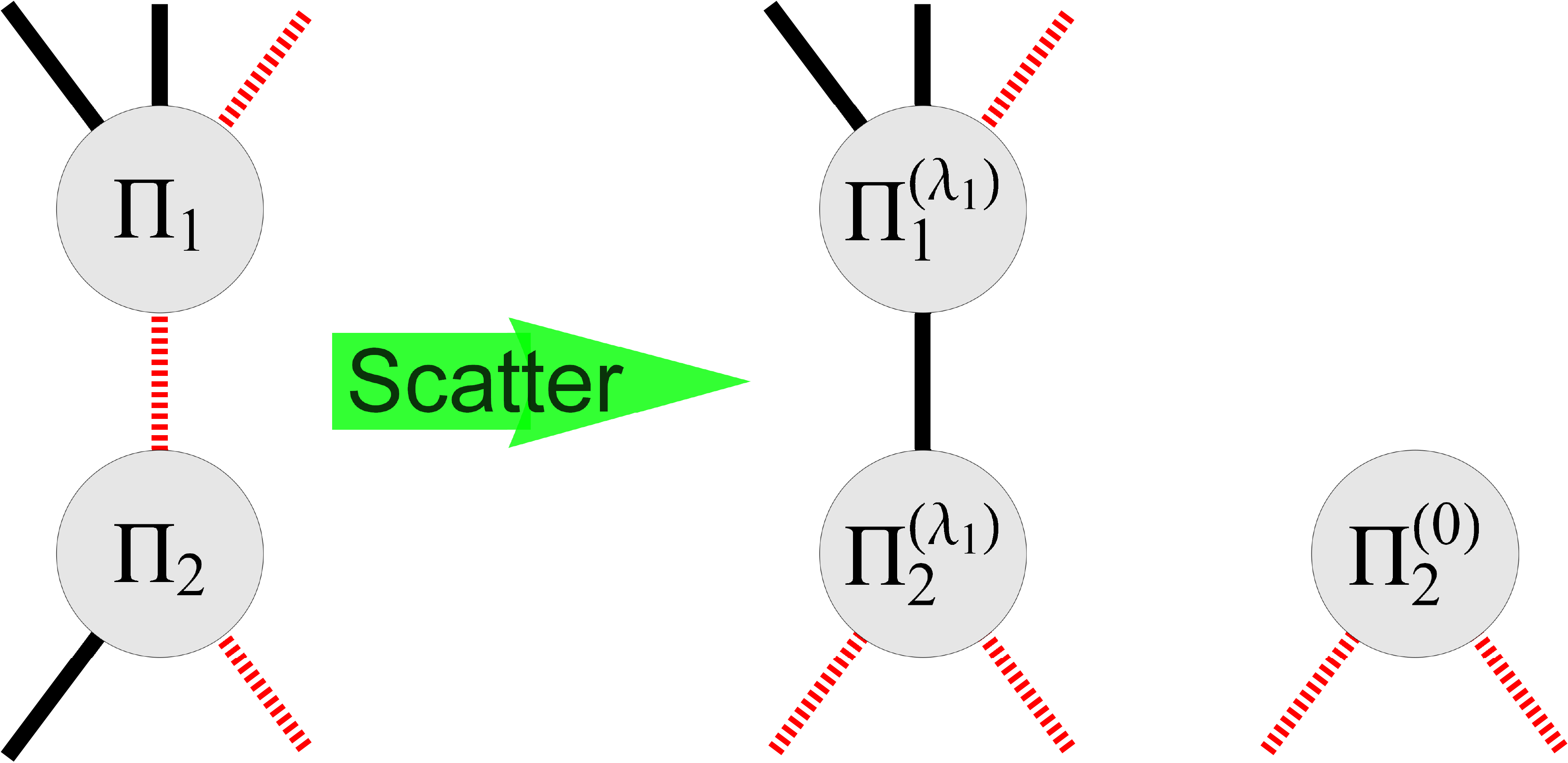}
				\par\end{center}%
	\end{minipage}}\caption{\label{fig:scattering update rule}Update rules for reflection relations
		after scattering. The red (dashed) edges represent unknown reflection relations,
		black (solid) edges represent properly-reflecting pairs, absent edges represent
		orthogonal pairs. One-sided edges stand for the reflection relations
		with other projections in the rest of the network. (a) In the generic
		case where each $\Pi_{i=1,2}$ breaks down into $\left\{ \Pi_{i}^{\left(\lambda_{k}\right)}\right\}$, the result is a series of properly reflecting pairs (for $\lambda=0$
		the pair is orthogonal) as described in Lemma \ref{lem: scattering of projections }.
		All the external edges are inherited by $\left\{ \Pi_{i}^{\left(\lambda_{k}\right)}\right\} $
		from $\Pi_{i}$ with the black (solid) edges being reset to red (dashed). (b) In the
		special case where $\Pi_{1}$ did not break down under scattering,
		we know $\Pi_{1}\equiv\Pi_{1}^{\left(\lambda_{1}\right)}$. In this case,
		$\Pi_{2}$ may break down to at most two projections (if it also did
		not break down then $\Pi_{1}$, $\Pi_{2}$ should just be relabeled as reflecting)
		such that $\Pi_{1}^{\left(\lambda_{1}\right)}$, $\Pi_{2}^{\left(\lambda_{1}\right)}$
		are properly reflecting and $\Pi_{2}^{\left(0\right)}$ is orthogonal
		to both. In (b) the update rule of external edges differs from the
		generic case (a) in that for unbroken projection $\Pi_{1}^{\left(\lambda_{1}\right)}$,
		the black (solid) edges are not reset to red (dashed).}
\end{figure}

Since both projections $\left\{ \Pi_{1},\Pi_{2}\right\} $ are part
of a larger network, we also have to specify how the resulting projections
$\left\{ \Pi_{i}^{\left(\lambda_{k}\right)}\right\} $ inherit the
relations with the rest of the elements in the network. First, we
note that orthogonality with other projections is preserved under
scattering so we do not need to add new edges that we did not already
have. Red edges also do not need to be updated since every unknown
relation that $\Pi_{i}$ had is still unknown for $\Pi_{i}^{\left(\lambda_{k}\right)}$.
Proper reflection relations, however, do not survive when one of the
projections is broken down into smaller rank projections, because
properly reflecting projections must have the same rank (see Proposition
\ref{prop:reflecting proj}). Therefore, black edges that $\Pi_{i}$
had before scattering should be reset to red when inherited by $\Pi_{i}^{\left(\lambda_{k}\right)}$,
unless the projection did not break, like in Fig \ref{fig:scattering update rule}(b),
in which case the black edges remain intact.

\subsubsection*{Procedure}

As we mentioned before, Step 1 of the algorithm produces the spectral
projections of the generators. Formally, we will refer to this step
of the algorithm as the procedure \textproc{GetAllSpectralProjections} but we will
not explicitly define it as it is self-evident.

We now have the definitions and the facts to define the procedure
of Step 2 of the algorithm:

\begin{algorithmic}[1] 
	
	\Procedure{ScatterAllProjections}{$SpecProjs$}
	\State $Projs \gets SpecProjs$
	\State $Relations \gets$ \Call{InitializeReflectionRelations}{$Projs$}
	\State $ReflectNet \gets \left\{Projs,  Relations \right\} $
	
	\While{\Call{IsEverythingReflecting}{$ReflectNet$} is  $\bold{false}$} 
	\State $Pair \gets$ \Call{PickNonReflectPair}{$ReflectNet$}
	\State $PostScatPair \gets$ \Call{ScatterProjectionsPair}{$Pair$}
	\State $ReflectNet \gets$ \Call{UpdateReflectionNetwork}{$ReflectNet, Pair, PostScatPair$}
	\EndWhile
	
	\State \Return $ReflectNet$
	\EndProcedure
	
\end{algorithmic}

The procedure starts by constructing the improper reflection network from the
initial spectral projections and initializing all edges to red except
the ones that are known to be reflecting (like rank 1 or orthogonal
projections). It then proceeds to iterations where it picks a pair
of projections connected by a red edge,\footnote{\noindent For better efficiency we should prioritize projections of
	lowest rank. Such projections are less likely to break down under scattering,
	which will reduce the number of resets of proper reflection relations
	that happen when we update the network after scattering.} scatters it,\footnote{As was discussed after Lemma \ref{lem: scattering of projections },
	the projections in \emph{Pair }may intersect on a subspace and the
	projection on this subspace will appear twice in \emph{PostScatPair.
	}Eliminating such duplicate projections is not necessary for the success
	of the algorithm but it will improve efficiency.} and updates the relations in the network according to the rules given
in Fig. \ref{fig:scattering update rule}. The procedure \textproc{ScatterAllProjections}
finishes when the reflection network is proper: that is, when all
projections are reflecting (all edges are black). This procedure is
guaranteed to terminate because every scattering iteration either identifies
a previously unknown reflecting pair or scatters a pair into a
series of reflecting pairs of lower rank. Eventually, all projections
will either be reflecting, or they will be reduced to rank 1 and thus again must be reflecting.

\subsubsection*{Toy Example (continued)}

\begin{figure}[t]
	\begin{centering}
		\fbox{\begin{minipage}[b]{0.6\columnwidth}%
				\begin{flushleft}
					(a)\qquad{}\includegraphics[height=0.1\paperheight]{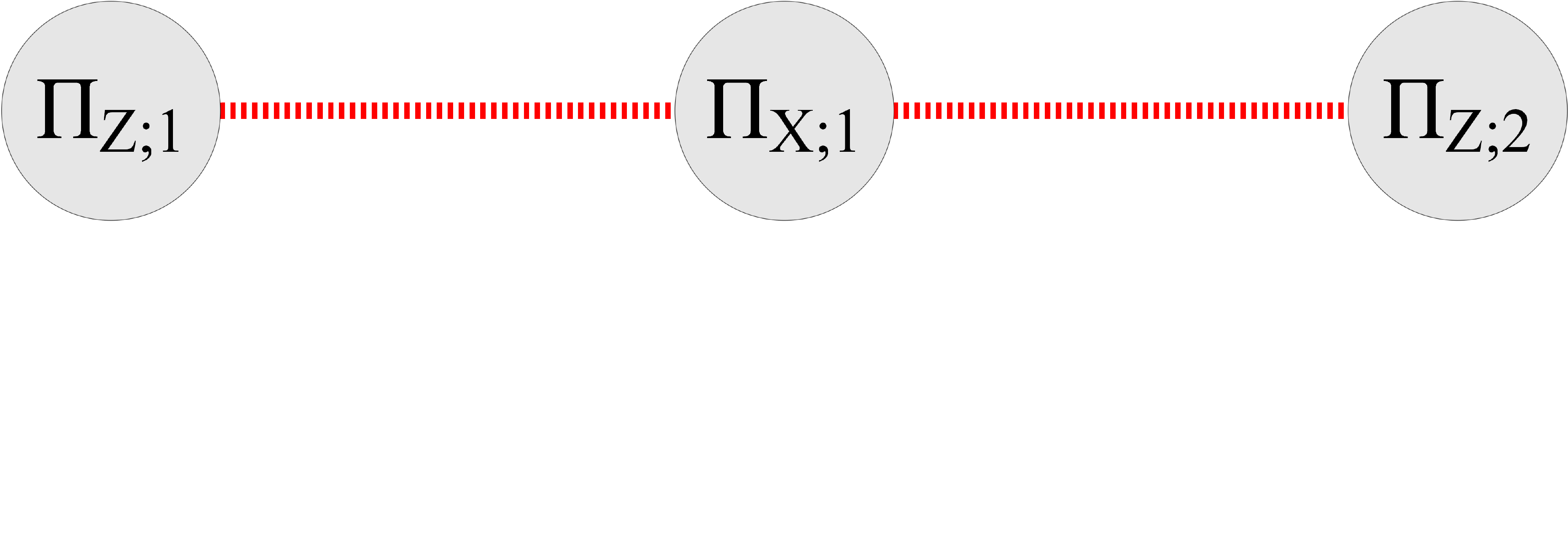}
					\par\end{flushleft}%
		\end{minipage}}
		\par\end{centering}
	\begin{centering}
		\fbox{\begin{minipage}[b]{0.6\columnwidth}%
				\begin{flushleft}
					(b)\qquad{}\includegraphics[height=0.1\paperheight]{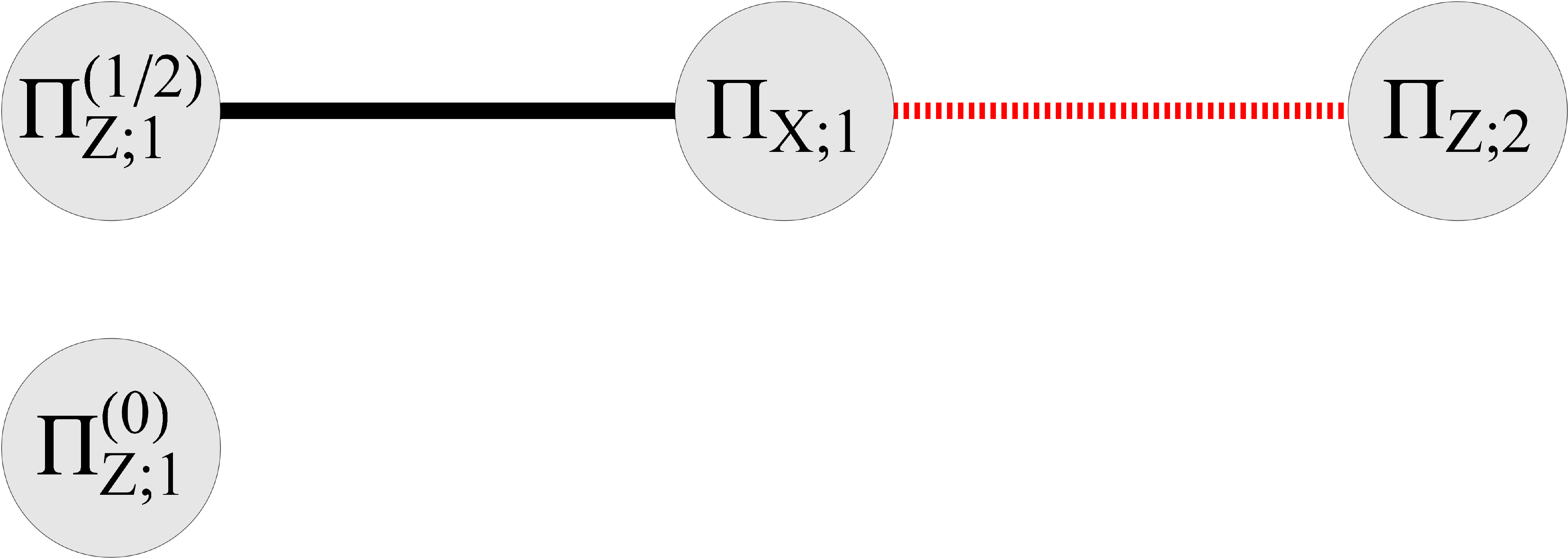}
					\par\end{flushleft}%
		\end{minipage}}
		\par\end{centering}
	\centering{}%
	\fbox{\begin{minipage}[b]{0.6\columnwidth}%
			\begin{flushleft}
				(c)\qquad{}\includegraphics[height=0.1\paperheight]{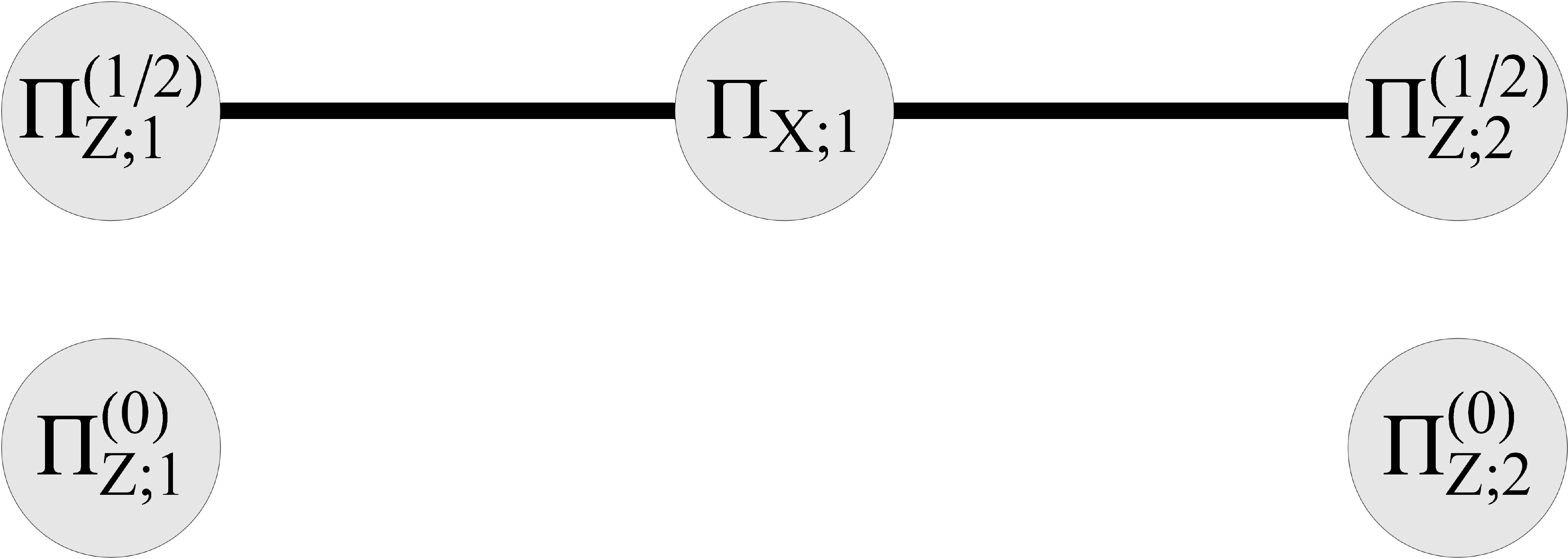}
				\par\end{flushleft}%
	\end{minipage}}\caption{\label{fig:toy example refnet} Evolution of the reflection network
		of the toy model during two scattering iterations. The red (dashed) edges represent unknown reflection relations,
		black (solid) edges represent properly reflecting pairs, absent edges represent
		orthogonal pairs. (a) Is the initial
		improper reflection network. (b) Is the intermediate network after
		one scattering iteration. (c) Is the final proper reflection network
		after two scatterings.}
\end{figure}

Before we consider the initial reflection network we note that one
of the four initial spectral projections $\left\{ \Pi_{Z;1},\Pi_{Z;2},\Pi_{X;1},\Pi_{X;2}\right\} $
is redundant in generating the algebra. That is because $\Pi_{X;2}=I-\Pi_{X;1}=\Pi_{Z;1}+\Pi_{Z;2}-\Pi_{X;1}$
so the algebra generated just by $\left\{ \Pi_{Z;1},\Pi_{Z;2},\Pi_{X;1}\right\} $
is the same as before.

The initial improper reflection network is show in Fig \ref{fig:toy example refnet}(a).
We begin by scattering the pair $\left\{ \Pi_{Z;1},\Pi_{X;1}\right\} $
\begin{align}
\Pi_{Z;1}\Pi_{X;1}\Pi_{Z;1}= & \Pi_{Z;1}\ket{_{-}^{+37}}\bra{_{-}^{+37}}\Pi_{Z;1}+\Pi_{Z;1}\ket{_{-}^{+1256}}\bra{_{-}^{+1256}}\Pi_{Z;1}\\
= & \frac{1}{2}\ket 3\bra 3+\frac{1}{2}\ket{_{-}^{+12}}\bra{_{-}^{+12}}.
\end{align}
There is only one eigenvalue $\lambda=\frac{1}{2}$ here which identifies
a single spectral projection 
\begin{equation}
\Pi_{Z;1}^{\left(1/2\right)}=\ket 3\bra 3+\ket{_{-}^{+12}}\bra{_{-}^{+12}}.
\end{equation}
Thus, the rank 4 projection $\Pi_{Z;1}$ breaks into two rank 2 projections
$\Pi_{Z;1}=\Pi_{Z;1}^{\left(1/2\right)}+\Pi_{Z;1}^{\left(0\right)}$
(see Definition \ref{def:Scattering }), where 
\begin{equation}
\Pi_{Z;1}^{\left(0\right)}=\ket 4\bra 4+\ket{_{-2}^{+1}}\bra{_{-2}^{+1}}.
\end{equation}

In principle, the breaking of the second projection $\Pi_{X;1}$ in
the scattering is calculated using Eq.\ (\ref{eq:scat of proj proof identity 3})
resulting in
\begin{equation}
\Pi_{X;1}^{\left(1/2\right)}=\frac{1}{1/2}\Pi_{X;1}\Pi_{Z;1}^{\left(1/2\right)}\Pi_{X;1}=\ket{_{-}^{+37}}\bra{_{-}^{+37}}+\ket{_{-}^{+1256}}\bra{_{-}^{+1256}}=\Pi_{X;1},
\end{equation}
which tells us that $\Pi_{X;1}$ did not break (when the scattering
has only one non-zero eigenvalue, as is this case, we already know at least one of
the projections does not break). The reflection network after the
first scattering is shown in Fig \ref{fig:toy example refnet}(b).

Repeating the same for the scattering of $\Pi_{Z;2}$ with $\Pi_{X;1}$
we get 
\begin{equation}
\Pi_{Z;2}\Pi_{X;1}\Pi_{Z;2}=\frac{1}{2}\ket 7\bra 7+\frac{1}{2}\ket{_{-}^{+56}}\bra{_{-}^{+56}}
\end{equation}
so $\Pi_{Z;2}=\Pi_{Z;2}^{\left(1/2\right)}+\Pi_{Z;2}^{\left(0\right)}$
and
\begin{align}
\Pi_{Z;2}^{\left(1/2\right)} & =\ket 7\bra 7+\ket{_{-}^{+56}}\bra{_{-}^{+56}}\\
\Pi_{Z;2}^{\left(0\right)} & =\ket 8\bra 8+\ket{_{-6}^{+5}}\bra{_{-6}^{+5}}.
\end{align}
As before, $\Pi_{X;1}$ does not break in the scattering.

The final proper reflection network is shown in Fig \ref{fig:toy example refnet}(c).

\subsection{Minimality and Completeness of Reflecting Projections \label{subsec:Minimality-and-Completnes}}

Now we will examine the properties of the set of projections that
comes out of Step 2 of the algorithm. As we discussed above, this step finishes when all projections are pairwise reflecting. In order
to construct the BPT, we will need at least one MSMP and any additional
minimal projections required to generate the whole algebra. Thus,
we will have to establish whether the final set of reflecting projections
meets the following criteria:
\begin{enumerate}
	\item \emph{Minimality}: All projections in the final set are minimal (Definition
	\ref{def:minimal projection}).
	\item \emph{Completeness}: The final set contains at least one MSMP (Definition
	\ref{def:MSMP}).
\end{enumerate}
We will now introduce correction procedures for when these criteria are not met. 

\subsubsection{Minimality}

Minimality of the reflecting projections can be established by considering
the paths in the reflection network. Each path is given by a sequence
of vertices $\boldsymbol{v}=\left(v_{1},v_{2},...,v_{n}\right)$ that
specify the projections along the path. By taking the product of all
projections along the path and normalizing with reflection coefficients\footnote{In practice we do not need to remember the reflection coefficients
	in order to construct these operators, since at each step the normalization is given
	by the non-zero singular value (which is unique, since all projections
	are reflecting) of $\Pi_{v_{1}}\Pi_{v_{2}}...\Pi_{v_{n}}$.}, we define the operator 
\begin{equation}
S_{\boldsymbol{v}}:=\frac{\Pi_{v_{1}}\Pi_{v_{2}}...\Pi_{v_{n}}}{\sqrt{\lambda_{v_{1}v_{2}}\lambda_{v_{2}v_{3}}...\lambda_{v_{n-1}v_{n}}}}.\label{eq:def of path isometry}
\end{equation}
Such an operator will be referred to as a \emph{path isometry} since it
is a partial isometry from the eigenspace of $\Pi_{v_{n}}$ to the
eigenspace of $\Pi_{v_{1}}$ along the path $\boldsymbol{v}$ (the
operator $S_{\boldsymbol{v}}^{\dagger}$ is a path isometry in the
opposite direction). In order to see that this is the case, consider
the path between two neighboring projections $\Pi_{v}$,$\Pi_{u}$
such that $S_{\left(v,u\right)}=$$\frac{1}{\sqrt{\lambda_{vu}}}$$\Pi_{v}\Pi_{u}$.
This is a partial isometry because
\begin{align}
S_{\left(v,u\right)}S_{\left(v,u\right)}^{\dagger} & =\frac{1}{\lambda_{vu}}\Pi_{v}\Pi_{u}\Pi_{v}=\Pi_{v}\\
S_{\left(v,u\right)}^{\dagger}S_{\left(v,u\right)} & =\frac{1}{\lambda_{vu}}\Pi_{u}\Pi_{v}\Pi_{u}=\Pi_{u}.
\end{align}
The general case follows in the same way by considering $S_{\boldsymbol{v}}S_{\boldsymbol{v}}^{\dagger}$,
$S_{\boldsymbol{v}}^{\dagger}S_{\boldsymbol{v}}$ and reducing the
products of projections by applying the reflection relations.

The minimality of reflecting projections can then be established with
the help of the following lemma
\begin{lem}
	\label{lem:Minimality of reflecting projections }Let $\left\{ \Pi_{v}\right\} $
	be a set of projections forming a proper reflection network and let
	$\left\{ S_{\boldsymbol{v}}\right\} $ be the set of all path isometries
	in the network as defined by Eq.\ (\ref{eq:def of path isometry}).
	Then, the following statements are equivalent:
	
	(1) Every $\Pi_{v}$ is a minimal projection in the algebra $\mathcal{A}:=\left\langle \left\{ \Pi_{v}\right\} \right\rangle $.
	
	(2) $S_{\boldsymbol{v}}\propto S_{\boldsymbol{u}}$ for all paths
	$\boldsymbol{v}$, $\boldsymbol{u}$ that share the same initial and
	final vertices.
\end{lem}

\begin{proof}
	Every element $M\in\mathcal{A}$ is a linear combination of products
	of $\left\{ \Pi_{v}\right\} $ so $\mathcal{A}=\mathsf{span}\left\{ S_{\boldsymbol{v}}\right\} $.
	Then, by Definition \ref{def:minimal projection} and linearity, the
	projections $\left\{ \Pi_{v}\right\} $ are minimal if and only if
	$\Pi_{v}S_{\boldsymbol{v}}\Pi_{v}\propto\Pi_{v}$ for all $v$ and
	$\boldsymbol{v}$. When $\Pi_{v}S_{\boldsymbol{v}}=0$ or $S_{\boldsymbol{v}}\Pi_{v}=0$,
	the relation $\Pi_{v}S_{\boldsymbol{v}}\Pi_{v}=0\propto\Pi_{v}$ holds
	trivially. Let us then consider $\Pi_{v}S_{\boldsymbol{v}}\Pi_{v}\neq0$
	for some $\boldsymbol{v}:=\left(v_{1},v_{2},...,v_{n}\right)$, which
	implies that $\boldsymbol{v}':=\left(v,v_{1},v_{2},...,v_{n},v\right)$
	is a circular path from $\Pi_{v}$ to itself. Recalling the definition
	in Eq.\ (\ref{eq:def of path isometry}) we can use both $\Pi_{v}S_{\boldsymbol{v}}\Pi_{v}\propto S_{\boldsymbol{v}'}$
	and $S_{\boldsymbol{v}'}\propto\Pi_{v}S_{\boldsymbol{v}}\Pi_{v}$,
	since the proportionality factor is not $0$. Thus, if statement 2
	holds then $S_{\boldsymbol{v}'}\propto S_{\left(v,v\right)}$ and
	\begin{equation}
	\Pi_{v}S_{\boldsymbol{v}}\Pi_{v}\propto S_{\boldsymbol{v}'}\propto S_{\left(v,v\right)}=\Pi_{v}.
	\end{equation}
	This proves 2$\Rightarrow$1.
	
	If statement 1 holds then $\Pi_{v}S_{\boldsymbol{v}}\Pi_{v}\propto\Pi_{v}$
	and 
	\begin{equation}
	S_{\boldsymbol{v}'}\propto\Pi_{v}S_{\boldsymbol{v}}\Pi_{v}\propto\Pi_{v}.\label{eq:proof of minimality 1}
	\end{equation}
	This proves statement 2 for all circular paths $v'$ since all $S_{\boldsymbol{v}'}$
	are (properly) proportional to the same initial projection $\Pi_{v}$
	and thus to each other. For non circular paths $\boldsymbol{v}=\left(v_{1},...,v_{n}\right)$,
	$\boldsymbol{u}=\left(u_{1},...,u_{m}\right)$ with $v_{1}=u_{1}$
	and $v_{n}=u_{m}$ let us assume that $S_{\boldsymbol{v}}\not\propto S_{\boldsymbol{u}}$
	so $S_{\boldsymbol{v}}S_{\boldsymbol{u}}^{\dagger}\not\propto S_{\boldsymbol{u}}S_{\boldsymbol{u}}^{\dagger}$.
	Then the path isometry $S_{\boldsymbol{v}'}=S_{\boldsymbol{v}}S_{\boldsymbol{u}}^{\dagger}$
	defined by the circular path $v':=\left(v_{1},...,v_{n}=u_{m},...,u_{1}\right)$
	is proportional to its initial projection $S_{\boldsymbol{v}}S_{\boldsymbol{u}}^{\dagger}\propto\Pi_{v_{1}}=\Pi_{u_{1}}=S_{\boldsymbol{u}}S_{\boldsymbol{u}}^{\dagger}$,
	in contradiction to $S_{\boldsymbol{v}}S_{\boldsymbol{u}}^{\dagger}\not\propto S_{\boldsymbol{u}}S_{\boldsymbol{u}}^{\dagger}$.
	Therefore, $S_{\boldsymbol{v}}\propto S_{\boldsymbol{u}}$ proving
	1$\Rightarrow$2.
\end{proof}
Thus, by checking whether the path isometries in a reflection network
depend only on the initial and final vertices and are independent
of the paths taken, we can verify that the projections are minimal.
In practice, it is not necessary to check all paths as there is usually
a lot of order in the reflection network and path independence can
be established based on this order. Things are even simpler when the projections
are of rank 1 (recall all projections in the same connected component of
a reflection network must have the same rank), then there is nothing
to check since rank 1 projections are always minimal.

In addition to providing a testable criterion for minimality, Lemma
\ref{lem:Minimality of reflecting projections } also implies a correction
for the case where the reflecting projections are not minimal.
\begin{prop}
	\label{prop:reduction of reflection network}In the setting of Lemma
	\ref{lem:Minimality of reflecting projections }, let $\boldsymbol{v}$,
	$\boldsymbol{u}$ be two paths that share the same initial $v_{1}=u_{1}$
	and final $v_{n}=u_{m}$ vertices but $S_{\boldsymbol{v}}\not\propto S_{\boldsymbol{u}}$.
	Then, the spectral projections $\left\{ \Pi^{\left(\omega\right)}\right\} $
	of $U:=S_{\boldsymbol{v}}S_{\boldsymbol{u}}^{\dagger}$ have the following
	properties:
	
	(1) Each $\Pi^{\left(\omega\right)}$ is in the algebra $\mathcal{A}:=\left\langle \left\{ \Pi_{v}\right\} \right\rangle $.
	
	(2) Each $\Pi^{\left(\omega\right)}$ is not reflecting with $\Pi_{v_{1}}$.
\end{prop}

\begin{proof}
	First note that the operator $U$ is a unitary on the eigenspace of
	$\Pi_{v_{1}}$ since $UU^{\dagger}=U^{\dagger}U=\Pi_{v_{1}}$. Since
	$U$ is in $\mathcal{A}$, and it is a normal operator, all its spectral
	projections are also in $\mathcal{A}$ \cite{farenick2012algebras}. This proves statement
	1.
	
	From $S_{\boldsymbol{v}}\not\propto S_{\boldsymbol{u}}$ we have $U=S_{\boldsymbol{v}}S_{\boldsymbol{u}}^{\dagger}\not\propto S_{\boldsymbol{u}}S_{\boldsymbol{u}}^{\dagger}=\Pi_{v_{1}}$.
	Since $U=\sum_{\omega}\omega\Pi^{\left(\omega\right)}$ and $\Pi_{v_{1}}=\sum_{\omega}\Pi^{\left(\omega\right)}$
	but $U\not\propto\Pi_{v_{1}}$ then there must be more then one spectral
	projection $\Pi^{\left(\omega\right)}$. Therefore, $\Pi_{v_{1}}\Pi^{\left(\omega\right)}\Pi_{v_{1}}=\Pi^{\left(\omega\right)}\not\propto\Pi_{v_{1}}$,
	proving statement 2.
\end{proof}
So, if the minimality condition of statement 2 in Lemma \ref{lem:Minimality of reflecting projections }
does not hold, we can take $U:=S_{\boldsymbol{v}}S_{\boldsymbol{u}}^{\dagger}$
for the two paths that violate it and use its spectral projections
$\left\{ \Pi^{\left(\omega\right)}\right\} $ to scatter $\Pi_{v_{1}}$.
By scattering the projections in the connected component of $\Pi_{v_{1}}$
with the spectral projections $\left\{ \Pi^{\left(\omega\right)}\right\} $
until everything is reflecting again, we will break down the connected
component into a reflection network of smaller rank projections. Then
we can check the condition of minimality again and repeat until it
is satisfied.

\subsubsection*{Procedure}

As was shown in Lemma \ref{lem:Minimality of reflecting projections },
in principle minimality of a reflection network can be established
by checking all path isometries connecting every pair of projections
and verifying that they are proportional to one another. We will formally
refer to this procedure as \textproc{EstablishMinimality}. This,
of course, is not a computationally tractable solution because of the exponentially
large number of paths in all but the most degenerate networks. Nonetheless,
when dealing with concrete examples, path invariance of path isometries
can be shown based on the specifics of the problem; this is what we
mean in practice when referring to the \textproc{EstablishMinimality}
procedure.

Even though we have not encountered non-minimal reflection networks
following the scattering procedure, we do know that such networks
exist.\footnote{Non-minimal reflection networks can be constructed directly by carefully
	choosing the reflecting projections. It is an open question whether there are conditions
	that guarantee that the reflecting projections that come out from
	the scattering procedure are always minimal.
	If that is not the case, a tractable procedure that establishes minimality
	of the reflection network without relying on the specifics of the
	problem would be desirable.} Therefore, for the sake of completeness, we have mentioned that even
in such cases there is a way to proceed, given by Proposition
\ref{prop:reduction of reflection network}.

\subsubsection*{Toy Example (continued)}

In the network shown in Fig \ref{fig:toy example refnet}(c) we have
three connected components but only the component $\left\{ \Pi_{Z;1}^{\left(1/2\right)},\Pi_{X;1},\Pi_{Z;2}^{\left(1/2\right)}\right\} $
has any paths. For every pair of projections in this component there
is only one simple path (that is a path that has no repeating vertices)
between them, so paths invariance trivially holds for simple paths.

Every non-simple path is of the form $\left(\left(\cdot\right),\Pi_{X;1},\left(\cdot\right),\Pi_{X;1},...,\left(\cdot\right),\Pi_{X;1},\left(\cdot\right),\Pi_{X;1},\left(\cdot\right)\right)$
where $\left(\cdot\right)$ is a placeholder for any of the other
two projections in the component, and it may be empty at the boundaries.
Any path isometry of such path is proportional to the path isometry
of the simple path $\left(\left(\cdot\right),\Pi_{X;1},\left(\cdot\right)\right)$,
because of the reflection relation $\Pi_{X;1}\left(\cdot\right)\Pi_{X;1}\propto\Pi_{X;1}$.
Therefore, in this network all path isometries between every pair
are proportional to one another.

\subsubsection{Completeness}

Completeness requires the reflection network to contain at least
one MSMP as defined in Definition \ref{def:MSMP}. That is,
assuming that the minimality of projections has been established,
we must identify a set of pairwise orthogonal projections that
sum to the identity of the algebra. Since the initial projections
in step 1 of the algorithm are the spectral projections of observables,
they must resolve the identity (otherwise the probabilities of outcomes
will not sum to $1$). The scattering at step 2 breaks them down
into smaller ranks, but they continue to resolve the identity. This means
that completeness is a given if the initial projections resolve the
identity to begin with. However, in more general applications
of the algorithm where we do not assume the inputs to consist of identity-resolving projections, it turns out that we can still reconstruct
an MSMP.

Given the reflection network of projections $\left\{ \Pi_{v}\right\} $
and the algebra $\mathcal{A}:=\left\langle \left\{ \Pi_{v}\right\} \right\rangle$, 
we will assume that all $\Pi_{v}$ are minimal in $\mathcal{A}$ (that is, minimality
has to be established before checking completeness). Consider the
largest subset of pairwise orthogonal projections $\left\{ \Pi_{v_{k}}\right\} \subseteq\left\{ \Pi_{v}\right\}$, with $\Pi_{v_{k}}\Pi_{v_{l}}=\delta_{kl}\Pi_{v_{k}}$, which is a
maximal independent set of vertices in the network (this set does not have
to be unique). The subset $\left\{ \Pi_{v_{k}}\right\} $ is an MSMP
if the operator
\begin{equation}
I_{\mathcal{A}}:=\sum_{k}\Pi_{v_{k}}\label{eq:def of identity in reflect. net.}
\end{equation}
is such that $I_{\mathcal{A}}\Pi_{v}=\Pi_{v}$ for all $v$. If it
is not, we can use the result of the following lemma to complete the subset into an MSMP.
\begin{lem}
	\label{lem:Completness}Let $\left\{ \Pi_{v_{k}}\right\} $ be the
	largest subset of pairwise orthogonal projections in the reflection
	network of $\left\{ \Pi_{v}\right\},$ where all $\Pi_{v}$ are minimal
	in the algebra $\mathcal{A}:=\left\langle \left\{ \Pi_{v}\right\} \right\rangle $.
	If there is a $v$ such that $I_{\mathcal{A}}\Pi_{v}\neq\Pi_{v}$,
	then, with the appropriate normalization factor $c$, the operator
	(here $I$ is the full identity matrix and $I_{\mathcal{A}}$ is given
	by Eq.\ (\ref{eq:def of identity in reflect. net.})) 
	\begin{equation}
	\tilde{\Pi}_{v}:=\frac{1}{c}\left(I-I_{\mathcal{A}}\right)\Pi_{v}\left(I-I_{\mathcal{A}}\right),\label{eq:def of complementary projection}
	\end{equation}
	has all of the following properties:
	
	(1) $\tilde{\Pi}_{v}$ is a minimal projection in $\mathcal{A}$.
	
	(2) $\tilde{\Pi}_{v}$ is orthogonal to all $\left\{ \Pi_{v_{k}}\right\} $.
	
	(3) The operator $\tilde{I}_{\mathcal{A}}:=I_{\mathcal{A}}+\tilde{\Pi}_{v}$
	is such that $\tilde{I}_{\mathcal{A}}\Pi_{v}=\Pi_{v}$.
\end{lem}

\begin{proof}
	If we distribute the terms in Eq.\ (\ref{eq:def of complementary projection})
	we will get $c\tilde{\Pi}_{v}=\Pi_{v}-I_{\mathcal{A}}\Pi_{v}-\Pi_{v}I_{\mathcal{A}}+I_{\mathcal{A}}\Pi_{v}I_{\mathcal{A}}$
	so $\tilde{\Pi}_{v}$ is an operator in $\mathcal{A}$. It is clearly
	self-adjoint and it squares to
	\begin{equation}
	\tilde{\Pi}_{v}\tilde{\Pi}_{v}=\frac{1}{c^{2}}\left(I-I_{\mathcal{A}}\right)\Pi_{v}\left(I-I_{\mathcal{A}}\right)\Pi_{v}\left(I-I_{\mathcal{A}}\right).\label{eq:completeness lemma 1}
	\end{equation}
	Since all $\Pi_{v}$ are minimal, $\Pi_{v}\left(I-I_{\mathcal{A}}\right)\Pi_{v}=\Pi_{v}-\Pi_{v}I_{\mathcal{A}}\Pi_{v}=\left(1-\alpha\right)\Pi_{v}$,
	where $\alpha$ is the proportionality factor in the minimality relation
	$\Pi_{v}I_{\mathcal{A}}\Pi_{v}\propto\Pi_{v}$, and $\alpha$ is not
	$1$ because that would contradict $I_{\mathcal{A}}\Pi_{v}\neq\Pi_{v}$.
	Thus, for $c=1-\alpha$, Eq.\ (\ref{eq:completeness lemma 1}) is equal
	to $\tilde{\Pi}_{v}$, so $\tilde{\Pi}_{v}$ is a projection. It is
	minimal because for any matrix $M\in\mathcal{A}$, 
	\begin{equation}
	\tilde{\Pi}_{v}M\tilde{\Pi}_{v}=\frac{1}{c^{2}}\left(I-I_{\mathcal{A}}\right)\Pi_{v}\tilde{M}\Pi_{v}\left(I-I_{\mathcal{A}}\right)
	\end{equation}
	where $\tilde{M}:=\left(I-I_{\mathcal{A}}\right)M\left(I-I_{\mathcal{A}}\right)$
	is also in $\mathcal{A}$, so $\Pi_{v}\tilde{M}\Pi_{v}\propto\Pi_{v}$
	and $\tilde{\Pi}_{v}M\tilde{\Pi}_{v}\propto\tilde{\Pi}_{v}$. This
	proves statement 1. Statement 2 follows from $\left(I-I_{\mathcal{A}}\right)\Pi_{v_{k}}=\Pi_{v_{k}}-\Pi_{v_{k}}=0$
	so $\tilde{\Pi}_{v}\Pi_{v_{k}}=0$. Lastly, using the minimality of
	$\Pi_{v}$ and $c=1-\alpha$ once again, we get 
	\begin{equation}
	\tilde{\Pi}_{v}\Pi_{v}=\frac{1}{c}\left(I-I_{\mathcal{A}}\right)\Pi_{v}\left(I-I_{\mathcal{A}}\right)\Pi_{v}=\left(I-I_{\mathcal{A}}\right)\Pi_{v}
	\end{equation}
	so $\tilde{I}_{\mathcal{A}}\Pi_{v}=I_{\mathcal{A}}\Pi_{v}+\left(I-I_{\mathcal{A}}\right)\Pi_{v}=\Pi_{v}$.
	This proves statement 3.
\end{proof}

\subsubsection*{Procedure}

The procedure to establish completeness is only necessary if the initial
projections are not known to resolve the identity.

\noindent 
\begin{algorithmic}[1] 
	
	\Procedure{EstablishCompleteness}{$ReflectNet$}
	\State $MSMP \gets$ \Call{PickMaxIndependentSet}{$ReflectNet$}
	
	\ForAll{$\Pi \in$ ($ReflectNet$ excluding MSMP)}  
	\If{\Call{SumAll}{$MSMP$}$\Pi\neq\Pi$}
	\State $\tilde{\Pi} \gets$ \Call{ConstructComplementaryProj}{$MSMP, \Pi$}
	\State $MSMP \gets$ \Call{AddProj}{$MSMP, \tilde{\Pi}$}
	\State $ReflectNet \gets$ \Call{AddProj}{$ReflectNet, \tilde{\Pi}$}
	\EndIf
	\EndFor
	
	\State \Return $ReflectNet$ 
	\EndProcedure
	
\end{algorithmic}

\noindent Completeness is achieved by choosing a maximal independent
set of orthogonal projections in the network $\left\{ \Pi_{v_{k}}\right\} $,\footnote{\noindent Actually, any subset of pairwise orthogonal projections
	will do, but a maximal independent set is what we end up constructing
	anyway.} and testing whether $I_{\mathcal{A}}=\sum_{k}\Pi_{v_{k}}$ acts as
the identity on all projections in the network. If it does, then $\left\{ \Pi_{v_{k}}\right\} $
is the MSMP. If it does not, then for each projection such that $I_{\mathcal{A}}\Pi_{v}\neq\Pi_{v}$
we construct the complementary projection $\tilde{\Pi}_{v}$ as defined
in Eq.\ (\ref{eq:def of complementary projection}) and add it to the
network\footnote{\noindent When adding a new minimal projection to the network we need
	to establish its reflection relations with all existing elements.
	The minimality of projections ensures that it will not trigger new
	scattering and breakdowns of , but we do need to know which
	existing projections are orthogonal to the new element and which are
	not.} and the independent set of orthogonal projections $\left\{ \Pi_{v_{k}}\right\} $.
Lemma \ref{lem:Completness} ensures that the final set $\left\{ \Pi_{v_{k}}\right\} $
always consists of pairwise orthogonal minimal projections in the
algebra that sum to the identity, i.e. an MSMP.

\subsubsection*{Toy Example (continued)}

\begin{figure}[t]
	\fbox{\begin{minipage}[b]{0.48\columnwidth}%
			\begin{flushleft}
				(a)\qquad{}
				\par\end{flushleft}
			\begin{center}
				\includegraphics[width=0.95\columnwidth]{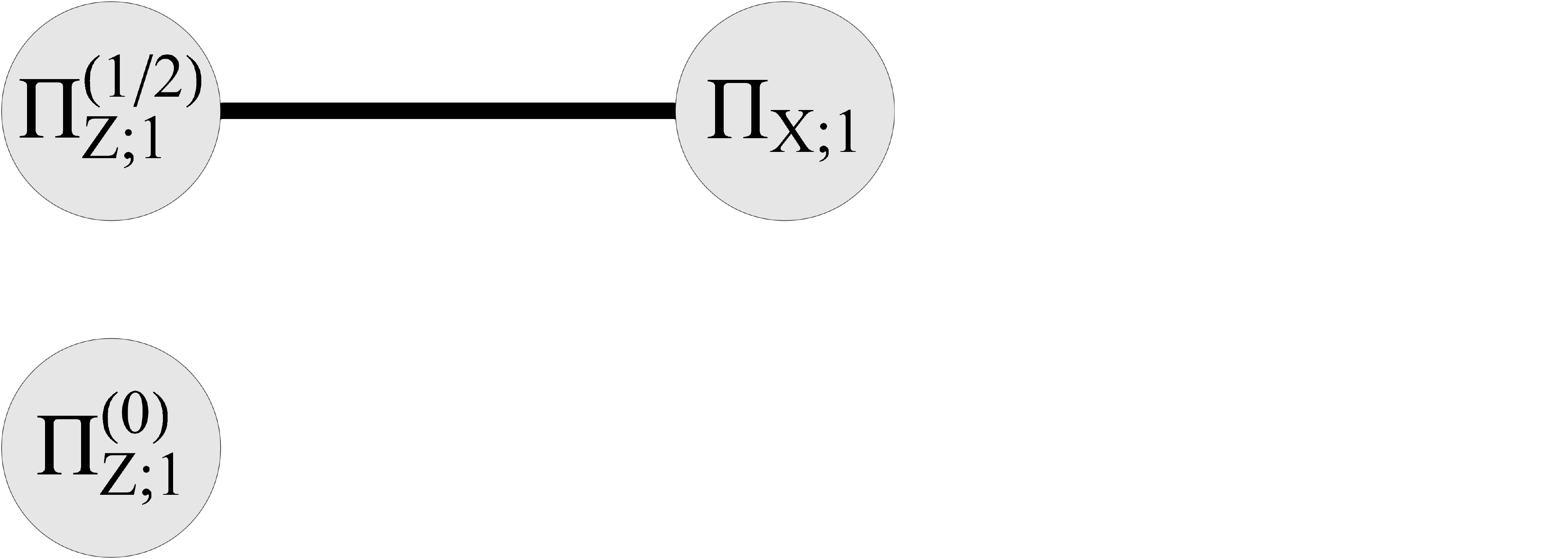}\\
				~\\
				\par\end{center}%
	\end{minipage}} %
	\fbox{\begin{minipage}[b]{0.48\columnwidth}%
			\begin{flushleft}
				(b)\qquad{}
				\par\end{flushleft}
			\begin{center}
				\includegraphics[width=0.95\columnwidth]{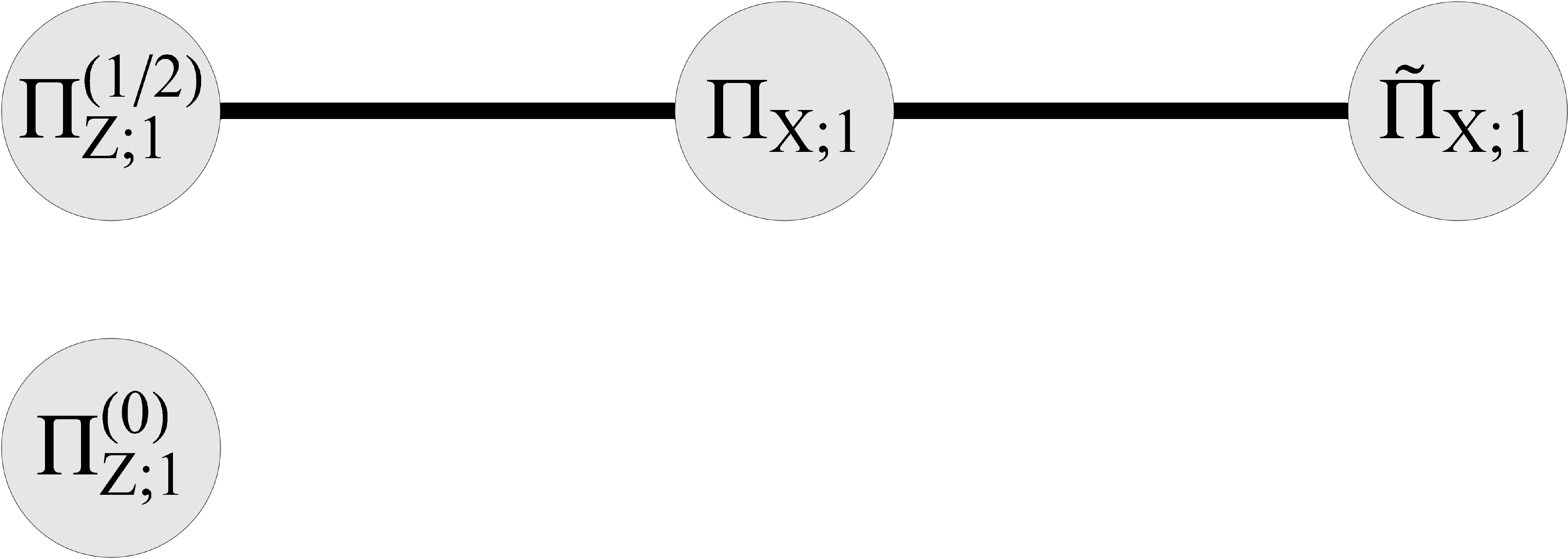}\\
				~\\
				\par\end{center}%
	\end{minipage}}\caption{\label{fig:toy model refnet completeness}Completion of a reflection
		network lacking an MSMP (a), to the one that has an MSMP (b).}
\end{figure}

In our example, the initial projections resolve the identity $\Pi_{Z;1}+\Pi_{Z;2}=I$,
so, as expected, the maximal independent set in the reflection network
of Fig. \ref{fig:toy example refnet}(c) is the MSMP since $\Pi_{Z;1}^{\left(1/2\right)}+\Pi_{Z;1}^{\left(0\right)}+\Pi_{Z;2}^{\left(1/2\right)}+\Pi_{Z;2}^{\left(0\right)}=I$.

In order to demonstrate how the MSMP can be constructed even if it
initially is missing, we will drop $\Pi_{Z;2}$ and consider the algebra
generated by $\left\langle \Pi_{Z;1},\Pi_{X;1}\right\rangle $. As
we calculated before, after scattering we have a proper reflection
network consisting of the reflecting pair $\left\{ \Pi_{Z;1}^{\left(1/2\right)},\Pi_{X;1}\right\} $
and the projection $\Pi_{Z;1}^{\left(0\right)}$ orthogonal
to both; see Fig \ref{fig:toy model refnet completeness}(a). The
maximal independent set consists of $\left\{ \Pi_{Z;1}^{\left(1/2\right)},\Pi_{Z;1}^{\left(0\right)}\right\} $
but $I_{\mathcal{A}}:=\Pi_{Z;1}^{\left(1/2\right)}+\Pi_{Z;1}^{\left(0\right)}=\Pi_{Z;1}$
does not act as the identity on $\Pi_{X;1}$. Using Eq.\ (\ref{eq:def of complementary projection})
we construct the complementary projection in this algebra
\begin{equation}
\tilde{\Pi}_{X;1}:=\frac{1}{c}\left(I-\Pi_{Z;1}\right)\Pi_{X;1}\left(I-\Pi_{Z;1}\right)=\frac{1}{c}\Pi_{Z;2}\Pi_{X;1}\Pi_{Z;2}=\frac{1}{c}\left(\frac{1}{2}\ket 7\bra 7+\frac{1}{2}\ket{_{-}^{+56}}\bra{_{-}^{+56}}\right),
\end{equation}
choosing $c=\frac{1}{2}$ for proper normalization, and add it to
the network. The new projection $\tilde{\Pi}_{X;1}$ is orthogonal
to both $\left\{ \Pi_{Z;1}^{\left(1/2\right)},\Pi_{Z;1}^{\left(0\right)}\right\} $,
and is reflecting with $\Pi_{X;1}$, which results in the reflection
network shown in Fig \ref{fig:toy model refnet completeness}(b).
Now the maximal independent set sums to $I_{\mathcal{A}}:=\Pi_{Z;1}^{\left(1/2\right)}+\Pi_{Z;1}^{\left(0\right)}+\tilde{\Pi}_{X;1}$
and we can check that it acts as the identity on $\Pi_{X;1}$, so $\left\{ \Pi_{Z;1}^{\left(1/2\right)},\Pi_{Z;1}^{\left(0\right)},\tilde{\Pi}_{X;1}\right\} $
is our MSMP.

\subsection{Construction of Bipartition Tables from Minimal Projections\label{subsec:Construction-of-Bipartition}}

As discussed at the start of the section, the structure captured by a reflection network that meets the criteria
of minimality and completeness can be translated into a bipartition
table in the following way: The elements of an MSMP correspond to columns
of the BPT. The isometries between the columns are given by the path isometries
between the elements of the MSMP (minimality ensures that the particular choice
of path is inconsequential). Elements of the MSMP that are not connected
by any path in the network are not related by an isometry, so
cannot be in the same block of the BPT. That is, distinct connected
components of the reflection network correspond to distinct blocks
of the BPT.

The formal construction of the BPT relies on the proof of the following
lemma:
\begin{lem}
	\label{lem:BPT construction from refnet}Let $\left\{ \Pi_{v}\right\} $
	be projections of a reflection network for which minimality and completeness
	holds, and let $\left\{ \Pi_{v_{k}}\right\} \subseteq\left\{ \Pi_{v}\right\} $
	be an MSMP. Then, there is a BPT with the BPOs $\left\{ S_{kl}^{q}\right\} $
	such that every $S_{kl}^{q}$ is a path isometry in the network and
	the set $\left\{ S_{kl}^{q}\right\} $ spans the algebra $\mathcal{A}:=\left\langle \left\{ \Pi_{v}\right\} \right\rangle $.
\end{lem}

\begin{proof}
	In order to construct the aforementioned BPT, we will first select
	a subset of path isometries in the network to be the BPOs.\footnote{It should be noted that the selection of BPOs is not unique and depends
		on the arbitrary selection of paths between the elements of the MSMP.
		This freedom, however, only changes the individual BPOs by a constant
		factor, which does not affect the generalized bipartition structure
		captured by the BPT.} Let $\left\{ \Pi_{v_{k}^{q}}\right\} $ be all the elements of the MSMP
	that belong to the connected component $q$ and let $\Pi_{v_{1}^{q}}$
	be a single, arbitrarily chosen, element. We first select the path isometries
	$\left\{ S_{k1}^{q}\right\} $ by arbitrarily choosing a path from
	$\Pi_{v_{1}^{q}}$ to each $\Pi_{v_{k}^{q}}$ for $k>1$ and $S_{11}^{q}:=\Pi_{v_{1}^{q}}$.
	We then define the BPOs for all $k,l\geq1$ to be $S_{kl}^{q}:=S_{k1}^{q}S_{l1}^{q\dagger}$
	which are just path isometries from $\Pi_{v_{l}^{q}}$ to $\Pi_{v_{k}^{q}}$
	that go through $\Pi_{v_{1}^{q}}$.
	
	For each connected component $q$, we now construct the corresponding block of
	the BPT. First, we choose an orthonormal basis $\left\{ \ket{e_{i1}^{q}}\right\} _{i=1..r_{q}}$
	for the eigenspace of $\Pi_{v_{1}^{q}}$, where $r_{q}$ is the rank
	of projections in the $q$the component. Then, we populate the first column of the
	block with $\ket{e_{i1}^{q}}$, such that $i$ is the row index, and
	each subsequent column $k>1$ is populated by the basis $\ket{e_{ik}^{q}}:=S_{k1}^{q}\ket{e_{i1}^{q}}$.
	As a result,
	\begin{equation}
	S_{kl}^{q}=S_{k1}^{q}S_{l1}^{q\dagger}=S_{k1}^{q}\Pi_{v_{1}^{q}}S_{l1}^{q\dagger}=\sum_{i=1..r_{q}}S_{k1}^{q}\ket{e_{i1}^{q}}\bra{e_{i1}^{q}}S_{l1}^{q\dagger}=\sum_{i=1..r_{q}}\ket{e_{ik}^{q}}\bra{e_{il}^{q}},
	\end{equation}
	so $\left\{ S_{kl}^{q}\right\} $ are indeed the BPOs of this block of the BPT.
	
	Since $\mathcal{A}$ is spanned by products of $\left\{ \Pi_{v}\right\} $
	which are proportional to path isometries $\left\{ S_{\boldsymbol{v}}\right\} $,
	it suffices to show that every $S_{\boldsymbol{v}}$ is spanned
	by $\left\{ S_{kl}^{q}\right\} $ in order to show that $\left\{ S_{kl}^{q}\right\} $
	spans $\mathcal{A}$. If $\left\{ \Pi_{v_{k}^{q}}\right\} $ is a
	MSMP, then by definition $I_{\mathcal{A}}=\sum_{q,k}\Pi_{v_{k}^{q}}$
	is the identity of the algebra and
	\begin{equation}
	S_{\boldsymbol{v}}=I_{\mathcal{A}}S_{\boldsymbol{v}}I_{\mathcal{A}}=\sum_{kl}\Pi_{v_{k}^{q}}S_{\boldsymbol{v}}\Pi_{v_{l}^{q}},\label{eq:BPT const from refnet eq 1}
	\end{equation}
	where $q$ is the connected component that contains the path $\boldsymbol{v}$.
	Every non-vanishing term $\Pi_{v_{k}^{q}}S_{\boldsymbol{v}}\Pi_{v_{l}^{q}}$
	is proportional to the path isometry $S_{\left(v_{k}^{q},\boldsymbol{v},v_{l}^{q}\right)}$
	from $\Pi_{v_{l}^{q}}$ to $\Pi_{v_{k}^{q}}$ along the path $\boldsymbol{v}$.
	Furthermore, if minimality holds, then according to Lemma \ref{lem:Minimality of reflecting projections }
	path isometries are path-independent, so $S_{\left(v_{k}^{q},\boldsymbol{v},v_{l}^{q}\right)}\propto S_{kl}^{q}$.
	Therefore, either $\Pi_{v_{k}^{q}}S_{\boldsymbol{v}}\Pi_{v_{l}^{q}}=0$
	or $\Pi_{v_{k}^{q}}S_{\boldsymbol{v}}\Pi_{v_{l}^{q}}\propto S_{kl}^{q},$
	so Eq.\ (\ref{eq:BPT const from refnet eq 1}) implies that $S_{\boldsymbol{v}}$
	is in the span of $\left\{ S_{kl}^{q}\right\} $.
\end{proof}
The practical takeaway from this lemma is that in order to construct
the BPT of an algebra generated by a reflection network we need to
(arbitrarily) pick a basis $\left\{ \ket{e_{i1}^{q}}\right\} _{i=1..r_{q}}$
for the eigenspace of a single MSMP element $\Pi_{v_{1}^{q}}$ in
each connected component $q$, and map those basis elements to the
eigenspaces of the rest of MSMP $\left\{ \Pi_{v_{k}^{q}}\right\} $
in $q$ using (arbitrarily chosen) path isometries $\left\{ S_{k1}^{q}\right\} $.
The resulting set $\left\{ \ket{e_{ik}^{q}}\right\} $ are the basis
elements that reside in block $q$, row $i$, column $k$ of the BPT.

\subsubsection*{Procedure}

The procedure for constructing the BPT is essentially what we did
in the proof of Lemma \ref{lem:BPT construction from refnet}

\begin{algorithmic}[1] 
	\Procedure{ConstructIrrepBasis}{$ReflectNet$}
	\State $BPT \gets \left\{ \right\} $
	
	\ForAll{$ConnComp \subseteq ReflectNet$} 
	\State $Block \gets \left\{ \right\} $
	\State $MaxIndepSet \gets$\Call{PickMaxIndependentSet}{$ConnComp$} 
	\State $\Pi_1 \gets$\Call{PickAnyElement}{$MaxIndepSet$} 
	\State $FirstColumnBasis \gets$  \Call{ConstructEigenBasis}{$\Pi_1$}
	\State $Block \gets$ \Call{AddColumn}{$Block, FirstColumnBasis$} 
	
	\ForAll{$\Pi_{k\neq 1} \in MaxIndepSet$} 
	\State $S_{k1} \gets$\Call{ConstructPathIsometry}{$ConnComp, \Pi_1, \Pi_k$} 
	\State $NewColumnBasis \gets$ \Call{MapBasis}{$S_{k1}, FirstColumnBasis$} 
	\State $Block \gets$ \Call{AddColumn}{$Block, NewColumnBasis$} 
	\EndFor
	
	\State $BPT \gets$ \Call{AddBlock}{$BPT, Block$}
	
	\EndFor
	
	\State \Return  $BPT$
	\EndProcedure
\end{algorithmic}

\noindent For each connected component, the procedure chooses a maximal
independent set of orthogonal projections, which is the subset of
the MSMP in the component, and uses it to construct the columns of
a single block of the BPT. In order to construct the block, it arbitrarily
picks a single projection $\Pi_{1}$ in the MSMP and arbitrarily constructs
the basis that span its eigenspace; these basis become the first column
of the block. The rest of the columns are constructed by picking each
of $\Pi_{k\neq1}$ in the MSMP and constructing a path isometry $S_{k1}$
from the eigenspace of $\Pi_{1}$ to $\Pi_{k}$. The path isometry
$S_{k1}$ is then used to map the elements of the\emph{ }first\emph{
}column to the elements of the $k$th column\emph{. }Once each block
is constructed, it is added to the BPT.

\subsubsection*{Toy Example (continued)}

In the reflection network of Fig. \ref{fig:toy example refnet}(c)
there are three connected components, but two of them, $\left\{ \Pi_{Z;1}^{\left(0\right)}\right\} $
and $\left\{ \Pi_{Z;2}^{\left(0\right)}\right\} $, consist of a single
projection which correspond to blocks with a single column. Arbitrarily choosing to use
the same basis as we have used before,
these single column blocks are %
\begin{tabular}{|c|}
	\hline 
	$4$\tabularnewline
	\hline 
	$_{-2}^{+1}$\tabularnewline
	\hline 
\end{tabular} and %
\begin{tabular}{|c|}
	\hline 
	$8$\tabularnewline
	\hline 
	$_{-6}^{+5}$\tabularnewline
	\hline 
\end{tabular}.

For the remaining block, we identify $\left\{ \Pi_{Z;1}^{\left(1/2\right)},\Pi_{Z;2}^{\left(1/2\right)}\right\} $
to be the block's maximal independent set. We pick $\Pi_{Z;1}^{\left(1/2\right)}$
to be the projection associated with the first column and we pick
its basis to be %
\begin{tabular}{|c|}
	\hline 
	$3$\tabularnewline
	\hline 
	$_{-}^{+12}$\tabularnewline
	\hline 
\end{tabular}. The path isometry that maps the first column to the second column
associated with $\Pi_{Z;2}^{\left(1/2\right)}$ is constructed by
taking the only simple path between them (we have fixed the normalization
after the fact)
\begin{equation}
S_{21}\propto\Pi_{Z;2}^{\left(1/2\right)}\Pi_{X;1}\Pi_{Z;1}^{\left(1/2\right)}=\frac{1}{2}\ket 7\bra 3+\frac{1}{2}\ket{_{-}^{+56}}\bra{_{-}^{+12}}.
\end{equation}
Then, by mapping the first column using this isometry, we get the
second column %
\begin{tabular}{|c|}
	\hline 
	$7$\tabularnewline
	\hline 
	$_{-}^{+56}$\tabularnewline
	\hline 
\end{tabular}.\footnote{Although we already identified this basis when we first wrote
	the projection $\Pi_{Z;2}^{\left(1/2\right)}$, we could not know,
	\emph{a priori}, how the path isometry would map the eigenbasis between
	projections. It is only due to the simplicity of this toy example
	that the basis we used to express the projections after
	scattering ended up as the basis in the BPT.} Combining all the columns into blocks completes the construction
of the BPT

\noindent %
\noindent\begin{minipage}[c]{1\columnwidth}%
	\begin{tabular}{>{\centering}m{0.9\columnwidth}>{\raggedright}m{0.1\columnwidth}}
		& \tabularnewline
		\centering{}%
		\begin{tabular}{c|c|cc}
			\cline{1-1} 
			\multicolumn{1}{|c|}{$4$} & \multicolumn{1}{c}{} &  & \tabularnewline
			\cline{1-1} 
			\multicolumn{1}{|c|}{$_{-2}^{+1}$} & \multicolumn{1}{c}{} &  & \tabularnewline
			\cline{1-2} \cline{2-2} 
			& $8$ &  & \tabularnewline
			\cline{2-2} 
			& $_{-6}^{+5}$ &  & \tabularnewline
			\cline{2-4} \cline{3-4} \cline{4-4} 
			\multicolumn{1}{c}{} &  & \multicolumn{1}{c|}{$3$} & \multicolumn{1}{c|}{$7$}\tabularnewline
			\cline{3-4} \cline{4-4} 
			\multicolumn{1}{c}{} &  & \multicolumn{1}{c|}{$_{-}^{+12}$} & \multicolumn{1}{c|}{$_{-}^{+56}$}\tabularnewline
			\cline{3-4} \cline{4-4} 
		\end{tabular} & %
		\begin{minipage}[c]{0.04\columnwidth}%
			~%
		\end{minipage}\tabularnewline
		& \tabularnewline
	\end{tabular}%
\end{minipage}

The above BPT tells us that the Hilbert space decomposes into irreps
as
\begin{equation}
\mathcal{H}=\mathcal{H}_{A_{1}}\oplus\mathcal{H}_{A_{2}}\oplus\mathcal{H}_{A_{3}}\otimes\mathcal{H}_{B_{3}},\label{eq:toy model irrep dec}
\end{equation}
where $A_{q}$ are the subsystems associated with the multiplicity
of irreps and $B_{q}$ are the subsystems on which the algebra acts
irreducibly. In this case, blocks $1$ and $2$ (the single column
blocks) specify one-dimensional irreps, so the one-dimensional subsystems
$B_{q=1,2}$ are absorbed into the two-dimensional multiplicities
$A_{q=1,2}$. The last block specifies a two dimensional irrep $B_{3}$,
with a two-dimensional multiplicity $A_{3}$.

According to Theorem \ref{thm:Wedderburn decomposition}, with respect
to this irrep decomposition, all operators in the algebra are of the
form
\begin{equation}
M=c_{1}I_{A_{1}}+c_{2}I_{A_{2}}+I_{A_{3}}\otimes M_{B_{3}}
\end{equation}
for any scalars $c_{1}$, $c_{2}$ and $2\times2$ matrices $M_{B_{3}}$.
In particular, the generators $\left\{ Z,X\right\} $ can also be
presented in this form. To see this explicitly, we change the original basis
into the irrep basis given by the BPT (reading the BPT from left to
right, top to bottom)
\begin{equation}
\left\{ \ket 1,\ket 2,\ket 3,\ket 4,\ket 5,\ket 6,\ket 7,\ket 8\right\} \longmapsto\left\{ \ket 4,\ket{_{-2}^{+1}},\ket 8,\ket{_{-6}^{+5}},\ket 3,\ket 7,\ket{_{-}^{+12}},\ket{_{-}^{+56}}\right\} .
\end{equation}
Assuming $Z=a\Pi_{Z;1}+b\Pi_{Z;2}$ and $X=c\Pi_{X;1}+d\Pi_{X;2}$,
for some eigenvalues $a,b,c,d$, we re-express their matrices using
the irrep basis, thus simultaneously block-diagonalizing both $Z$
and $X$ 
\begin{align}
Z= & \begin{pmatrix}a\\
& a\\
&  & a\\
&  &  & a\\
&  &  &  & b\\
&  &  &  &  & b\\
&  &  &  &  &  & b\\
&  &  &  &  &  &  & b
\end{pmatrix} & \longmapsto & \begin{pmatrix}a\\
& a\\
&  & b\\
&  &  & b\\
&  &  &  & a\\
&  &  &  &  & b\\
&  &  &  &  &  & a\\
&  &  &  &  &  &  & b
\end{pmatrix}\\
X= & \left(\begin{array}{cccccccc}
\frac{c+3d}{4} & \frac{c-d}{4} &  &  & \frac{c-d}{4} & \frac{c-d}{4}\\
\frac{c-d}{4} & \frac{c+3d}{4} &  &  & \frac{c-d}{4} & \frac{c-d}{4}\\
&  & \frac{c+d}{2} &  &  &  & \frac{c-d}{2}\\
&  &  & d\\
\frac{c-d}{4} & \frac{c-d}{4} &  &  & \frac{c+3d}{4} & \frac{c-d}{4}\\
\frac{c-d}{4} & \frac{c-d}{4} &  &  & \frac{c-d}{4} & \frac{c+3d}{4}\\
&  & \frac{c-d}{2} &  &  &  & \frac{c+d}{2}\\
&  &  &  &  &  &  & d
\end{array}\right) & \longmapsto & \left(\begin{array}{cccccccc}
d\\
& d\\
&  & d\\
&  &  & d\\
&  &  &  & \frac{c+d}{2} & \frac{c-d}{2}\\
&  &  &  & \frac{c-d}{2} & \frac{c+d}{2}\\
&  &  &  &  &  & \frac{c+d}{2} & \frac{c-d}{2}\\
&  &  &  &  &  & \frac{c-d}{2} & \frac{c+d}{2}
\end{array}\right).
\end{align}

Now the irrep decomposition (\ref{eq:toy model irrep dec}) of the
algebra generated by $\left\{ Z,X\right\} $ can be easily observed
from the block-diagonal form of the generators.

\subsection{The Algorithm and Its Proof \label{subsec:The-Algorithm}}

Following the results of the previous subsections, we are now in the position
to formally present and prove the correctness of the irrep decomposition
algorithm. The input of this algorithm is a finite set of self-adjoint
matrices\footnote{It is not necessary to assume self-adjoint generators, but it is convenient.
	We can always express non-self-adjoint generators as sums of self-adjoint
	matrices in the algebra, so this assumption is not restrictive.} $\mathcal{M}\subseteq\mathcal{L}\left(\mathcal{H}\right)$ that generate
the algebra $\mathcal{A}:=\left\langle \mathcal{M}\right\rangle $.
The output is a set of basis elements $\left\{ \ket{e_{ik}^{q}}\right\} $
labeled by the irreps of $\mathcal{A}$ where $q$ specifies the distinct
irreps, $i$ specifies the multiple instances of identical irreps,
and $k$ specifies the distinct basis elements inside each irrep.
The irrep decomposition 
\begin{equation}
\mathcal{H}\cong\bigoplus_{q}\mathcal{H}_{A_{q}}\otimes\mathcal{H}_{B_{q}}
\end{equation}
is then given by reinterpreting the basis $\ket{e_{ik}^{q}}$ as the
product basis $\ket{a_{i}^{q}}\ket{b_{k}^{q}}$ of $\mathcal{H}_{A_{q}}\otimes\mathcal{H}_{B_{q}}$
(formally, we will define the isometry $V:=\sum_{q,i,k}\ket{a_{i}^{q}}\ket{b_{k}^{q}}\bra{e_{ik}^{q}}$
that maps $\mathcal{H}$ into $\bigoplus_{q}\mathcal{H}_{A_{q}}\otimes\mathcal{H}_{B_{q}}$
and thus specifies the decomposition).

With the procedures defined in the previous subsections, the top-level
procedure of the algorithm is as follows:

\noindent 
\begin{algorithm}[H]
	
	\begin{algorithmic}[1] 
		
		\Procedure{IrrepDecomposition}{$\mathcal{M}$}     
		
		\State $SpecProjs \gets$  \Call{GetAllSpectralProjections}{$\mathcal{M}$}
		
		\State $ReflectNet \gets$ \Call{ScatterAllProjections}{$SpecProjs$}
		
		\State $ReflectNet \gets$ \Call{EstablishMinimality}{$ReflectNet$}
		
		\State $ReflectNet \gets$ \Call{EstablishCompleteness}{$ReflectNet$}
		
		\State $BPT \gets$ \Call{ConstructIrrepBasis}{$ReflectNet$}
		
		\State \Return $BPT$
		\EndProcedure
		
	\end{algorithmic}
	
	\caption{\label{alg:Irrepd-decomposition}Irrep decomposition of matrix algebra}
\end{algorithm}

\noindent The algorithm returns a BPT since this is the natural data
structure to organize the irrep basis.

The correctness of this algorithm follows from the proof of the following
theorem:
\begin{thm}
	\label{thm:Algorithm correctness}Let $\mathcal{H}$ be a finite-dimensional
	Hilbert space, $\mathcal{M}\subseteq\mathcal{L}\left(\mathcal{H}\right)$
	a finite set of self-adjoint matrices, and $\mathcal{\mathcal{A}}:=\left\langle \mathcal{M}\right\rangle $
	the matrix algebra generated by $\mathcal{M}$. Then, Algorithm \ref{alg:Irrepd-decomposition}
	produces the basis $\left\{ \ket{e_{ik}^{q}}\right\} $ of $\mathcal{H}$
	such that the isometry $V:=\sum_{q,i,k}\ket{a_{i}^{q}}\ket{b_{k}^{q}}\bra{e_{ik}^{q}}$
	explicitly specifies the irrep decomposition
	\begin{equation}
	\mathcal{H}\cong\bigoplus_{q}\mathcal{H}_{A_{q}}\otimes\mathcal{H}_{B_{q}}\label{eq:Alg correctness Hilb space decomp}
	\end{equation}
	defined in Theorem \ref{thm:Wedderburn decomposition} for the algebra
	$\mathcal{\mathcal{A}}$. \label{full theorem}
\end{thm}

\begin{proof}
	From Proposition \ref{prop:MA of single operator } we know that for
	each $M\in\mathcal{M}$, the spectral projections of $M$ generate
	the algebra $\left\langle M\right\rangle $ so the set of projections
	produced by \textproc{GetAllSpectralProjections} generates $\mathcal{A}$.
	
	During the procedure \textproc{ScatterAllProjections}, we break down
	pairs of projections $\left\{ \Pi_{1},\Pi_{2}\right\} $ by the scattering
	operation as defined in Definition \ref{def:Scattering }. The resulting
	projections $\left\{ \Pi_{1}^{\left(\lambda\right)},\Pi_{2}^{\left(\lambda\right)}\right\} $
	are in the algebra generated by $\left\{ \Pi_{1},\Pi_{2}\right\} $
	because -- again using Proposition \ref{prop:MA of single operator }
	-- they are the spectral projections of the operators $\Pi_{i}\Pi_{j}\Pi_{i}$
	that are in the algebra generated by $\left\{ \Pi_{1},\Pi_{2}\right\} $
	(for null elements $\Pi_{i}^{\left(0\right)}$ it is true by the definition of
	Eq.\ (\ref{eq:def of null projections in scattering})). Conversely,
	the predecessor projections $\left\{ \Pi_{1},\Pi_{2}\right\} $ are
	in the algebra generated by $\left\{ \Pi_{1}^{\left(\lambda\right)},\Pi_{2}^{\left(\lambda\right)}\right\} $
	because each $\left\{ \Pi_{i}^{\left(\lambda\right)}\right\} $ sums
	to $\Pi_{i}$. Therefore, the elements $\left\{ \Pi_{1},\Pi_{2}\right\} $
	before, and the elements $\left\{ \Pi_{1}^{\left(\lambda\right)},\Pi_{2}^{\left(\lambda\right)}\right\} $
	after scattering, generate the same algebra. So, after every iteration
	of scattering, the resulting reflection network generates the same
	algebra $\mathcal{A}$ as before.
	
	The procedure \textproc{ScatterAllProjections} keeps track of the
	known and unknown reflection relations as the network evolves, so
	when it stops, all pairs of projections must be reflecting. This procedure
	is guaranteed to stop because, according to Lemma \ref{lem: scattering of projections },
	every scattering iteration either identifies a previously unknown
	reflecting pair, or, the pair scatters into a series of reflecting
	pairs of lower rank. Eventually, all projections will either be reflecting
	or they will be reduced to rank 1 and then they must again be reflecting.
	Therefore, the procedure \textproc{ScatterAllProjections} produces,
	in a finite number of iterations, a proper reflection network that
	consists of projections that generate $\mathcal{A}$.
	
	The procedure \textproc{EstablishMinimality} establishes that all
	elements of a proper reflection network are minimal projections of
	the algebra $\mathcal{A}$ by checking the condition of minimality
	given by Lemma \ref{lem:Minimality of reflecting projections }.
	
	The procedure \textproc{EstablishCompleteness} completes the maximal
	orthogonal set of projections in the reflection network to an MSMP
	as prescribed by Lemma \ref{lem:Completness}.
	
	At this point, we have a reflection network that generates $\mathcal{A}$
	and is known to be minimal and complete, so the conditions of Lemma \ref{lem:BPT construction from refnet}
	hold. The procedure \textproc{ConstructIrrepBasis} constructs the
	basis $\left\{ \ket{e_{ik}^{q}}\right\} $ according to the procedure
	described in the proof of Lemma \ref{lem:BPT construction from refnet},
	so the partial isometries 
	\begin{equation}
	S_{kl}^{q}:=\sum_{i}\ket{e_{ik}^{q}}\bra{e_{il}^{q}}
	\end{equation}
	are the path isometries given by that Lemma that span the algebra $\mathcal{A}$.
	
	With respect to the decomposition (\ref{eq:Alg correctness Hilb space decomp})
	specified by the isometry $V:=\sum_{q,i,k}\ket{a_{i}^{q}}\ket{b_{k}^{q}}\bra{e_{ik}^{q}}$,
	the operators $S_{kl}^{q}$ take the form
	\begin{equation}
	VS_{kl}^{q}V^{\dagger}=\sum_{i}\ket{a_{i}^{q}}\ket{b_{k}^{q}}\bra{a_{i}^{q}}\bra{b_{l}^{q}}=I_{A_{q}}\otimes\ket{b_{k}^{q}}\bra{b_{l}^{q}},
	\end{equation}
	so they span all matrices of the form 
	\begin{equation}
	M=\bigoplus_{q}I_{A_{q}}\otimes M_{B_{q}}.\label{eq:Alg correctness matrix decompos}
	\end{equation}
	Therefore, with respect to the decomposition (\ref{eq:Alg correctness Hilb space decomp}),
	the algebra 
	\begin{equation}
	\mathcal{A}=\mathsf{span}\left\{ S_{kl}^{q}\cong I_{A_{q}}\otimes\ket{b_{k}^{q}}\bra{b_{l}^{q}}\right\} 
	\end{equation}
	consists of all, and only, the matrices of the form (\ref{eq:Alg correctness matrix decompos}),
	as promised by Theorem \ref{thm:Wedderburn decomposition}.
\end{proof}

\section{Examples of State Reduction via Irrep Decomposition of Operational Constraints \label{sec:Examples-of-State}}

\subsection{Particle with Orbital and Spin Angular Momentum}

Here we consider a single particle with orbital angular momentum $l$ and
spin $1/2$. In light of our discussion in Section \ref{subsec: State Reduction from Operational Constraints},
we would like to know how the quantum state of the particle reduces
if operationally we cannot distinguish between spin and orbital angular momentum and are
constrained to measurements of total angular momentum. This question,
of course, can be addressed with the standard formalism of group representation
theory (or ``addition of angular momentum'' as it is called in physics
textbooks). From this formalism we know that the total angular momentum
operators are reducible and split the Hilbert space of a spin-orbit
particle into $l+1/2$ and $l-1/2$ sectors of total angular momentum,
which are captured by Clebsch-Gordan coefficients. We will now show
that the same conclusion can be reached, including the particular
Clebsch-Gordan coefficients, without relying on the formalism of angular
momentum addition but by instead using the scattering of projections as described
in Section \ref{sec:Irrep-Decompositions-of}.

The Hilbert space of a spin $1/2$ particle with orbital angular momentum $l$
is the tensor product $\mathcal{H}_{L}\otimes\mathcal{H}_{S}$ of
orbital and spin degrees of freedom of dimensions $2l+1$ and $2$
respectively. The total angular momentum component along the $r$
axis (where $r$ can stand for any direction) is given by the operator
\begin{equation}
J_{r}:=L_{r}\otimes I+I\otimes S_{r}\,,
\end{equation}
where $L_{r}$ and $S_{r}$ are the operators of orbital angular momentum
and spin along $r$. Given our operational constraints, we
should look for the irrep structure of the algebra $\left\langle \left\{ J_{r}\right\} \right\rangle $
where $r$ assumes all directions.

We will denote with $\ket{r;m_{L},m_{S}}$ the simultaneous eigenstates
of $L_{r}\otimes I$ and $I\otimes S_{r}$ with the eigenvalues $m_{L}=-l,...,l$
and $m_{S}=\pm\frac{1}{2}$ respectively; we will call these states
the spin-orbit basis. Then, since $J_{r}\ket{r;m_{L},m_{S}}=m_{J}\ket{r;m_{L},m_{S}}$
where $m_{J}=m_{L}+m_{S}$, the spectral decomposition of $J_{r}$
is given by
\begin{equation}
J_{r}=\sum_{m_{J}=-l-\frac{1}{2}}^{l+\frac{1}{2}}m_{J}\Pi_{r;m_{J}}
\end{equation}
and the spectral projections are 
\begin{align}
\Pi_{r;m_{J}} & :=\begin{cases}
\ket{r;\pm l,\pm\frac{1}{2}}\bra{r;\pm l,\pm\frac{1}{2}} & \left|m_{J}\right|=l+\frac{1}{2}\\
\sum_{m_{S}=\pm1/2}\ket{r;m_{J}-m_{S},m_{S}}\bra{r;m_{J}-m_{S},m_{S}} & \left|m_{J}\right|<l+\frac{1}{2}.
\end{cases}
\end{align}
Note that two of the spectral projections (with $\left|m_{J}\right|=l+\frac{1}{2}$)
are rank 1 and the rest are rank 2.

The algebra generated by $\left\{ J_{r}\right\} $, for all $r$,
is also generated by just the two operators $\left\{ J_{z},J_{x}\right\} $.
This is because the rotations $e^{-i\theta J_{x}}$ and $e^{-i\varphi J_{z}}$
are elements of the algebra $\left\langle J_{z},J_{x}\right\rangle $
and every $J_{r}$ can be produced by rotating $e^{-i\varphi J_{z}}e^{-i\theta J_{x}}J_{z}e^{i\theta J_{x}}e^{i\varphi J_{z}}$
with the appropriate angles $\theta,\varphi$. Therefore, in order
to find the irrep structure of the algebra $\left\langle \left\{ J_{r}\right\} \right\rangle $
it is sufficient to consider the algebra generated by the spectral
projections $\left\{ \Pi_{z;m_{J}},\Pi_{x;m_{J}}\right\} $ of $\left\{ J_{z},J_{x}\right\} $.
In the following we denote with $r$ the variable that takes
the values of the two axis $z,x$ and similarly for the capitalized
version $R=Z,X$, which we will later use to label states.

The initial improper reflection network consists of all projections
$\left\{ \Pi_{z;m_{J}},\Pi_{x;m_{J}}\right\} $ for $m_{J}=-l-\frac{1}{2},...,l+\frac{1}{2}$.
Our scattering strategy will be to take the rank 1 projection $\Pi_{x;l+\frac{1}{2}}$
and use it to break all the rank 2 projections $\Pi_{z;m_{J}}$, and
similarly, take the rank 1 projection $\Pi_{z;-l-\frac{1}{2}}$ and
use it to break all the rank 2 projections $\Pi_{x;m_{J}}$. \footnote{The strategy of choosing the scattering pairs is not important from
	the perspective of the raw algorithm we presented in the previous section, but it does make a difference
	in how hard it is to carry it out analytically, as we are doing here.} These scatterings will result in all projections reducing to rank
1, so the reflection network becomes proper and minimal. After that
we will only have to identify the connected components of this network.

The scattering of the pairs $\Pi_{z;m_{J}},\Pi_{x;l+\frac{1}{2}}$
and $\Pi_{x;m_{J}},\Pi_{z;-l-\frac{1}{2}}$ comes down to the spectral
decomposition of
\begin{align}
\Pi_{z;m_{J}}\Pi_{x;l+\frac{1}{2}}\Pi_{z;m_{J}} & =\Pi_{z;m_{J}}\ket{x;l,\frac{1}{2}}\bra{x;l,\frac{1}{2}}\Pi_{z;m_{J}}\\
\Pi_{x;m_{J}}\Pi_{z;-l-\frac{1}{2}}\Pi_{x;m_{J}} & =\Pi_{x;m_{J}}\ket{z;-l,-\frac{1}{2}}\bra{z;-l,-\frac{1}{2}}\Pi_{x;m_{J}}.
\end{align}
Since $\Pi_{r;\pm l\pm\frac{1}{2}}$ ($r=x,z$) are rank 1, they do
not break, so we only need to figure out how the remaining $\Pi_{r;m_{J}}$
break in these scatterings. For that purpose we define the following
states and the associated projections
\begin{align}
\ket{Z;l+\frac{1}{2},m_{J}} & :=\frac{1}{\sqrt{N_{m_{J}}}}\Pi_{z;m_{J}}\ket{x;l,\frac{1}{2}} &  &  & \Pi_{z;m_{J}}^{\left(l+\frac{1}{2}\right)}:=\ket{Z;l+\frac{1}{2},m_{J}}\bra{Z;l+\frac{1}{2},m_{J}}\\
\ket{X;l+\frac{1}{2},m_{J}} & :=\frac{1}{\sqrt{N_{m_{J}}}}\Pi_{x;m_{J}}\ket{z;-l,-\frac{1}{2}} &  &  & \Pi_{x;m_{J}}^{\left(l+\frac{1}{2}\right)}:=\ket{X;l+\frac{1}{2},m_{J}}\bra{X;l+\frac{1}{2},m_{J}}
\end{align}
with the (unimportant) normalization factor $\sqrt{N_{m_{J}}}$.
The capitalized labels $Z,X$ of the axis symbolize the fact that
these are the eigenstates of total angular momentum with eigenvalue
$l+\frac{1}{2}$, as we will see shortly. The scatterings can then
be expressed as
\begin{align}
\Pi_{z;m_{J}}\Pi_{x;l+\frac{1}{2}}\Pi_{z;m_{J}} & =N_{m_{J}}\Pi_{z;m_{J}}^{\left(l+\frac{1}{2}\right)}\\
\Pi_{x;m_{J}}\Pi_{z;-l-\frac{1}{2}}\Pi_{x;m_{J}} & =N_{m_{J}}\Pi_{x;m_{J}}^{\left(l+\frac{1}{2}\right)},
\end{align}
so the rank 1 projections $\Pi_{r;m_{J}}^{\left(l+\frac{1}{2}\right)}$
are one of the spectral projections that come out of scattering, corresponding
to the eigenvalue $N_{m_{J}}$. For $\left|m_{J}\right|=l+\frac{1}{2}$
there is no additional spectral projection since the projections $\Pi_{r;\pm l\pm\frac{1}{2}}$
are rank 1 and they do not break so $\Pi_{r;\pm l\pm\frac{1}{2}}=\Pi_{r;\pm l\pm\frac{1}{2}}^{\left(l+\frac{1}{2}\right)}$,
or equivalently, $\ket{R;l+\frac{1}{2},\pm\left(l+\frac{1}{2}\right)}=\ket{r;\pm l,\pm\frac{1}{2}}$.
For $\left|m_{J}\right|<l+\frac{1}{2}$, the second spectral projection
is given by $\Pi_{r;m_{J}}-\Pi_{r;m_{J}}^{\left(l+\frac{1}{2}\right)}$
and it corresponds to the eigenvalue $0$. Since this is just the
projection on the orthogonal complement of $\ket{R;l+\frac{1}{2},m_{J}}$
in the eigenspace of $\Pi_{r;m_{J}}$, we will have to identify the
orthogonal complements of the states $\ket{R;l+\frac{1}{2},m_{J}}$.

Using their definition above, the states $\ket{Z;l+\frac{1}{2},m_{J}}$, for $\left|m_{J}\right|<l+\frac{1}{2}$,
can be expressed (without worrying about the normalization) as
\begin{align}
\ket{Z;l+\frac{1}{2},m_{J}} & \propto\Pi_{z;m_{J}}\ket{x;l,\frac{1}{2}}\\
& =\sum_{m_{S}=\pm1/2}\ket{z;m_{J}-m_{S},m_{S}}\braket{z;m_{J}-m_{S},m_{S}}{x;l,\frac{1}{2}}\\
& \propto\ket{z;m_{J}-\frac{1}{2},\frac{1}{2}}d_{m_{J}-\frac{1}{2},l}^{l}\left(\frac{\pi}{2}\right)+\ket{z;m_{J}+\frac{1}{2},-\frac{1}{2}}d_{m_{J}+\frac{1}{2},l}^{l}\left(\frac{\pi}{2}\right).
\end{align}
Here we have used the Wigner's ``small'' d-matrix element \cite{rose1995elementary}  $d_{m_{J}-m_{S},l}^{l}\left(\frac{\pi}{2}\right)$
that is obtained from the orbital part of the inner product (the spin
part gives $1/\sqrt{2}$ which we disregard as a normalization factor).
In particular, the specific d-matrix elements we need are given by
\begin{equation}
d_{m_{J}\mp\frac{1}{2},l}^{l}\left(\frac{\pi}{2}\right)=\left(\frac{1}{\sqrt{2}}\right)^{2l}\sqrt{\binom{2l}{l-m_{J}\pm\frac{1}{2}}}
\end{equation}
so, using this expression and normalizing, we obtain the state 
\begin{align*}
\ket{Z;l+\frac{1}{2},m_{J}} & =\ket{z;m_{J}-\frac{1}{2},\frac{1}{2}}\sqrt{\frac{l-m_{J}-\frac{1}{2}}{2l+1}}+\ket{z;m_{J}+\frac{1}{2},-\frac{1}{2}}\sqrt{\frac{l-m_{J}+\frac{1}{2}}{2l+1}}\\
\\
& =\ket{z;m_{J}-\frac{1}{2},\frac{1}{2}}c_{-}^{l+1,m_{J}}+\ket{z;m_{J}+\frac{1}{2},-\frac{1}{2}}c_{+}^{l+1,m_{J}}.
\end{align*}
The coefficients $c_{\pm}^{l+1,m_{J}}:=\sqrt{\frac{l-m_{J}\pm\frac{1}{2}}{2l+1}}$
are the well known Clebsch-Gordan coefficients that arise in spin-orbit
coupling, so we know the states $\ket{Z;l+\frac{1}{2},m_{J}}$ are the states
of total angular momentum $l+1/2$ with $m_{J}$ component along the
$z$ axis. Its orthogonal complement in the two-dimensional subspace
spanned by $\left\{ \ket{z;m_{J}-\frac{1}{2},\frac{1}{2}},\ket{z;m_{J}+\frac{1}{2},-\frac{1}{2}}\right\} $
is just the antipodal point of $\ket{Z;l+\frac{1}{2},m_{J}}$ on the
Bloch sphere:
\begin{equation}
\ket{Z;l-\frac{1}{2},m_{J}}:=\ket{z;m_{J}-\frac{1}{2},\frac{1}{2}}c_{+}^{l+1,m_{J}}-\ket{z;m_{J}+\frac{1}{2},-\frac{1}{2}}c_{-}^{l+1,m_{J}}.
\end{equation}
With this arrangement of Clebsch-Gordan coefficients, these states
are the states of total angular momentum $l-1/2$ with $m_{J}$ component
along the $z$ axis.

The same characterization for the $\ket{X;l+\frac{1}{2},m_{J}}$ states
can be derived from the observation 
\begin{align}
\ket{X;l+\frac{1}{2},m_{J}} & =\frac{1}{\sqrt{N_{m_{J}}}}\Pi_{x;m_{J}}\ket{z;-l,-\frac{1}{2}}\\
& =\frac{1}{\sqrt{N_{m_{J}}}}e^{-i\frac{\pi}{2}J_{y}}\Pi_{z;m_{J}}e^{i\frac{\pi}{2}J_{y}}e^{-i\frac{\pi}{2}J_{y}}\ket{x;l,\frac{1}{2}}\\
& =e^{-i\frac{\pi}{2}J_{y}}\ket{Z;l+\frac{1}{2},m_{J}}
\end{align}
so their orthogonal complements are 
\begin{equation}
\ket{X;l-\frac{1}{2},m_{J}}:=e^{-i\frac{\pi}{2}J_{y}}\ket{Z;l-\frac{1}{2},m_{J}}.
\end{equation}

With the rank 1 projections $\Pi_{r;m_{J}}^{\left(l-\frac{1}{2}\right)}$
on the states $\ket{R;l-\frac{1}{2},m_{J}}$, we can finally conclude
that for $\left|m_{J}\right|<l+\frac{1}{2}$, the projections $\Pi_{z;m_{J}}$
break into $\Pi_{z;m_{J}}^{\left(l+\frac{1}{2}\right)}+\Pi_{z;m_{J}}^{\left(l-\frac{1}{2}\right)}$
and $\Pi_{x;m_{J}}$ break into $\Pi_{x;m_{J}}^{\left(l+\frac{1}{2}\right)}+\Pi_{x;m_{J}}^{\left(l-\frac{1}{2}\right)}$.
The resulting reflection network consists of the projections $\left\{ \Pi_{z;m_{J}}^{\left(l+\frac{1}{2}\right)},\Pi_{x;m_{J}}^{\left(l+\frac{1}{2}\right)},\Pi_{z;m_{J}}^{\left(l-\frac{1}{2}\right)},\Pi_{x;m_{J}}^{\left(l-\frac{1}{2}\right)}\right\} $
for $m_{J}=-l-\frac{1}{2},...,l+\frac{1}{2}$ where the four projections
$\Pi_{z;\pm l\pm\frac{1}{2}}^{\left(l+\frac{1}{2}\right)}$,$\Pi_{x;\pm l\pm\frac{1}{2}}^{\left(l+\frac{1}{2}\right)}$
(for $\left|m_{J}\right|=l+\frac{1}{2}$) are just relabeled $\Pi_{z;\pm l\pm\frac{1}{2}}$,
$\Pi_{x;\pm l\pm\frac{1}{2}}$ and the rest are the result of scatterings.
Since all projections are rank 1, this is a proper minimal reflection
network.

From the fact that $\ket{X;l\pm\frac{1}{2},m_{J}}=e^{-i\frac{\pi}{2}J_{y}}\ket{Z;l\pm\frac{1}{2},m_{J}}$,
and that states of different total angular momentum are orthogonal
to each other, it should be clear that all $\left\{ \Pi_{z;m_{J}}^{\left(l+\frac{1}{2}\right)},\Pi_{x;m_{J}}^{\left(l+\frac{1}{2}\right)}\right\} $
are orthogonal to all $\left\{ \Pi_{z;m_{J}}^{\left(l-\frac{1}{2}\right)},\Pi_{x;m_{J}}^{\left(l-\frac{1}{2}\right)}\right\} $.
At the same time, all $\left\{ \Pi_{z;m_{J}}^{\left(l+\frac{1}{2}\right)}\right\} $
are properly reflecting with all $\left\{ \Pi_{x;m_{J}}^{\left(l+\frac{1}{2}\right)}\right\} $
and similarly for $l-\frac{1}{2}$. Therefore, the reflection network
has two connected components for $l+\frac{1}{2}$ and $l-\frac{1}{2}$.
We choose the maximal independent sets in the connected components
to be $\left\{ \Pi_{z;m_{J}}^{\left(l+\frac{1}{2}\right)}\right\} $
and $\left\{ \Pi_{z;m_{J}}^{\left(l-\frac{1}{2}\right)}\right\} $.
Since all projections are rank 1, there is no freedom in the alignment
of columns in the BPT; it is just two blocks with a single row of
eigenbasis of $\left\{ \Pi_{z;m_{J}}^{\left(l+\frac{1}{2}\right)}\right\} $
and $\left\{ \Pi_{z;m_{J}}^{\left(l-\frac{1}{2}\right)}\right\} $:
\begin{center}
	\begin{tabular}{|c|c|c|c|c|c|}
		\cline{1-3} \cline{2-3} \cline{3-3} 
		$l+\frac{1}{2},l+\frac{1}{2}$ & $\cdots$ & $l+\frac{1}{2},-l-\frac{1}{2}$ & \multicolumn{1}{c}{$\begin{array}{c}
			\\
			\\
			\end{array}$} & \multicolumn{1}{c}{} & \multicolumn{1}{c}{}\tabularnewline
		\hline 
		\multicolumn{1}{c}{} & \multicolumn{1}{c}{} & $\begin{array}{c}
		\\
		\\
		\end{array}$ & $l-\frac{1}{2},l-\frac{1}{2}$ & $\cdots$ & $l-\frac{1}{2},-l+\frac{1}{2}$\tabularnewline
		\cline{4-6} \cline{5-6} \cline{6-6} 
	\end{tabular}
	\par\end{center}

\noindent Each cell corresponds to the state $\ket{Z;l\pm\frac{1}{2},m_{J}}$,
where we have suppressed the $Z$ axis label.

The resulting Hilbert space decomposition 
\begin{equation}
\mathcal{H}_{L}\otimes\mathcal{H}_{S}\cong\mathcal{H}^{\left(l+\frac{1}{2}\right)}\oplus\mathcal{H}^{\left(l-\frac{1}{2}\right)}
\end{equation}
indicates that the restriction to total angular momentum measurements
will result in a superselection between the two total angular momentum
sectors. The accessible state is therefore obtained, according to
Eq.\ (\ref{eq:state reduction for MAs}) of Section \ref{subsec: State Reduction from Operational Constraints},
from the state-reduction map
\begin{equation}
\rho\longmapsto\Pi^{\left(l+\frac{1}{2}\right)}\rho\Pi^{\left(l+\frac{1}{2}\right)}+\Pi^{\left(l-\frac{1}{2}\right)}\rho\Pi^{\left(l-\frac{1}{2}\right)},
\end{equation}
where $\Pi^{\left(l\pm\frac{1}{2}\right)}$ are projections on the
sectors $\mathcal{H}^{\left(l\pm\frac{1}{2}\right)}$. So the coherence terms between total angular momentum sectors are
unobservable if only total angular momentum measurements are allowed.

This conclusion is of course not surprising if we know the theory
of angular momentum addition. But, the fact that the same result,
including the explicit derivation of the total angular momentum states
$\ket{Z;l\pm\frac{1}{2},m_{J}}$ in the spin-orbit basis, can be obtained
by scattering of projections, is a strong confirmation of the viability
of this approach to derivation of irreps. In the next example we will
consider a case where the group representation theory is not as well
developed, yet the projection-scattering method yields the irrep decomposition
in a straight forward way.

\subsection{A Bound Pair of Particles on a Lattice}

In this example we consider a periodic one-dimensional lattice of length $D$ with
two identical particles on it. The two particles are assumed to be
bound in the sense that their relative position and relative momentum
cannot exceed 1 lattice site. This is a simple toy model for a bound
pair of particles on a lattice that oscillate around a common center
of mass with limited energy. The operational constraint that we will
consider is the inability to resolve the composite pair as two separate
particles, which is manifested by a restriction to the center of mass
measurements $\left\{ X_{cm},P_{cm}\right\} $ of both position and
momentum. Once again, as was discussed in Section \ref{subsec: State Reduction from Operational Constraints},
the main challenge is to find the irrep structure of the algebra $\left\langle X_{cm},P_{cm}\right\rangle $.

The $D^{2}$ dimensional Hilbert space $\mathcal{H}_{1}\otimes\mathcal{H}_{2}$
is spanned by the position basis $\ket{x;n_{1},n_{2}}$ for $n_{i}=0,...,D-1$.
The momentum basis states $\ket{p;m_{1},m_{2}}$ are related to the position
basis via the lattice Fourier transform 
\begin{equation}
\ket{p;m_{1},m_{2}}:=F\ket{x;m_{1},m_{2}}=\frac{1}{D}\sum_{n_{1},n_{2}=0}^{D-1}e^{i2\pi\left(m_{1}n_{1}+m_{2}n_{2}\right)/D}\ket{x;n_{1},n{}_{2}}.\label{eq:2 particle furrier trans}
\end{equation}
The center of mass operators are given by

\begin{align}
X_{cm} & :=\frac{1}{2}\left(X_{1}\otimes I_{2}+I_{1}\otimes X_{2}\right)\\
P_{cm} & :=\frac{1}{2}\left(P_{1}\otimes I_{2}+I_{1}\otimes P_{2}\right),
\end{align}
where $X_{i}$, $P_{i}$ are the position and momentum operators on
each particle. In general, $X_{cm}\ket{x;n_{1},n_{2}}=n_{cm}\ket{x;n_{1},n_{2}}$
where $n_{cm}=\left(n_{1}+n_{2}\right)/2$, but, assuming that the
particles cannot occupy the same lattice site simultaneously, \footnote{We refrain from calling the particles fermions because we have no reason to
	assume that their states must be anti-symmetric under particle exchange.} in the bound state we have $n_{2}=n_{1}\pm1$ so $n_{cm}=n_{1}\pm1/2$.
For a shorter notation, we will use the integer $n$ instead of the
half-integer $n_{cm}$ related by $n_{cm}=n+1/2$. Then, for each
possible eigenvalue $n_{cm}$ for bound particles, there are two possible
eigenstates $\ket{x;n,n+1}$ and $\ket{x;n+1,n}$. The same notation
applies to $P_{cm}$.

Therefore, the spectral projections of $X_{cm}$ and $P_{cm}$, when
considering bound particles that cannot occupy the same site, are
given by
\begin{align}
\Pi_{x;n} & :=\ket{x;n,n+1}\bra{x;n,n+1}+\ket{x;n+1,n}\bra{x;n+1,n}\\
\Pi_{p;m} & :=\ket{p;m,m+1}\bra{p;m,m+1}+\ket{p;m+1,m}\bra{p;m+1,m}
\end{align}
for $n,m=0,...,D-1$ and the summation is modulo $D$. The algebra
$\left\langle X_{cm},P_{cm}\right\rangle $ is then generated by the
improper reflection network of $\left\{ \Pi_{x;n},\Pi_{p;m}\right\}, $
which we will now reduce to a proper network by scattering of projections.

The result of scattering of any pair of projections $\{\Pi_{x;n},\Pi_{p;m}\}$
depends on the spectral decomposition of $\Pi^{n;x}\Pi^{m;p}\Pi^{n;x}$.
For this calculation we first define the states

\begin{align}
\ket{\chi_{n}\left(\varphi\right)} & :=\frac{1}{\sqrt{2}}\left(\ket{x;n,n+1}+e^{i\varphi}\ket{x;n+1,n}\right)\\
\ket{\psi_{m}\left(\varphi\right)} & :=\frac{1}{\sqrt{2}}\left(\ket{p;m,m+1}+e^{i\varphi}\ket{p;m+1,m}\right).
\end{align}
Then, using the Fourier transform of Eq.\ (\ref{eq:2 particle furrier trans}),
we derive 
\begin{align}
\Pi_{x;n}\Pi_{p;m}\Pi_{x;n} & =\Pi_{x;n}\ket{p;m,m+1}\bra{p;m,m+1}\Pi_{x;n}+\Pi_{x;n}\ket{p;m+1,m}\bra{p;m+1,m}\Pi_{x;n}\\
& =\frac{2}{D^{2}}\ket{\chi_{n}\left(-\frac{2\pi}{D}\right)}\bra{\chi_{n}\left(-\frac{2\pi}{D}\right)}+\frac{2}{D^{2}}\ket{\chi_{n}\left(\frac{2\pi}{D}\right)}\bra{\chi_{n}\left(\frac{2\pi}{D}\right)}.
\end{align}
The two states $\ket{\chi_{n}\left(\pm\frac{2\pi}{D}\right)}$ are
not orthogonal to each other, so this is not yet the spectral decomposition.
One can check that
the eigenstates of $\Pi_{x;n}\Pi_{p;m}\Pi_{x;n}$ are $\ket{\chi_{n}\left(0\right)}$
and $\ket{\chi_{n}\left(\pi\right)}$ with the distinct eigenvalues
$\frac{2}{D^{2}}\left(1\pm\cos\left(\frac{2\pi}{D}\right)\right)$.
(One can also visualize this fact using the representation of $\ket{\chi_{n}\left(\pm\frac{2\pi}{D}\right)}$ as two
vectors in the $x-y$ plane of the Bloch sphere of the qubit spanned
by $\ket{x;n,n+1}$ and $\ket{x;n+1,n}$.)
Thus, the projection $\Pi_{x;n}$ breaks into $\Pi_{x;n}^{\left(0\right)}+\Pi_{x;n}^{\left(\pi\right)}$,
where $\Pi_{x;n}^{\left(\varphi\right)}:=\ket{\chi_{n}\left(\varphi\right)}\bra{\chi_{n}\left(\varphi\right)}$,
and this result does not depend on the $m$ argument of $\Pi_{p;m}$. A similar
calculation for the scattering of $\Pi_{p;m}$ via the spectral decomposition
of $\Pi_{p;m}\Pi_{x;n}\Pi_{p;m}$, results in it breaking into $\Pi_{p;m}^{\left(0\right)}+\Pi_{p;m}^{\left(\pi\right)}$
where $\Pi_{p;m}^{\left(\varphi\right)}:=\ket{\psi_{m}\left(\varphi\right)}\bra{\psi_{m}\left(\varphi\right)}$.

Therefore, by scattering all (arbitrarily chosen) pairs $\Pi_{x;n},\Pi_{p;m}$,
the initial reflection network reduces to $\left\{ \Pi_{x;n}^{\left(0\right)},\Pi_{x;n}^{\left(\pi\right)},\Pi_{p;m}^{\left(0\right)},\Pi_{p;m}^{\left(\pi\right)}\right\} $
for $n,m=0,...,D-1$. Since now all projections are rank 1, the network
is again proper and minimal. In order to see how it decomposes into connected
components, we note that for $a,b=0,1$ we have (using $e^{-ia\pi}=e^{ia\pi}$)
\begin{equation}
\braket{\chi_{n}\left(a\pi\right)}{\psi_{m}\left(b\pi\right)}=\sqrt{2}\cos\left(\frac{\left(b+a\right)\pi}{2}\right)\left(e^{i2\pi/D}+e^{-ia\pi}\right)e^{i\left(b+a\right)\pi/2}e^{i2\pi\left(2nm+m+n\right)/D}.
\end{equation}
The cosine term tells us that these states are orthogonal for $a\neq b$
and are not orthogonal for $a=b$. Then the subsets $\left\{ \Pi_{x;n}^{\left(0\right)},\Pi_{p;m}^{\left(0\right)}\right\} $
and $\left\{ \Pi_{x;n}^{\left(\pi\right)},\Pi_{p;m}^{\left(\pi\right)}\right\} $
form two separate connected components in the network.

We choose the maximal independent sets in the connected components
to be $\left\{ \Pi_{x;n}^{\left(0\right)}\right\} $ and $\left\{ \Pi_{x;n}^{\left(\pi\right)}\right\} $. (This choice is arbitrary; we could just as well have chosen the momentum basis.)
Then, similarly to the example of a particle with orbital and spin angular momentum, the BPT is
just two blocks with single rows of eigenbasis of $\left\{ \Pi_{x;n}^{\left(0\right)}\right\} $
and $\left\{ \Pi_{x;n}^{\left(\pi\right)}\right\} $: 
\begin{center}
	\begin{tabular}{|c|c|c|c|c|c|c|c|}
		\cline{1-4} \cline{2-4} \cline{3-4} \cline{4-4} 
		$\chi_{0}\left(0\right)$ & $\chi_{1}\left(0\right)$ & $\cdots$ & $\chi_{D-1}\left(0\right)$ & \multicolumn{1}{c}{} & \multicolumn{1}{c}{} & \multicolumn{1}{c}{} & \multicolumn{1}{c}{}\tabularnewline
		\hline 
		\multicolumn{1}{c}{} & \multicolumn{1}{c}{} & \multicolumn{1}{c}{} &  & $\chi_{0}\left(\pi\right)$ & $\chi_{1}\left(\pi\right)$ & $\cdots$ & $\chi_{D-1}\left(\pi\right)$\tabularnewline
		\cline{5-8} \cline{6-8} \cline{7-8} \cline{8-8} 
	\end{tabular}
	\par\end{center}

\noindent \begin{flushleft}
	This BPT indicates the irrep decomposition into two sectors:
	\begin{equation}
	\mathcal{H}_{1}\otimes\mathcal{H}_{2}\cong\mathcal{H}^{\left(0\right)}\oplus\mathcal{H}^{\left(\pi\right)}.
	\end{equation}
	\par\end{flushleft}

\noindent \begin{flushleft}
	Thus we learn that, under restriction to the center of mass measurements,
	the Hilbert space splits into two superselection sectors with symmetric
	$\ket{\chi_{n}\left(0\right)}$ and anti-symmetric $\ket{\chi_{n}\left(\pi\right)}$
	configurations of the bound pair of particles. We can now see that
	this BPT specifies the commutant algebra of particle exchange symmetry,
	and indeed, $X_{cm}$ and $P_{cm}$ commute with exchange of particles
	so they belong to the commutant of this symmetry. This, however, does
	not mean that a priori it was obvious that $\left\{ X_{cm},P_{cm}\right\} $
	generate the whole commutant algebra of this symmetry; it is possible
	that they only generate a subalgebra of the commutant. Only by explicitly
	finding the irreps with the projection scattering method we can be
	certain that $\left\langle X_{cm},P_{cm}\right\rangle $ is the commutant
	algebra of particle exchange.
	\par\end{flushleft}

As discussed in Section \ref{subsec: State Reduction from Operational Constraints},
the bound pair's state reduces by enforcing the superselection with
projections on the superselection sectors $\mathcal{H}^{\left(0\right)}$,
$\mathcal{H}^{\left(\pi\right)}$:

\begin{equation}
\rho\longmapsto\Pi^{\left(0\right)}\rho\Pi^{\left(0\right)}+\Pi^{\left(\pi\right)}\rho\Pi^{\left(\pi\right)}.
\end{equation}
This state reduction accounts for the operational constraints of an
observer that cannot resolve the individual particles. From such an observer's
perspective, each sector $q\pi$ for $q=0,1$ is effectively a single
composite particle with position states $\ket{\chi_{n}\left(q\pi\right)}$
and momentum states $\ket{\psi_{n}\left(q\pi\right)}=F\ket{\chi_{n}\left(q\pi\right)}$.
The distinction between the two sectors is then associated with some
``charge'' $q=0,1$ of the composite particle. Whether this charge
is constant in time depends on the full dynamics of the system. If
the charge is not conserved, meaning the dynamics have tunneling terms
between the symmetric and anti-symmetric states of the pair, the constrained
observer can describe the charge variation as the result of interactions
with an ``environment''. The ``environment'' in this case is the
composite particle's intrinsic degrees of freedom, which are inaccessible
with $\left\{ X_{cm},P_{cm}\right\} $.


\section{Beyond Matrix Algebras: Partial Bipartitions}
\label{PartialBPT}
Thus far we have discussed the case of matrix algebras, where Hilbert space is decomposed into a collection of direct-sum sectors of tensor products. These generalized bipartitions, as described by Eq.\ (\ref{eq:def generalized bipartition}), are represented using their bipartition table (BPT) structure as block-diagonal arrangements of rectangular tables. We will now extend this construction of generalized bipartitions to include the case where some or all of the direct-sum sectors are represented by \emph{non-rectangular} tables.  We will refer to these non-rectangular cases as \emph{partial bipartitions}. The power of partial bipartitions will be relevant when, for example, the set of measurements that can be implemented by an observer in the laboratory does not form an algebra. 

As a motivating example, consider two spin-$\frac{1}{2}$ particles, spanned by the total spin basis labeled by $\{\ket{S_{z},\mu}\}$, where $S_z$ is the total spin-z of the two spins and $\mu$ labels the information about the multiplet nature of the state, with $\mu = s$ for singlet and $\mu=t$ for triplet. A relevant situation is when an experimenter in the lab only has access to measurements of the total spin of the two particles, and not the multiplet information of the quantum state. Written in terms of the computational tensor product basis $\{\ket{0},\ket{1}\}^{\otimes 2}$, we have
\begin{eqnarray}
\ket{1,t} = \ket{1,1} & \\
\ket{0,s} = \frac{  \ket{01} - \ket{10}}{\sqrt{2}} & \\
\ket{0,t} = \frac{  \ket{01} + \ket{10}}{\sqrt{2}} & \\
\ket{-1,t} = \ket{00} \: ,
\end{eqnarray}

This four-dimensional Hilbert space is \emph{not} factorizable into a tensor product structure where one factor describes the total spin-z degree of freedom and the other factor corresponding to the multiplet information. Partial bipartitions offer a natural construction to capture such splits of Hilbert space. Partial bipartitions were first introduced in \cite{2018PhRvA..97e2130K} in the context of quantum coarse-graining and some examples were discussed. In this paper, we will use the concept of partial bipartitions in Sections \ref{sec:variational} and \ref{sec:ising} below, where we will discuss decoherence and coarse-graining of Hilbert space using a variational approach based on an underlying Hamiltonian which governs evolution. Our exposition here of the concept and construction of partial bipartitions, in particular some of the notation, will be with an eye towards the variational approach. 

Let us first consider the case of a single direct-sum factor, so that the BPT is a single non-rectangular table describing a partial bipartition. By virtue of being non-rectangular, the split of the Hilbert space is no longer that of a tensor product structure between the row and column degrees of freedom of the BPT, as was the case for a rectangular BPT, but rather captures a more general partition of the space into two. 
Consider a finite-dimensional Hilbert space $\Hil$ of dimension $\Dim\Hil = d < \infty$ spanned by a choice of orthonormal basis,
\begin{equation}
\Hil \cong \mrm{span} \big\{ \ket{e_{i k}} \big\} \: .
\end{equation}
A partial bipartition of $\Hil$ is specified by an arrangement of the $d$ basis elements into a non-rectangular bipartition table, with $N_C$ columns and $N_{R}$ rows such that $d < N_{C} N_{R}$. As suggested by the notation, the basis element $ \ket{e_{i k}}$, is located in the BPT in the $i$-th row with $i = 1,2,\cdots,N_{R}$ and $k$-th column with $k = 1,2,\cdots, N_{C}$. The BPT is then specified by the heights $\{h_{k}\}$ for each of the $N_{C}$ columns which is the number of basis elements which go in the $k$-th column. In what follows, we will focus on \emph{compact}  non-rectangular BPTs which correspond to the following conditions on the BPT:
\begin{enumerate}
\item{The number of rows of the BPT is equal to the height of the largest column i.e. $\mrm{max}\{h_{k}\} = N_{R}$.}
\item{The $h_{k}$ basis elements which populate the $k$-th column are stacked together, starting from the first row without having any breaks in them.}
\end{enumerate}
A compact BPT minimizes loss of coherence under the action of the state-reduction map defined by the BPT. Such loss of coherence under state reduction is akin to superselection which is different than the dynamical decoherence induced by the Hamiltonian we will be interested in in the following sections. In Eq.\ (\ref{eq:partial_BPT}) below, we depict a generic compact non-rectangular BPT specifying a partial bipartition of $\Hil \cong \Hil_{A} \oslash \Hil_{B}$. Arrows point toward the associated states of partial subsystems.

\noindent %
\noindent\begin{minipage}[c]{1\columnwidth}%
\begin{tabular}{>{\centering}m{0.9\columnwidth}>{\raggedright}m{0.1\columnwidth}}
 & \tabularnewline
\centering{}%
\begin{tabular}{|c|c|c|c|c|ccc|c|}
\cline{1-7} \cline{9-9} 
$e_{1,1}$ & ... & $e_{1,k}$ & ... & $e_{1,w_{i}}$ & \multicolumn{1}{c|}{...} & \multicolumn{1}{c|}{$e_{1,N_{C}}$} & $\shortrightarrow$ & $\alpha_{1}$\tabularnewline
\cline{1-7} \cline{9-9} 
$\vdots$ & $\vdots$ & $\vdots$ & $\vdots$ & $\vdots$ & \multicolumn{1}{c|}{$\vdots$} & \multicolumn{1}{c|}{$\vdots$} &  & $\vdots$\tabularnewline
\cline{1-7} \cline{9-9} 
$e_{i,1}$ & ... & $e_{i,k}$ & ... & $e_{i,w_{i}}$ &  &  & $\shortrightarrow$ & $\alpha_{i}$\tabularnewline
\cline{1-5} \cline{9-9} 
$\vdots$ & $\vdots$ & $\vdots$ & $\vdots$ & $\vdots$ &  &  &  & $\vdots$\tabularnewline
\cline{1-5} \cline{9-9} 
$e_{h_{k},1}$ & ... & $e_{h_{k},k}$ & \multicolumn{1}{c}{} & \multicolumn{1}{c}{} &  &  & $\shortrightarrow$ & $\alpha_{h_{k}}$\tabularnewline
\cline{1-3} \cline{9-9} 
$\vdots$ & $\vdots$ & \multicolumn{1}{c}{} & \multicolumn{1}{c}{} & \multicolumn{1}{c}{} &  &  &  & $\vdots$\tabularnewline
\cline{1-2} \cline{9-9} 
$e_{N_{R},1}$ & ... & \multicolumn{1}{c}{} & \multicolumn{1}{c}{} & \multicolumn{1}{c}{} &  &  & $\shortrightarrow$ & $\alpha_{N_{R}}$\tabularnewline
\cline{1-2} \cline{9-9} 
\multicolumn{1}{c}{$\shortdownarrow$} & \multicolumn{1}{c}{} & \multicolumn{1}{c}{$\shortdownarrow$} & \multicolumn{1}{c}{} & \multicolumn{1}{c}{$\shortdownarrow$} &  & $\shortdownarrow$ & \multicolumn{1}{c}{} & \multicolumn{1}{c}{}\tabularnewline
\cline{1-7} 
$\beta_{1}$ & ... & $\beta_{k}$ & ... & $\beta_{w_{i}}$ & \multicolumn{1}{c|}{...} & \multicolumn{1}{c|}{$\beta_{N_{C}}$} & \multicolumn{1}{c}{} & \multicolumn{1}{c}{}\tabularnewline
\cline{1-7} 
\end{tabular}
\begin{minipage}[c]{0.04\columnwidth}%
\begin{center}
\begin{equation}
\label{eq:partial_BPT}
\end{equation}
\par\end{center}
~%
\end{minipage}\tabularnewline
 & \tabularnewline
\end{tabular}%
\end{minipage}

It should be noted that as long as the compact form condition is met, there is still some freedom, albeit inconsequential, in the locational arrangement of basis elements in the BPT which will have no consequence in the state-reduction map defined by the BPT. For example, in Eq.\ (\ref{eq:partial_BPT}), one can swap any two columns, which is equivalent to swapping the order of basis in the reduced state space, and that will still leave the partial bipartition encoded in the BPT.

As we discussed in Section \ref{sec:Prelim - Bipartition Tables}, since generalized bipartitions describe tensor-product splits of Hilbert space, and direct-sum sectors thereof; we can immediately infer that a partial bipartition describes splits of Hilbert space more general than tensor factorization. The span of the row (column) kets $\left\{ \left|\alpha_{i}\right\rangle \right\} _{i=1}^{N_{R}}$ ($\left\{ \left|\beta_{k}\right\rangle \right\} _{k=1}^{N_{C}}$) is defined to be the row (column) Hilbert space $\Hil_{A}$  ($\Hil_{B}$) as illustrated in Eq.\ (\ref{eq:partial_BPT}). These can be identified as \emph{partial subsystems} of the full underlying Hilbert space $\Hil$ and we represent this partial factorization as,
\begin{equation}
\Hil \cong \Hil_{A} \oslash \Hil_{B} \: .
\end{equation}
One can always isometrically embed a partial bipartition of a Hilbert space into a larger tensor product Hilbert space defined by $\Hil_{AB} \cong \Hil_
{A} \otimes \Hil_{B}$, such that for every $\ket{e_{i k}} \in \Hil$, there is a matching $\ket{\alpha_{i}}\ket{\beta_{k}} \in \Hil_{AB}$ but not vice-versa. The extra pairs in $\Hil_{AB}$ which do not have a match in $\Hil$ correspond to the missing elements of the BPT that would complete it to a rectangular, and hence, tensor product form.

Tensor product structures which correspond to generalized bipartitions are thus a special case of partial bipartitions which have rectangular BPTs, satisfying the condition $d = N_{C} N_{R}$. 

Once the partial subsystem $\Hil_{A}$ is identified, we can define a state-reduction map which will ``trace'' out $\Hil_A$, akin to a partial-trace map in the case of tensor products, but defined appropriately for partial subsystems. We denote this state-reduction map for the case of partial subsystems as $tr_{\left(A\right)}$ which maps the density matrices between the operator spaces as
\begin{equation}
tr_{\left(A\right)}:\mathcal{L}\left(\mathcal{H}\right)\longrightarrow\mathcal{L}\left(\mathcal{H}_{B}\right),
\end{equation}
so the reduced state-space is indeed described by the partial subsystem $\Hil_B$ as expected. We use a bracketed subscript $_{\left(A\right)}$ in $tr_{\left(A\right)}$ to denote the state-reduction map of a partial system, as opposed to the unbracketed one $_{A}$, which refers to the usual partial-trace map for tensor factors.

The action of $tr_{\left(A\right)}$ on the matrix elements in the bipartition basis $\left|e_{i,k}\right\rangle $,
\begin{align}
tr_{\left(A\right)}:\left|e_{i,k}\right\rangle \left\langle e_{j,l}\right| & \longmapsto\delta_{ij}\left|\beta_{k}\right\rangle \left\langle \beta_{l}\right|.\label{eq:Tr_(a) acting on matrix element}
\end{align}
thus traces over the row indices $i$, $j$ as if they label basis elements
of a proper tensor factor of Hilbert space.

Based on the BPT structure, the original Hilbert space $\Hil$ can be decomposed into a direct-sum sectors, each corresponding to the subspace spanned by basis elements of a single column,
\begin{equation}
\Hil \cong \bigoplus_{k = 1}^{N_{C}} \Hil_{k} \: .
\end{equation}
Similar to the case of generalized bipartitions, we can define bipartition operators (BPOs) for partial subsystems,
\begin{align}
S_{kl} & = \sum_{i=1}^{\mathsf{\mathsf{min}}\left(h_{k},h_{l}\right)}\left|e_{i,k}\right\rangle \left\langle e_{i,l}\right| \: , &  k,l = 1,2,3,\cdots,N_{C} \: ,
\label{eq:def of bipartition operators}
\end{align}
that map between columns of the bipartition table by preserving the
row index $i$ of each element (where it should be understood that the element is skipped in the sum if the row
is not present in the destination column). BPOs of the form $S_{kk}$ correspond to projectors on the column $\Hil_{k}$ subspace and the ones of the form $S_{kl}$ with $k \neq l$   implement partial isometries from (a subspace of)  $\Hil_{l}$ to (a subspace of)  $\Hil_{k}$ (depending on which dimension is lower). Written in terms of bipartition operators, the state-reduction map maps a density matrix $\rho \in \mathcal{L}(\Hil)$ to a reduced, traced out state $\rho_{B} \in \mathcal{L}(\Hil_{B})$ ,
\begin{align}
\rho_B = tr_{\left(A\right)}\left(\rho\right) & =\sum_{k,l=1}^{N_{C}}tr\left(S_{kl}\rho\right)\left|\beta_{l}\right\rangle \left\langle \beta_{k}\right|.\label{eq:tr_(A)  action with S_kl}
\end{align}

As an illustrative example, consider the 6 dimensional Hilbert space $\mathcal{H}$ spanned by the orthonormal basis $\{\ket s\}$ for $s=1,...,6$. A
partial bipartition of $\mathcal{H}$ is chosen such that in the basis
$\{\ket s\}$, it is specified by the bipartition table,
\begin{center}
\begin{tabular}{c|c|c}
\hline 
\multicolumn{1}{|c|}{$1$}& $2$ & \multicolumn{1}{c|}{$3$}\tabularnewline
\hline 
\multicolumn{1}{|c|}{$4$} & $5$ & \tabularnewline
\cline{1-2} 
& \multicolumn{1}{|c|}{$6$} &  \tabularnewline
\cline{2-2} 
\end{tabular}
\begin{equation}
\end{equation}
\par\end{center}
While one can identify a notational correspondence between states $\{\ket{s}\}$ and $\{\ket{e_{i,k}}\}$ using their row/column location in the BPT, we will stick with the $\ket{s}$ notation since it will allow ease of representation of matrix elements of operators in this basis, such as the density matrix. It should be noted that the above BPT is compact.
Now, for a given density matrix $\rho$ written in the bipartition basis $\{\ket{s}\}$ ordered by appearance in the bipartition
table (read from left to right and top to bottom), the action of the state-reduction map $tr_{\left(A\right)}$ to trace out the partial subsystem $\Hil_{A}$, is 
\begin{equation}
\begin{array}{c}
\begin{pmatrix}{\color{red}\rho_{11}} & {\color{red}\rho_{12}} & {\color{red}\rho_{13}} & \rho_{14} & \rho_{15} & \rho_{16}\\
{\color{red}\rho_{21}} & {\color{red}\rho_{22}} & {\color{red}\rho_{23}} & \rho_{24} & \rho_{25} & \rho_{26}\\
{\color{red}\rho_{31}} & {\color{red}\rho_{32}} & {\color{red}\rho_{33}} & \rho_{34} & \rho_{35} & \rho_{36}\\
\rho_{41} & \rho_{42} & \rho_{43} & {\color{green}\rho_{44}} & {\color{green}\rho_{45}} & \rho_{46}\\
\rho_{51} & \rho_{52} & \rho_{53} & {\color{green}\rho_{54}} & {\color{green}\rho_{55}} & \rho_{56}\\
\rho_{61} & \rho_{62} & \rho_{63} & \rho_{64} & \rho_{65} & {\color{blue}\rho_{66}}
\end{pmatrix}\\
\\
\downarrow tr_{\left(A\right)}\\
\\
\begin{pmatrix}{\color{red}\rho_{11}}+{\color{green}\rho_{44}} & {\color{red}\rho_{12}}+{\color{green}\rho_{45}} & {\color{red}\rho_{13}}\\
{\color{red}\rho_{21}}+{\color{green}\rho_{54}} & {\color{red}\rho_{22}}+{\color{green}\rho_{55}}+{\color{blue}\rho_{66}} & {\color{red}\rho_{23}}\\
{\color{red}\rho_{31}} & {\color{red}\rho_{32}} & {\color{red}\rho_{33}}
\end{pmatrix}
\end{array}
\end{equation}

From this we can understand the action of state-reduction map from the bipartition
table:
\begin{enumerate}
\item Coherences between basis elements $\left|e_{i,k}\right\rangle \left\langle e_{j,l}\right|$
in different rows ($i\neq j$) of the bipartition table are discarded. Coherences between basis elements in the same row of the BPT are preserved.
\item For each pair of columns $k,l$ (including $k=l$), the sum of coherences
between $\left|e_{i,k}\right\rangle \left\langle e_{i,l}\right|$
over all rows $i$ is the new coherence term for the reduced element
$\left|\beta_{k}\right\rangle \left\langle \beta_{l}\right|$.
\end{enumerate}   
The number of matrix elements of $\rho$ which also appear in $\rho_{B}$ after the state-reduction map is applied depends on the alignment structure of the cells in the BPT. In particular, some elements do not appear in the reduced density matrix. A natural question to ask is what information is preserved by the state-reduction map induced by the partial bipartition. It was shown in \cite{2018PhRvA..97e2130K} that the bipartition operators $S_{kl}$ span the operator subspace of all (and only) the observables whose information is preserved under state reduction. Then we can interpret the reduced state $\rho_{B}$ as the state that contains all (and only) the information that is accessible with the observables in the operator space $\textrm{\ensuremath{\mathsf{span}}}\left\{ S_{kl}\right\} $. This naturally reduces to the standard picture in the familiar case of a tensor-product bipartition $\mathcal{H} \cong\mathcal{H}_{A}\otimes\mathcal{H}_{B}$,
where the bipartition operators take the form 
\begin{equation}
S_{kl}=I_{A}\otimes\ket{\beta_{k}}\bra{\beta_{l}}.
\end{equation}
The restricted set of observables $\textrm{\ensuremath{\mathsf{span}}}\left\{ S_{kl}\right\} =I_{A}\otimes\mathcal{L}\left(\mathcal{H}_{B}\right)$
imply that the observer can only measure system $\Hil_{B}$. 

We can also generalize the partial bipartition structure to include direct-sum sectors thereof, which corresponds to the following decomposition of Hilbert space,
\begin{equation}
\Hil \cong \bigoplus_{q} \left( \Hil_{A_{q}} \oslash \Hil_{B_{q}} \right) \: ,
\end{equation}
where each sector $q$ is spanned by the basis elements $\ket{e_{ik}^{q}}$
of the block $q$ and each sector is further decomposed into a partial bipartition according to the arrangement of elements inside
the block.  Such a decomposition can be captured as a bipartition table with a block-diagonal arrangement of non-rectangular tables,
\begin{center}
\begin{tabular}{ccc}
\begin{tabular}{|c|c|c}
\hline 
$e_{1;1,1}$ & $e_{1;1,2}$ & \multicolumn{1}{c|}{$...$}\tabularnewline
\hline 
$e_{1;2,1}$ & $\ddots$ & \tabularnewline
\cline{1-2} 
$\vdots$ & \multicolumn{1}{c}{} & \tabularnewline
\cline{1-1} 
\end{tabular} &  & \tabularnewline
 & %
\begin{tabular}{|c|c|c}
\hline 
$e_{2;1,1}$ & $e_{2;1,2}$ & \multicolumn{1}{c|}{$...$}\tabularnewline
\hline 
$e_{2;2,1}$ & $\ddots$ & \tabularnewline
\cline{1-2} 
$\vdots$ & \multicolumn{1}{c}{} & \tabularnewline
\cline{1-1} 
\end{tabular} & \tabularnewline
 &  & $\ddots$\tabularnewline
\end{tabular}\be\label{eq:pb}\ee\
\par\end{center} 

 For each sector $q$, we can define a set of bipartition operators $\{S_{kl}^q \}$ using the basis elements in that sector. By construction, under the state-reduction map specified by such a BPT, coherences between different direct-sum sectors are lost, and the resultant density matrix will be block-diagonal corresponding to different blocks $q$.

\subsection*{Examples}
 
Let us return to the example of the two spin-$\frac{1}{2}$ particles we raised at the beginning of this section. Again, consider the total spin basis labeled by $\{\ket{S_{z},\mu}\}$, where $S_z$ is the total spin-z of the two spins and $\mu$ labels the information about the multiplet nature of the state (with $\mu = s$ for singlet and $\mu=t$ for triplet), which written in terms of the computational tensor product basis $\{\ket{0},\ket{1}\}^{\otimes 2}$,
\begin{eqnarray}
\ket{1,t} = \ket{1,1} & \\
\ket{0,s} = \frac{  \ket{01} - \ket{10}}{\sqrt{2}} & \\
\ket{0,t} = \frac{  \ket{01} + \ket{10}}{\sqrt{2}} & \\
\ket{-1,t} = \ket{00} \: ,
\end{eqnarray}
and a partial bipartition of this Hilbert space,

\begin{center}
\begin{tabular}{|c|c|c| c |c|}
\cline{1-3}\cline{5-5}
$1,t$ & $0,t$ & $-1,t$ & $\shortrightarrow$ & $\alpha_{t}$\tabularnewline
\cline{1-3}\cline{5-5}
\multicolumn{1}{c|}{} & $0,s$ & \multicolumn{1}{c}{} & $\shortrightarrow$ & $\alpha_{s}$\tabularnewline
\cline{2-2} \cline{5-5}
\multicolumn{1}{c}{$\shortdownarrow$} & \multicolumn{1}{c}{$\shortdownarrow$} & \multicolumn{1}{c}{$\shortdownarrow$} & \multicolumn{1}{c}{} &  \multicolumn{1}{c}{} \tabularnewline
\cline{1-3}
\multicolumn{1}{|c}{$\beta_{-1}$} & \multicolumn{1}{|c}{$\beta_0$} & \multicolumn{1}{|c}{$\beta_1$} & \multicolumn{1}{|c}{} &  \multicolumn{1}{c}{} \tabularnewline
\cline{1-3}
\end{tabular}
\begin{equation}
\end{equation}
\par\end{center}

The degree of freedom fixed by each column is the total
spin-$z$ and what varies within the columns is the multiplet (singlet-triplet)
label. With the non-rectangular nature of the BPT, the column space $\Hil_{B}$ (spanned by $\{\ket{\beta_{-1}} , \ket{\beta_{0}}, \ket{\beta_1}\}$) forms a partial subsystem which encodes variation of the total spin-z and
the row space $\Hil_{A}$ (spanned by $\{\ket{\alpha_t} , \ket{\alpha_{s}}\}$) forms a partial subsystem which encodes variation of multiplet information. We can now define the BPOs labeled by the value of total spin-$z$ from the BPT column structure, 
\begin{center}
\begin{tabular}{lllllllll}
$S_{+1,+1}=\ket{1,t}\bra{1,t}$ &  &  &  & $S_{+1,0}=\ket{1,t}\bra{0,t}$ &  &  &  & $S_{+1,-1}=\ket{1,t}\bra{-1,t}$\tabularnewline
 &  &  &  &  &  &  &  & \tabularnewline
 &  &  &  & $S_{0,0}=\ket{0,t}\bra{0,t}+\ket{0,s}\bra{0,s}$ &  &  &  & $S_{0,-1}=\ket{0,t}\bra{-1,t}$\tabularnewline
 &  &  &  &  &  &  &  & \tabularnewline
 &  &  &  &  &  &  &  & $S_{-1,-1}=\ket{-1,t}\bra{-1,t}$\tabularnewline
\end{tabular}
\begin{equation}
\end{equation}
\par\end{center}

The other three BPOs can simply be obtained from $S_{kl}=S_{lk}^{\dagger}$.
The state-reduction map induced by this BPT,
\begin{equation}
\rho\longmapsto\sum_{k,l=+1,0,-1}tr\left(S_{kl}\rho\right)\left|l\right\rangle \left\langle k\right|
\end{equation}
can be interpreted as tracing out the multiplet degree of
freedom. The resulting state has the degrees of freedom associated
with the total spin-$z$ of the original system: that is has the Hilbert space of a single composite particle with this spin. The total spin operators
$S_{x}^{tot}$, $S_{y}^{tot}$, $S_{z}^{tot}$ on the two particles
are in the span of $\left\{ S_{kl}\right\} $ 
\begin{align*}
S_{z}^{tot} & =S_{+1,+1}-S_{-1,-1}\\
S_{+}^{tot} & =S_{+1,0}+S_{0,-1}\\
S_{-}^{tot} & =S_{0,+1}+S_{-1,0} \:,
\end{align*}
where $S_{x}^{tot}$ and $S_{y}^{tot}$ can be constructed from the ladder operators
$S_{\pm}^{tot}$. Therefore, the reduced state preserves information
about total spin operators $S_{x}^{tot}$, $S_{y}^{tot}$, $S_{z}^{tot}$. It should be noted that for such partial bipartitions,
the span of $\left\{ S_{kl}\right\} $ is not necessarily an algebra
(it may not be closed under products) so even if we know
that the reduced state preserves information about $S_{x}^{tot}$,
$S_{y}^{tot}$, $S_{z}^{tot}$, it may not retain information about
their products which we usually taken for granted. Thus, we see a more general picture emerging where we can define partitions of Hilbert space based on a restricted set of observables which need not generate an algebra. 
Partial bipartitions offer a construction to account for such cases.

This construction easily extends, for example, to an arbitrary number $N$ of spin
$\frac{1}{2}$ participles. For even $N$, the BPT takes the form,
\begin{center}
{\small{}}%
\begin{tabular}{ccc|c|c|c|ccc}
\hline 
\multicolumn{1}{|c|}{{\small{}\strut $\frac{N}{2},+\frac{N}{2}$}} & \multicolumn{1}{c|}{{\small{}$\cdots$}} & {\small{}$\frac{N}{2},+2$} & {\small{}$\frac{N}{2},+1$} & {\small{}$\frac{N}{2},0$} & {\small{}$\frac{N}{2},-1$} & \multicolumn{1}{c|}{{\small{}$\frac{N}{2},-2$}} & \multicolumn{1}{c|}{{\small{}$\cdots$}} & \multicolumn{1}{c|}{{\small{}$\frac{N}{2},-\frac{N}{2}$}}\tabularnewline
\hline
{\small{}\strut } & \multicolumn{1}{c|}{} & {\small{}$\vdots$} & {\small{}$\vdots$} & {\small{}$\vdots$} & {\small{}$\vdots$} & \multicolumn{1}{c|}{{\small{}$\vdots$}} &  & \tabularnewline
\cline{3-7} \cline{4-7} \cline{5-7} \cline{6-7} \cline{7-7} 
{\small{}\strut } & \multicolumn{1}{c|}{} & {\small{}$2,+2$} & {\small{}$2,+1$} & {\small{}$2,0$} & {\small{}$2,-1$} & \multicolumn{1}{c|}{{\small{}$2,-2$}} &  & \tabularnewline
\cline{3-7} \cline{4-7} \cline{5-7} \cline{6-7} \cline{7-7} 
{\small{}\strut } &  &  & {\small{}$1,+1$} & {\small{}$1,0$} & {\small{}$1,-1$} &  &  & \tabularnewline
\cline{4-6} \cline{5-6} \cline{6-6} 
{\small{}\strut } &  &  & {\small{}$\vdots$} & {\small{}$\vdots$} & {\small{}$\vdots$} &  &  & \tabularnewline
\cline{4-6} \cline{5-6} \cline{6-6} 
{\small{}\strut } &  &  & {\small{}$1,+1$} & {\small{}$1,0$} & {\small{}$1,-1$} &  &  & \tabularnewline
\cline{4-6} \cline{5-6} \cline{6-6} 
{\small{}\strut } &  & \multicolumn{1}{c}{} &  & {\small{}$0,0$} & \multicolumn{1}{c}{} &  &  & \tabularnewline
\cline{5-5} 
{\small{}\strut } &  & \multicolumn{1}{c}{} &  & {\small{}$\vdots$} & \multicolumn{1}{c}{} &  &  & \tabularnewline
\cline{5-5} 
{\small{}\strut } &  & \multicolumn{1}{c}{} &  & {\small{}$0,0$} & \multicolumn{1}{c}{} &  &  & \tabularnewline
\cline{5-5} 
\end{tabular}{\small\par}
\begin{equation}
\end{equation}
\par\end{center}

\noindent where for each value of the total spin $j$ there are multiple
equivalent representations that we have stacked on top of each other (suppressing
the label for identical representations with the same $j$). Each
row is associated with a specific ``copy'' $\nu$ of the total spin
$j$ representation and the pair $j,\nu$ specifies one multiplet (note the different notation used in this table in contrast to the $N=2$ spin case above for clarity of exposition of the idea).
As before we use the columns to define the BPO $\left\{ S_{kl}\right\} $
and then the map that traces over the multiplets.The resulting state
is associated with a single spin $\frac{N}{2}$ particle that encodes
information about the state of the total spin of this system.

We will now use these partial bipartitions in the following Sections \ref{sec:variational} and \ref{sec:ising} where we will construct a paradigm to find quasi-classical coarse-grainings of Hilbert space, based on a collective feature of the system compatible with Hamiltonian evolution. Such coarse-grainings will typically not correspond to tensor factorizations of Hilbert space, and hence using this technology of partial bipartitions, we will be able to capture more general partitions suited for the purpose.


\section{Classicality from Coarse-Grainings using Partial Bipartitions: Variational Approach} \label{sec:variational}
As discussed in Section \ref{PartialBPT}, partial bipartitions offer a more general way than a standard tensor product structure to decompose Hilbert space into two parts. One particular application of partial bipartitions is to coarse-grain Hilbert space, since, in many situations, the relevant information preserved by the coarse-graining will not correspond to a tensor factor of Hilbert space. In this section, we outline a paradigm to find quasi-classical coarse-grainings of Hilbert space based on Hamiltonian evolution of the system. We call this the \emph{variational} approach since we will iterate/vary over all possible BPTs (in some restricted set) to find the one(s) which demonstrate quasi-classical behavior.

A natural feature of many coarse-grainings is that they focus on collective or average properties of the underlying degrees of freedom and ignore its internal structure; for example, one can, under appropriate circumstances, coarse-grain a rigid-body system of $N$-particles into its center of mass coordinate, which is a collective feature, while discarding information about the relative locations of the particles, then study how the coarse-grained variable evolves and what characteristics the coarse-graining preserves. We will focus on such coarse-grainings based on a collective property of the system, and their compatibility with dynamics which demonstrate quasi-classical behavior.

\subsection{Microfactorizations and Compatible Collective Observables}

Consider a finite-dimensional Hilbert space of a collection of $N$ underlying degrees of freedom (dofs) specified by a tensor-product structure,
\begin{equation}
\label{variational_microdecomposition}
\Hil \cong \bigotimes_{\mu = 1}^{N} \Hil_{\mu} \: ,
\end{equation}
 which evolve under Hamiltonian evolution given by $\Ham$.  We consider these $N$ degrees of freedom to be fixed, specified by the physical system under consideration, e.g. a collection of $N$ particles, etc. It is assumed that the Hamiltonian in general admits interactions between all $N$ dofs, and in case there exist any subsets of these dofs which are decoupled under the action of the Hamiltonian, we consider each such decoupled subspace individually in this prescription. Our goal is to develop a coarse-graining algorithm informed by the Hamiltonian $\Ham$ which chooses a partial bipartition, $ \Hil \cong \mathcal{S}_{\text{collective}} \oslash \mathcal{S}_{\text{internal}}$. The partial subsystem $\mathcal{S}_{\text{collective}}$ is the coarse-grained version of $\Hil$ we wish to preserve under the coarse-graining/state-reduction map based on a characteristic collective feature of the system (involving all $N$ degrees of freedom) compatible with the Hamilltonian which behaves classically (in a sense we define below) by tracing over the space of internal features $\mathcal{S}_{\text{internal}}$. This will correspond to a BPT of a partial bipartition where the columns will define the coarse-grained subspace $\mathcal{S}_{\mrm{collective}}$ and the rows will define the $\mathcal{S}_{\mrm{internal}}$ subspace that will be traced over. 
 
 Let us define the set of \emph{collective} observables of the full Hilbert space as those that can be written as
\begin{equation}
\label{Mc_definition}
M_{c} = \sum_{\mu = 1}^{N} M_{\mu} \: ,
\end{equation}
with,
\begin{equation}
\label{M_mu_operator}
M_{\mu} = \eye_{1} \otimes \eye_{2} \otimes \ldots \otimes \eye_{\mu - 1} \otimes m_{\mu} \otimes \eye_{\mu + 1}\otimes \ldots \otimes \eye_{N} \: ,
\end{equation}
where each $M_{\mu}$ acts \emph{non-trivially only} on $\Hil_{\mu}$ whose dimension we take to be $\Dim\Hil_{\mu} = d_{\mu}$. The operator $m_{\mu}$ can be parameterized by
\begin{equation}
\label{Mmu_parametrization}
m_{\mu} = \left( \sum_{a = 1}^{d^{2}_{\mu} - 1} c^{(\mu)}_{a} {\Lambda}^{(\mu)}_{a}  \right),
\end{equation}
where ${\Lambda}^{(\mu)}_{a}$ are the $d^{2}_{\mu} - 1$ Generalized Gell-Mann generators (of $SU(d_{\mu})$) which form a complete basis of non-trivial (i.e. without the identity $\eye_{\mu}$) operators acting on $\Hil_{\mu}$. To ensure that these operators have a non-trivial action on the degree of freedom they act on, we impose the restriction that at least one of the $c^{(\mu)}_{a} \neq 0$, for each $\mu$. In addition, we will mostly work with normalized operators on $\mathcal{L}(\Hil)$, the space of linear operators on $\Hil$, to be able to focus on features true to the structure of different operators, and not explicitly due to difference in overall multiplicative factors. For concreteness, we choose to use the Frobenius norm\footnote{The Frobenius norm of a linear operator ${A} \in \mathcal{L}(\Hil)$, also referred to as the Hilbert-Schmidt norm, is defined as,
\begin{equation}
||{A}||_{\mrm{f}} = \sqrt{\Tr\left({A}^{\dag} {A}\right)} \: .
\end{equation}
}
 in this paper, under which the collective observable will be normalized, i.e., $||M_{c} ||_{\mrm{f}} = 1$. 

Our coarse-graining prescription aims for a collective observable as one of the defining properties of $\mathcal{S}_{\text{collective}}$ which is \emph{most compatible} with---that is, \emph{stationary} with respect to---the Hamiltonian $\Ham$ (we will use a normalized version of the Hamiltonian under the Frobenius norm too).  Thus, one can pick out the most compatible collective observable relevant to the coarse-graining by minimizing the norm of the commutator $\lcb \Ham, M_{c} \rcb$ over all choices of collective observables $M_{c}$,
\begin{equation}
\label{Mc_compatibility}
M_{c} \: \: : \: \:  \mrm{min}_{\{ c^{(\mu)}_{a}\}} \Bigl| \Bigl| \lcb \Ham, M_{c} \rcb \Bigr| \Bigr|_{f} \: .
\end{equation}
This is in close parallel with the ideas of the \emph{predictability sieve} \cite{Zurek:1994zq} used in the decoherence literature, where one sifts through different states in Hilbert space to determine the set which is most compatible with the Hamiltonian and is used to define pointer states of the system which are classical. One defining feature of classical dynamics is \emph{robustness} of a set of states (the classical ones) reflected in their effective deterministic classical character. Said differently, the pointer states are special low-entropy states which under evolution stay robust to entanglement production: a given classical degree of freedom in the system does not arbitrarily entangle with all other degrees of freedom  at short time scales. This is intimately linked to the form of and constraints on the Hamiltonian, such as locality \cite{mereology}. Beginning with low entropy states is natural given the second law of thermodynamics, and classicality constrains the rate of entanglement growth for classical states.  In this sense, eigenstates of the collective observable chosen by the compatibility condition of Eq.\ (\ref{Mc_compatibility}) are classical, macroscopic pointer states which capture an average, collective property of the underlying dofs which is as robust under evolution as possible. Take note that the collective observable \emph{cannot} be the identity operator (which would trivially commute with $\Ham$) since we are only considering non-trivial observables which have no support on the identity operator. The total freedom in the choice of $M_{c}$ are the $\sum_{\mu = 1}^{N} \left(d^{2}_{\mu} - 1 \right)$ number of parameters $\{c^{(\mu)}_{a}\}$. Further restrictions on the set $\{c^{(\mu)}_{a}\}$ can be imposed by looking at the symmetry structure of the Hamiltonian and the Hilbert space, if any. For example, if $\Hil$ contains a collection of identical but distinguishable particles on which the Hamiltonian acts symmetrically, then this can be used to constrain the form of $M_{\mu}$ to be the same for this collection of particles. 

Now, as expected, due to the collective observable containing a slew of identity operators in each term in the sum, $M_{c}$ will have a high degeneracy in its eigenspectrum; therefore, the distinct eigenvalues of $M_{c}$ will be used to label distinct columns of the BPT which will define the coarse-graining $\mathcal{S}_{\text{collective}}$. The compatibility condition Eq.\ (\ref{Mc_compatibility}) of the collective observable $M_{c}$ with the Hamiltonian of Eq.\ (\ref{Mc_compatibility}) will ensure that transition of an eigenstate of $M_{c}$ (which corresponds to a deterministic value of the collective variable) into other eigenstates will be minimized under time evolution, and hence that the columns of the BPT correspond to robust collective macrostates. Once the collective observable has been selected, it will give us a total of $N_{C} \leq \Dim\Hil$ distinct eigenvalues and corresponding $N_{C}$ subspaces $\Hil^{(c)}_{k}$ with dimension $\Dim\Hil^{(c)}_{k} = h_{k}$ for $k = 1,2,3,\cdots,N_{C}$. Each such subspace labels a distinct value of the collective observable, specifying a macroscopic pointer state of the coarse-graining. In addition to the subspace structure determined by specification of the compatible $M_{c}$, one can use the tensor product decomposition of Hilbert space into the underlying $N$ dofs of Eq.\ (\ref{variational_microdecomposition}) to resolve each of these subspaces $\Hil^{(c)}_k$ by spanning them with the tensor product eigenbasis of $M_{c}$  which have the same eigenvalue labeled by $k$.

Thus, we now have a direct-sum structure to $\Hil$,
\begin{equation}
\Hil \cong \bigoplus_{k=1}^{N_C} \Hil^{(c)}_{k} \: ,
\end{equation}
where each direct-sum subspace is specified by the span of the \emph{tensor product eigenbasis} of $M_c$ with a given distinct eigenvalue, and this direct-sum structure satisfies $\sum_{k=1}^{N_{c}} h_{k} = \Dim\Hil$. 

Now that we have identified the column structure based on the degeneracy structure of the compatible collective observable $M_{c}$ and the basis elements which enter the BPT, we are left with the the task of assigning the row structure which will fix the BPT. Here, we have a discrete set of combinatoric choices of row alignments we can do, given the column structure and the specification of the basis elements of the compatible collective observable. For this purpose, we now turn to understanding the conditions under which the coarse-graining is quasi-classical.

\subsection{Superselection Sectors and Emergent Quasi-Classicality}

We can use the transition structure of the Hamiltonian in the tensor product eigenbasis of $M_c$ to further split our partial bipartition into direct-sum sectors which will act as superselection sectors in our coarse-graining scheme. Based on the Hamiltonian expressed in this tensor product $M_c$ basis, one can identify unions of column subspaces, $\tilde{\Hil}_{q} \cong \bigoplus_{k_{q}} \Hil^{(c)}_{k_{q}}$ for some $q = 1,2,\cdots,N_{\mrm{sectors}}$, for which the Hamiltonian has no tunneling terms connecting these subspace unions, such that for all states $\ket{\psi_{q}} \in \tilde{\Hil}_{q}$ and $\ket{\psi_{q'}} \in \tilde{\Hil}_{q'}$,
\begin{equation}
\label{eq:variational_superselection}
\langle \psi_{q} | \Ham | \psi_{q'} \rangle \: = \: 0,  \: \: \: \:  \mrm{if} \: \: q \neq q' \: .
\end{equation}
Each such union of sectors represents a superselection sector for our coarse-graining since these different unions do not interact. Each direct-sum sector will be arranged as a distinct block in a BPT of a partial bipartition in the compact form with the row structure yet to be identified based on a criterion of emergent quasi-classicality. Compact form within each sector will allow minimal loss of coherence under the state-reduction map induced by the BPT, so the coarse-graining we find will indeed reflect emergent quasi-classicality from dynamical decoherence and not the mere discarding of quantum coherences by misalignments between basis states in the structure of the BPT.

To fix the row structure within each direct-sum superselection sector, we now turn back to the question of emergent quasi-classicality. A feature of our coarse-graining will be that dynamics in the reduced space following the BPT state-reduction map will reflect features of classicality. We have already identified the column structure of our BPT which labels our macroscopic pointer states, based on compatibility of a collective observable with the Hamiltonian. The compact form of our BPTs ensure that minimal coherence between basis states is lost due to the action of the state-reduction map itself so we can now focus on the action of the Hamiltonian to induce \emph{dynamical} decoherence. In quasi-classical coarse-grainings, we expect the row alignment of the BPT to allow Hamiltonian evolution to decohere superpositions of our macroscopic pointer states by ``interaction" with $\mathcal{S}_{\mrm{internal}}$. For such quasi-classical BPTs, we can demand the rate at which this dynamical decoherence happens to be fastest and hence, most effective.

We will, for concreteness, focus on small time evolution since classical states, as opposed to non-classical ones, will exhibit decoherence starting at short time scales, and are expected to stay decohered as time progresses. We thus will quantify entanglement growth rate using the linear entanglement entropy\footnote{One could equally well use von~Neumann entanglement entropy too, of which the linear entropy forms the leading order contribution.}. Consider a pure state of the full Hilbert space $\rho(t) = \ket{\psi(t)}\bra{\psi(t)} \in \mathcal{L}(\Hil)$ evolving under evolution by the Hamiltonian $\Ham$, and a BPT which induces a state-reduction map $\Tr_{(R)}$ by tracing out the partial subsystem defined by its row subspace. Under this state-reduction map, the pure state $\rho$ gets mapped to
\begin{equation}
\rho_c(t) \: \equiv \: \Tr_{(R)}\rho(t) \: ,
\end{equation}
whose entanglement can be quantified by the linear entropy,
\begin{equation}
S_{\mrm{lin}}(t) = 1 - \Tr(\rho^{2}_{c}(t)) \: .
\end{equation}
It can be shown, as was done in \cite{mereology}, that for initially pure, \emph{unentangled} states, the linear entanglement entropy grows at $\mathcal{O}(t^2)$ to leading order and hence one can quantify the growth rate of entanglement entropy as,
\begin{equation}
\label{eq:Slindotdot}
\ddot{S}_{\mrm{lin}}(0) = -\Tr \left( \rho_{c}(0)\ddot{\rho}_{c}(0) + \ddot{\rho}_{c}(0) \rho_{c}(0) + 2\dot{\rho}_{c}(0) \right) \: ,
\end{equation}
where a dot over a quantity represents its time derivative and we have,
\begin{equation}
\dot{\rho}_{c}(0) = \Tr_{(R)}\left( \dot{\rho}(0)\right) =  \Tr_{(R)}\left( -i \left[\Ham,\rho(0) \right]\right) \: ,
\end{equation}
\begin{equation}
\ddot{\rho}_{c}(0) = \Tr_{(R)}\left( \ddot{\rho}(0)\right) =  \Tr_{(R)}\left( - \left[\Ham, \left[\Ham,\rho(0) \right]\right]\right) \: .
\end{equation}
We now proceed to use this quantifier $\ddot{S}_{\mrm{lin}}(0)$ of entanglement growth rate to quantify the classicality of a given BPT. The most natural initial states suited for a decoherence analysis offered by a partial bipartition are the ones supported on basis states of a single row in the BPT. 
In the familiar case of a tensor product structure corresponding to a rectangular BPT, such a state would correspond to an unentangled state -- a tensor product of a superposition for the state of the column subsystem $\Hil_{A}$ with a single basis state of the row system $\Hil_{B}$. For example, in the decoherence literature, considering a system and environment split $\Hil \cong \Hil_{S}\otimes\Hil_{E}$, one considers initial states of the form $\ket{\psi(0)} = \left(\sum_{s}c_{s}\ket{s}\right)\otimes\ket{E_{\mrm{ready}}}$ for some ready state on the environment. Thus states supported on a single row of a BPT of a partial bipartition are natural extensions of such initial pure states which are unentangled and therefore are good candidates to measure the dynamical decoherence of. Borrowing intuition and language from the decoherence paradigm, the state first branches, i.e.\ the environment (the row variable in our BPT) states evolve conditionally depending on the pointer state $\ket{s}$ of the system (corresponding to the columns of the BPT), following which there is dynamic decoherence where these conditional states of environment become orthogonal in time and stay so. The branching of such initial states happens at $\mathcal{O}(t)$ following which we expect these conditionally evolved states to decohere, which we can capture by the entanglement growth rate via Eq.\ (\ref{eq:Slindotdot}) which grows $\mathcal{O}(t^{2})$. Using this understanding, we propose a metric to quantify this dynamic decoherence as a probe of emergent quasi-classicality of a given BPT: One can construct, for each row $i = 1,2,\cdots,N_{R}$ of the BPT, a \emph{uniform} superposition state over all basis states in that row (each labelling a different macroscopic pointer state),
\begin{equation}
\ket{\phi_{i}} = \frac{1}{\sqrt{w_{i}}}\sum_{k \in \{1...w_i\}} \ket{e_{i,k}}  \: ,
\end{equation}
where $w_{i}$ is the number of basis elements in the $i$-th row of the BPT and $k$ iterates over all such basis elements. For each such uniform superposition state, defined on each row, we can compute the entanglement growth rate $\ddot{S}_{\mrm{lin, i}}(0)$ as a measure of dynamical decoherence and then quantify the emergent classicality of the BPT, $Q_{\mrm{BPT}}$ as the average of these entanglement growth rates over all rows of the BPT,
\begin{equation}
Q_\mrm{{BPT}} = \frac{1}{N_{R}}\sum_{i = 1}^{N_{R}} \ddot{S}_{\mrm{lin, i}}(0) \: .
\label{eq:Q_BPT}
\end{equation}
The average over all rows can be interpreted as a statistical mixture over different basis states of the partial system (label by the rows), representing a probabilistic treatment of not knowing the state of the partial system which will be traced over. This is one such metric which captures the idea of emergent classicality using a notion of effective dynamical decoherence. We adopt this as a demonstration of principle, but emphasize that, depending on the context of the coarse-graining being constructed, one can come up with more amenable definitions of quantities which capture the emergent quasi-classical nature of the BPT. Given this metric, one can now vary over all possible BPT row arrangements, which are discrete and finite choices of basis element arrangement within each column, and choose the BPT which maximizes $Q_{\mrm{BPT}}$ representing most effective dynamical decoherence and hence, is the most quasi-classical. It should be noted that this prescription may not always yield a unique preferred BPT reflecting a preferred underlying classical partial bipartition $\Hil \sim \mathcal{S}_{\mrm{collective}}\oslash \mathcal{S}_{\mrm{internal}}$, but rather will often select a class of BPTs which have the same classicality quantification based on the metric above. One can interpret this residual freedom as gauge choices of coarse-grainings, which, even though they induce different state-reduction maps based on the BPT, have the same measure of emergent classicality in the reduced subspace of the macroscopic variable based on the chosen metric. Often, this will be intimately tied with some symmetry structure in the Hamiltonian which does not distinguish between the underlying different degrees of freedom and hence leads to a class of BPTs with the same emergent quasi-classical behavior. 

The above algorithm just described for obtaining such collective, quasi-classical coarse-grainings can now be summarized as follows:
\begin{enumerate}
\item{Based on the given microfactorization of Hilbert space into degrees of freedom, find a collective observable $M_{c}$ which is most compatible with the Hamiltonian as given by Eq.\ (\ref{Mc_compatibility}).}
\item{Eigenspaces of $M_{c}$ corresponding to distinct eigenvalues will label different column subspaces of the BPT as macroscopic, collective pointer states robust under Hamitonian evolution. These will make up the partial subsystem $\mathcal{S}_{\mrm{collective}}$ which will be the coarse-graining of $\Hil$.}
\item{Eigenspaces of $M_{c}$ with distinct eigenvalues, along with the microfactorization, furnish an orthonormal basis for Hilbert space and resolve the column subspaces with tensor product basis elements with distinct eigenvalues.}
\item{Once the column structure of the BPT is fixed, use the transition structure of the Hamiltonian in this tensor product basis of $M_{c}$ to identify superselection sectors as done in Eq.\ (\ref{eq:variational_superselection}), each of which will form a disjoint block of the BPT. Each block will be arranged in the compact form to minimize loss of coherence due to the action of the state-reduction map induced by the BPT.}
\item{Now consider dynamical decoherence to fix the remaining freedom in each such block in the compact form to identify the alignment of the rows in the BPT. Iterate over the finite, discrete permutations of row arrangements and select (the class of) BPT(s) which maximize entanglement production as a measure of effective dynamical decoherence as done in Eq.\ (\ref{eq:Q_BPT}).}

\end{enumerate}
In the next section, using a concrete example of the Ising model, we will demonstrate this algorithm for constructing a quasi-classical coarse-graining based on a collective variable compatible with the Hamiltonian. 

\section{Example of the Variational Approach: Coarse-Graining the Ising Model as a Partial Bipartition}
\label{sec:ising}

Let us now consider a concrete example where we can apply the coarse-graining scheme developed above. We will focus on the Ising model in 1-D and see how we can capture collective features of the model which are effectively classical. Consider $N \geq 2$ spin-$\frac{1}{2}$ particles described by a tensor-product\footnote{For completeness, we mention that while one can study several dual pairs of lattice theories such as the Ising model \cite{Radicevic:2016tlt,Cotler:2017abq}, which differ by global decomposition changes of Hilbert space, in this paper we focus on a fixed micro-decomposition of the underlying degrees of freedom.} Hilbert space $\Hil \cong \bigotimes_{\mu=1}^{N}\Hil_{\mu}$ on a 1-D lattice evolving under the Ising Hamiltonian,
\begin{equation}
\label{isingHam}
\Ham  \sim - \sum_{\mu = 1}^{N-1} \left( {\sigma}^{(\mu)}_{z} \otimes {\sigma}^{(\mu + 1)}_{z} \right) - g\sum_{\mu = 1}^{N} {\sigma}^{(\mu)}_{x} \: ,
\end{equation}
where $g > 0$ characterizes the strength of the external magnetic field in the $x$-direction and the symbol $\sim$ in the definition of an operator implies that we will normalize it under the Frobenius norm. As usual, ${\sigma}^{(\mu)}_{z}$ is the Pauli z-operator on the $\mu$-th spin on the lattice and ${\sigma}^{(\mu)}_{x}$ is the Pauli x-operator. Note that our Ising Hamiltonian does \emph{not} have periodic boundary conditions and corresponds to an open chain with $N$ sites. We choose this specific boundary condition since the results in this case are more compact to describe and therefore help in the exposition of the idea. The same analysis could also be implemented for different boundary conditions and the results could be interpreted along similar lines. Our goal is now to look for the collective observable of the $N$ spins most compatible with the Hamiltonian. We model the operator $M_{\mu}$ in Eqs. (\ref{Mc_definition}) and (\ref{M_mu_operator}) as a unit-normed operator under the Frobenius norm,
\begin{equation}
\label{M_mu_ising}
M_{\mu} = \sqrt{\frac{2}{1 + \alpha^2} }\left( {\sigma}^{(\mu)}_{z} + \alpha  {\sigma}^{(\mu)}_{x} \right) \: \: ,\: \: \forall \: \: \mu \: \: ,
\end{equation}
characterized by the parameter $\alpha \geq 0$ quantifying the mix between Pauli x and z operators. As a simplifying assumption, since the Hamiltonian only contains Pauli-x and Pauli-z operators, we do not take support on $ {\sigma}^{(\mu)}_{y}$ in Eq. (\ref{M_mu_ising}) and only consider mixing between x and z to determine the most compatible collective observable. This assumption could be relaxed to perform a more complete analysis.
It should be noted that we take the operator $M_{\mu}$ characterized by the same parameter $\alpha$ for each spin $\mu$. Under the Hamiltonian, all but the edge spins are treated on an equal footing and are indistinguishable from the point of view of dynamics. Hence one should expect a similar parametrization, because in this variational approach, we only have access to the Hamiltonian, and any structure that emerges should respect the underlying symmetry of the dynamics. The edge spins, represented by $M_{1}$ and $M_{N}$, should in general be treated differently due to non-periodic boundary conditions and while this analysis can be carried out in a straightforward way, we choose to parametrize their contribution to $M_c$ by the same value of $\alpha$. This can be justified on two grounds. First, there are only 2 edge spins compared to $(N-2)$ bulk ones and for moderately sized chains and larger, any difference due to edge spins will be sub-dominant. Second, choosing the same parametrization for each spin will allow a more elegant understanding of the collective observable $M_{c}$ as an \emph{average} quantity over the spin chain and a clean interpretation of the coarse-graining scheme where the macroscopic variable will be labeled by distinct values of this average quantity. 

We can now compute the Frobenius norm of the commutator of $M_c$ and the Ising Hamiltonian $\Ham$ parametrized by $\alpha$, which gives us
\begin{equation}
\label{ising_commutator_norm}
\Bigl| \Bigl| \lcb \Ham, M_{c} \rcb \Bigr| \Bigr|_{f} = \frac{(N-1)\alpha^2 + 2Ng^2}{2^{N-3} \: N \: (4Ng^2 + N - 1) (1 + \alpha^2)} \: .
\end{equation}
Minimization of this norm above with respect to the parameter $\alpha$ will give us a collective observable most compatible with the Hamiltonian. To minimize this norm, we write it in a more suggestive way,
\begin{equation}
\Bigl| \Bigl| \lcb \Ham, M_{c} \rcb \Bigr| \Bigr|_{f} = \frac{1}{1 + (1/\alpha^2)} \left[\left(\frac{1}{2^{N-3} N}\right) \frac{(N-1)}{4Ng^2 + N - 1)} \right] +\frac{1}{1 + \alpha^2} \left[\left(\frac{1}{2^{N-3} N}\right) \frac{2Ng^2}{4Ng^2 + N - 1)} \right]  \: ,
\end{equation} 
where we have factored out the $g$ dependence in each term in the sum above, which we identify as,
\begin{equation}
\Bigl| \Bigl| \lcb \Ham, M_{c} \rcb \Bigr| \Bigr|_{f} = \left(\frac{1}{1 + (1/\alpha^2)}\right) T_{1}(g) + \left(\frac{1}{1 + \alpha^2}\right) T_{2}(g) \: .
\end{equation}
We notice a turning point in the $g$ dependence of $T_{1}(g)$ and $T_{2}(g)$. We find that this norm is minimized for the following condition depending on the value of $g$ which controls the relative importance of the two different terms in the normalized Ising Hamiltonian. For $g^{2} < (N-1)/2N$, we see that $T_{1}(g) > T_{2}(g)$ and hence to minimize the norm in Eq.\ (\ref{ising_commutator_norm}), the $\alpha$ dependent prefactor of $T_{1}(g)$ should be minimized which implies $\alpha = 0$. On the other hand, when  $g^{2} > (N-1)/2N$, we see that $T_{2}(g) > T_{1}(g)$ and hence $\alpha = \infty$ ensures minimization of the dominant term and hence the norm itself. One can confirm these results by formally differentiating, and checking for minima conditions in the relations above. Thus, we find that depending on the value of $g$ in the normalized Hamiltonian, the most compatible collective observable corresponds to,
\begin{equation}
\label{ising_min_condition}
\alpha(g) = \begin{cases}
0 \:, &    \: \: \: g < g_{\mrm{crit}}\\
\infty \:, &    \: \: \: g > g_{\mrm{crit}} \: \: ,\\
\end{cases}
\end{equation}
where $g_{\mrm{crit}}$ is the critical value\footnote{We have derived the value of $g_{\mrm{crit}}$ from the compatibility condition of Eq.\ (\ref{Mc_compatibility}); it should not be confused with the, in general different, value of $g$ where the phase transition in the Ising model takes place.} of $g$ given by
\begin{equation}
\label{gcrit}
g_{\mrm{crit}} \: = \: \sqrt{\frac{N-1}{2N}} \: ,
\end{equation}
such that for $g < g_{\mrm{crit}}$ when the Pauli z-z interaction term dominates, the most compatible collective observable $M_c$ is the average spin-z of the Ising chain (which corresponds to $\alpha =0$),
\begin{equation}
\label{Mc_Zcollective}
M_{c} \sim \sum_{\mu=1}^{N} {\sigma}^{(\mu)}_{z} \: \: , \: \: \: g < g_{\mrm{crit}} \: ,
\end{equation}
and for $g > g_{\mrm{crit}}$, when the external magnetic field in the $x$-direction dominates, the collective observable which is most compatible with the Hamiltonian is the average spin-x of the chain (which corresponds to $\alpha  = \infty$),
\begin{equation}
\label{Mc_Xcollective}
M_{c} \sim \sum_{\mu=1}^{N} {\sigma}^{(\mu)}_{x} \: \: , \: \: \: g > g_{\mrm{crit}} \: .
\end{equation}
Thus we see a phenomenon akin to a phase transition where depending on the dominant term in the Hamiltonian the most compatible $M_c$ is the one which is commuting with the dominant term. 

Once we have obtained the most compatible collective observable $M_c$, we can immediately use its distinct eigenvalues $m_k$ to label the macroscopic states of our coarse-graining by
\begin{equation}
m_{k} \in \bigg\{ -\frac{N}{2}, -\frac{N-2}{2} , \ldots ,\frac{N-2}{2}, \frac{N}{2} \bigg\} \: ,
\end{equation}
and these will be used to label distinct columns of the BPT which will specify the coarse-graining. Since the $N+1$ distinct values of $M_c$ will serve as labels of our macrostates under the coarse-graining prescription, our coarse-grained space will have a dimension $\Dim\mathcal{S}_{\mrm{collective}} = N_C$. Already we see a major benefit of our scheme in dimension reduction: our collective scheme will map our original Hilbert space of $\Dim\Hil = 2^N$ to a reduced, coarse-grained space with $\Dim\mathcal{S}_{\mrm{collective}} = (N+1)$. The dimension of the $k$-th collective (macroscopic) subspace corresponding to the $M_{c}$ eigenvalue $\bigl| m_{k} \bigr| = (N - 2(k-1))/2$ for $k = 1,2,\cdots,\left\lceil  \frac{N_C}{2} \right\rceil$ is then the binomial coefficient,
\begin{equation}
\Dim\Hil^{(c)}_{k} = {N\choose{k-1}} \: .
\end{equation}
Given that we are working with $N$ underlying spins specified by the tensor decomposition of Eq.\ (\ref{variational_microdecomposition}), we can use the natural tensor-product basis of $M_{c}$ to specify the orthonormal basis which we will be working with to fill the cells of our BPT. When $g < g_{\mrm{crit}}$, corresponding to $M_{c}$ being the average spin-z of the lattice, we use the $\{\ket{0},\ket{1} \}^{\otimes N}$ basis where $\{\ket{0},\ket{1} \}$ are the eigenstates of ${\sigma}_{z}$ and in the other case when $g > g_{\mrm{crit}}$ so that $M_c$ is the average spin-x of the chain, we can use the $\{\ket{+},\ket{-} \}^{\otimes N}$ where $\{\ket{+},\ket{-} \}$ are the eigenstates of ${\sigma}_{x}$. Thus we now have a \emph{fixed} orthornormal basis we will use to construct a BPT for the partial bipartition for the Ising model and a specification of different columns of the BPT labeled by distinct eigenvalues of the compatible collective observable $M_{c}$. The only freedom we now have is the choice of row alignments in our BPT for which we will turn to effective dynamical decoherence as a quantifier of quasi-classical behavior of the coarse-graining. We will take each disjoint block in our BPT corresponding to superselection sectors in the partial bipartition to be in compact form to make sure that minimal coherence is lost due to the action of the state reduction itself and any decoherence will be due to Hamiltonian evolution. As we will see, the two cases of $g > g_{\mrm{crit}}$ and $g < g_{\mrm{crit}}$ will have very different superselection properties based on the Hamiltonian, hence we will deal with them separately and describe the results for each case in detail.

\subsection{$g < g_{\mrm{crit}}$: Average spin-z as Collective Observable}

Let us first focus on the case when $g < g_{\mrm{crit}}$, so that the $z-z$ interaction term in the Ising Hamiltonian dominates which sets the most compatible collective observable $M_{c} \sim \sum_{\mu=1}^{N} {\sigma}^{(\mu)}_{z}$, the average spin-z of the Ising chain. As discussed before, this choice of $M_{c}$ offers us an orthonormal basis of $\{\ket{0},\ket{1} \}^{\otimes N}$ to work with in the BPT. The action of the Hamiltonian on these tensor product basis furnished from eigenstates of $M_{c}$ is to \emph{flip single bits} in the $\{\ket{0},\ket{1} \}$ basis due to the presence of the external magnetic field in the $x$-direction (in the case when $g > 0$). Due to this transition structure of the Hamiltonian in the $M_c$ basis, bit flips can successively connect each of the $N_C$ column subspaces and hence there is no superselection sector structure in this case. 

We illustrate the results for the case of $N = 3$ spins since there the results are tractable and easy to follow to demonstrate the physics behind them\footnote{While for such small number of spins, one might want to treat the edge spins on a different footing than the bulk ones since the edge contribution may not be sub-dominant, we take the same parametrization for each spin as in Eq. (\ref{M_mu_ising}) as a demonstration of principle with the $M_c$ being an average quantity over the entire chain.}. For $N = 3$, the compatible collective observable $M_{c}$ has $N_{c} = 4$ distinct eigenvalues which will label different columns of our BPT. In Table \ref{tab:N3_Z_eigen}, we list out these eigenstates of $M_{c}$ arranged in columns by their distinct eigenvalues. This Table \ref{tab:N3_Z_eigen} is not yet a BPT since we haven't yet considered row alignments, just a listing of eigenstates arranged by columns labeled by distinct eigenvalues of the compatible $M_{c}$. 
\begin{table}[h]
\begin{center}
\begin{tabular}{l*{4}{c}}
$M_{c}$ value:              & $-\frac{3}{2}$ & $-\frac{1}{2}$ & $\frac{1}{2}$ & $\frac{3}{2}$  \\  [0.2cm]
\hline
\hline
 & $\ket{000}$ & $\ket{100}$ & $\ket{011}$ & $\ket{111}$  \\ 
 &  & $\ket{010}$ & $\ket{101}$ &   \\
 &  & $\ket{001}$ & $\ket{110}$ &   \\
\end{tabular}
\caption{Tensor Product Eigenstates for $M_{c} = \sum_{\mu=1}^{3} {\sigma}^{(\mu)}_{z}$ for $N=3$ spins, arranged in columns labeled by distinct eigenvalues. Note that this is \emph{not} a BPT, just an enumeration of the eigenstates arranged by the column structure governed by the compatible $M_{c}$.}
\label{tab:N3_Z_eigen}
\end{center}
\end{table}
We will now consider different BPTs in compact form by iterating over different row arrangements of the eigenstates within each of the columns fixed by the collective observable $M_{c}$. By this token, there will be a total of $\prod_{k =1}^{N_{c}} \left(h_{k}! \right)$ number of permutations of row arrangements which will be the set of BPTs we will consider. There will, of course, be many redundancies in this way of enumerating different BPTs (such as inconsequential rearrangements differing by row swaps) compatible with the $M_{c}$ column structure but we iterate over them anyway to keep the permutations easy to track. In the case of $N=3$ spins, we will have a total of $36$ BPTs to iterate over and for each such BPT, we compute $Q_{\mrm{BPT}}$, the average entanglement growth rate over pure, uniform states defined on each row as defined in Eq.\ (\ref{eq:Q_BPT}),  and choose the class of BPTs which maximize this quantifier, representing effective dynamical decoherence as the most classical and compatible coarse-graining given the Hamiltonian.
\begin{figure}[h]
\includegraphics[width = 1\textwidth]{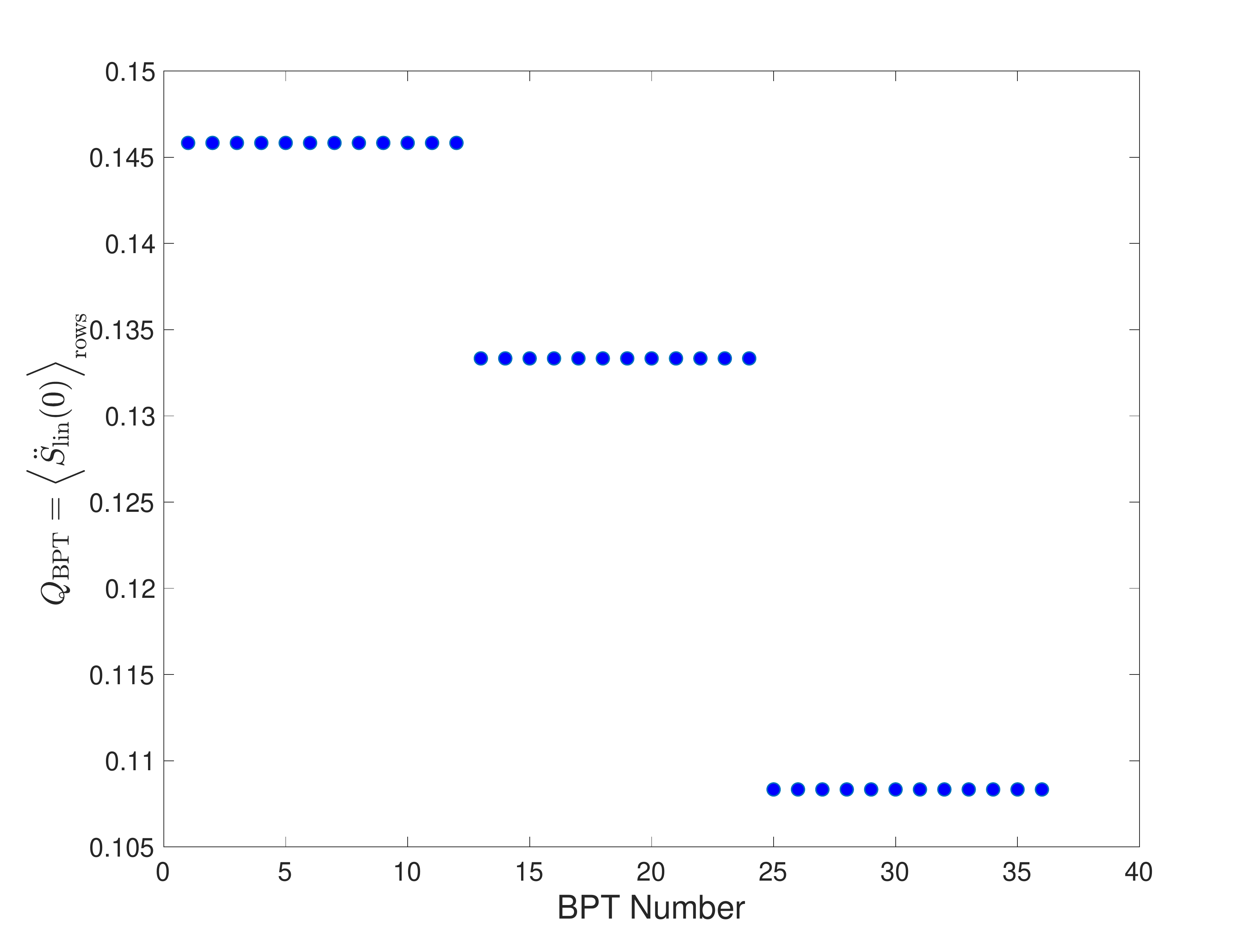}
\caption{Plot of average entanglement growth rate $Q_{\mrm{BPT}}$ over different BPTs (different row arrangements) for $N=3$ spins with the compatible collective observable $M_{c} = \sum_{\mu=1}^{3}{\sigma}^{(\mu)}_{z}$ corresponding to a value of $g = 0.5 < g_{\mrm{crit}}$.}
\label{fig:N3_Z_entropy}
\end{figure}
In Figure \ref{fig:N3_Z_entropy}, we plot the average entanglement growth rate $Q_{\mrm{BPT}}$ for these $36$ BPTs for the case of $N=3$ spins with the collective observable being $M_{c} \sim \sum_{\mu} {\sigma}^{(\mu)}_{z}$. We notice that these BPTs come in three distinct classes differentiated by entanglement growth rates. The class of BPTs with the maximum entanglement growth rate is selected as the most quasi-classical one and we find there are six such distinct BPTs belonging to this selected quasi-classical equivalence class. (While the plot in figure \ref{fig:N3_Z_entropy} shows $12$ such BPTs with the largest value of $Q_{\mrm{BPT}}$, as mentioned there are redundancies in our enumeration and only $6$ of them are distinct from the perspective of the state-reduction map they induce.)  In Table \ref{tab:N3_Z_selected}, we display these $6$ selected, quasi-classical BPTs.

\begin{figure}[h]
	\includegraphics[width = 1\textwidth]{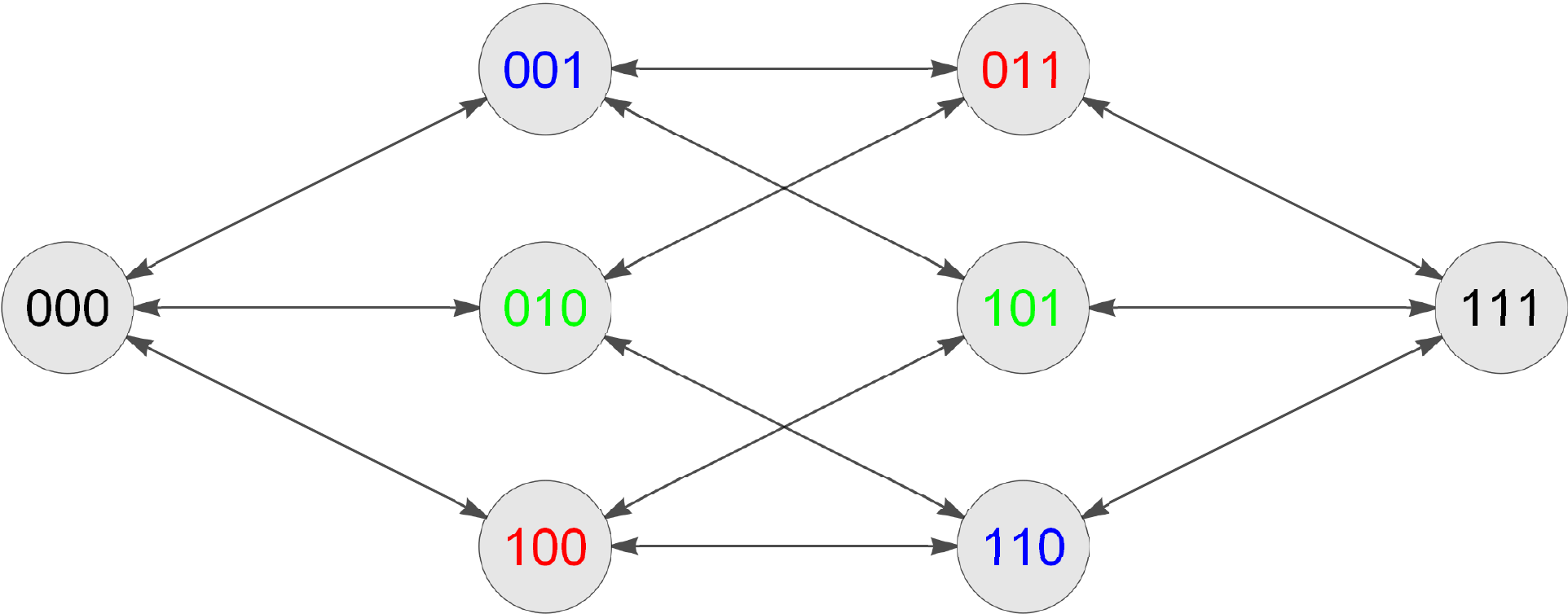}
	\caption{Transition structure of the Hamiltonian in the tensor product basis of $M_{c} = \sum_{\mu=1}^{3} {\sigma}^{(\mu)}_{z}$ for $N=3$ spins. It should be noted that this is not a BPT representation but only illustrates the transition structure of the Hamiltonian in the chosen basis.}
	\label{fig:N3_Z_graph}
\end{figure}


\begin{table}[h]
\begin{center}
\begin{tabular}{l c r}

\begin{tabular}{c|c|c|c}
\hline 
\multicolumn{1}{|c|}{$000$}& \textcolor{blue}{$001$} & \multicolumn{1}{c|}{\textcolor{red}{$\mathbf{011}$}}& \multicolumn{1}{c|}{$111$}\tabularnewline
\hline 
 & \textcolor{green}{$010$} & \multicolumn{1}{|c|}{\textcolor{blue}{$\mathbf{110}$}} \tabularnewline
\cline{2-3} 
& \multicolumn{1}{|c|}{\textcolor{red}{$100$}} & \multicolumn{1}{|c|}{\textcolor{green}{$\mathbf{101}$}}  \tabularnewline
\cline{2-3} 
\end{tabular} 
\hspace{1cm}

\begin{tabular}{c|c|c|c}
\hline 
\multicolumn{1}{|c|}{$000$}& \textcolor{green}{$010$} & \multicolumn{1}{c|}{\textcolor{red}{$\mathbf{011}$}}& \multicolumn{1}{c|}{$111$}\tabularnewline
\hline 
 & \textcolor{blue}{$001$} & \multicolumn{1}{|c|}{\textcolor{green}{$\mathbf{101}$}} \tabularnewline
\cline{2-3} 
& \multicolumn{1}{|c|}{\textcolor{red}{$100$}} & \multicolumn{1}{|c|}{\textcolor{blue}{$\mathbf{110}$}}  \tabularnewline
\cline{2-3} 
\end{tabular}  
\hspace{1cm}

\begin{tabular}{c|c|c|c}
\hline 
\multicolumn{1}{|c|}{$000$}& \textcolor{red}{$100$} & \multicolumn{1}{c|}{\textcolor{green}{$\mathbf{101}$}}& \multicolumn{1}{c|}{$111$}\tabularnewline
\hline 
 & \textcolor{blue}{$001$} & \multicolumn{1}{|c|}{\textcolor{red}{$\mathbf{011}$}} \tabularnewline
\cline{2-3} 
& \multicolumn{1}{|c|}{\textcolor{green}{$010$}} & \multicolumn{1}{|c|}{\textcolor{blue}{$\mathbf{110}$}}  \tabularnewline
\cline{2-3} 
\end{tabular}

 \tabularnewline

\\
\begin{tabular}{c|c|c|c}
\hline 
\multicolumn{1}{|c|}{$000$}& \textcolor{blue}{$001$} & \multicolumn{1}{c|}{\textcolor{green}{$\mathbf{101}$}}& \multicolumn{1}{c|}{$111$}\tabularnewline
\hline 
 & \textcolor{green}{$010$} & \multicolumn{1}{|c|}{\textcolor{red}{$\mathbf{011}$}} \tabularnewline
\cline{2-3} 
& \multicolumn{1}{|c|}{\textcolor{red}{$100$}} & \multicolumn{1}{|c|}{\textcolor{blue}{$\mathbf{110}$}}  \tabularnewline
\cline{2-3} 
\end{tabular} 
 \hspace{1cm}
 
\begin{tabular}{c|c|c|c}
\hline 
\multicolumn{1}{|c|}{$000$}& \textcolor{green}{$010$} & \multicolumn{1}{c|}{\textcolor{blue}{$\mathbf{110}$}}& \multicolumn{1}{c|}{$111$}\tabularnewline
\hline 
 & \textcolor{blue}{$001$} & \multicolumn{1}{|c|}{\textcolor{red}{$\mathbf{011}$}} \tabularnewline
\cline{2-3} 
& \multicolumn{1}{|c|}{\textcolor{red}{$100$}} & \multicolumn{1}{|c|}{\textcolor{green}{$\mathbf{101}$}}  \tabularnewline
\cline{2-3} 
\end{tabular} 

\hspace{1cm}
\begin{tabular}{c|c|c|c}
\hline 
\multicolumn{1}{|c|}{$000$}& \textcolor{red}{$100$} & \multicolumn{1}{c|}{\textcolor{blue}{$\mathbf{110}$}}& \multicolumn{1}{c|}{$111$}\tabularnewline
\hline 
 & \textcolor{blue}{$001$} & \multicolumn{1}{|c|}{\textcolor{green}{$\mathbf{101}$}} \tabularnewline
\cline{2-3} 
& \multicolumn{1}{|c|}{\textcolor{green}{$010$}} & \multicolumn{1}{|c|}{\textcolor{red}{$\mathbf{011}$}}  \tabularnewline
\cline{2-3} 
\end{tabular} 
\end{tabular}
\par\end{center} 
\caption{(color online) The $6$ selected, quasi-classical BPTs which maximize $Q_{\mrm{BPT}}$ as a measure of dynamical coherence for $N=3$ spins corresponding to the compatible collective observable $M_{c} = \sum_{\mu=1}^{3} {\sigma}^{(\mu)}_{z}$. Allowed transitions by the Hamiltonian flip single bits in the $\{\ket{0},\ket{1} \}$ basis. States in the middle two columns not connected by Hamiltonian transitions are shown by the same color.}
\label{tab:N3_Z_selected}
\end{table}

These selected BPTs have a common transition structure given by the Hamiltonian, which we portray in Fig. \ref{fig:N3_Z_graph} to better understand these results. The arrows depict transitions between different basis states given by the Hamiltonian. It can be seen that the selected quasi-classical BPTs are ones which induce transitions under the Hamiltonian by spreading maximally across rows of the BPT. This way, maximum coherence is lost for pure states supported on one row, leading to decoherence of different macroscopic pointer states.

Similar to the $N=3$ case, one can run an analysis on $N = 4$ spins in which case there will be $N_{C} = 5$ columns labeled by distinct eigenvalues of $M_{c}$. The results we find are very similar to the $N=3$ spin case. The selected BPTs have row alignments for which the Hamiltonian transitions maximize dynamical decoherence between different macroscopic pointer states under the state-reduction map induced by the BPT. In table \ref{tab:N4_Z_selected_instance}, we show one instance of the class of selected BPTs with the largest $Q_{\mrm{BPT}}$. Given that the Hamiltonian again induces single bit flips, we  see that this BPT has a transition structure to maximize dynamical loss of coherence. In figure \ref{fig:N4_Z_entropy}, we plot the average entanglement growth rate $Q_{\mrm{BPT}}$ for all the different row alignments possible given the column structure fixed by the collective $M_{c}$. As with the $N=3$ case, we see different classes of BPTs emerge which correspond to different entanglement growths. It is interesting to note how distinct the first few classes with the largest entanglement growth rates are which correspond to quasi-classical behavior (as shown in the inset in figure \ref{fig:N4_Z_entropy}), in contrast with generic permutations where the entanglement growth varies in a more smooth fashion, representing the generic nature of typical BPTs being away from quasi-classicality. 
\begin{table}[h!]
\begin{center}
\begin{tabular}{c c|c|cc}
\hline 
\multicolumn{1}{|c|}{$0000$}&   \multicolumn{1}{c|}{$1000$} & \multicolumn{1}{c|}{$1010$}& \multicolumn{1}{c|}{$1101$} & \multicolumn{1}{c|}{$1111$} \tabularnewline
\hline 
&   \multicolumn{1}{|c|}{$0100$} & \multicolumn{1}{c|}{$0101$}& \multicolumn{1}{|c|}{$1110$} &  \tabularnewline
\cline{2-4} 
&   \multicolumn{1}{|c|}{$0010$} & \multicolumn{1}{c|}{$0110$}& \multicolumn{1}{|c|}{$1011$} &  \tabularnewline
\cline{2-4} 
&   \multicolumn{1}{|c|}{$0001$} & \multicolumn{1}{c|}{$1001$}& \multicolumn{1}{|c|}{$0111$} &  \tabularnewline
\cline{2-4} 
&    & \multicolumn{1}{c|}{$0011$}&  &  \tabularnewline
\cline{3-3} 
&    & \multicolumn{1}{c|}{$1100$}&  &  \tabularnewline
\cline{3-3} 
\end{tabular} 
\caption{One instance of the class of selected quasi-classical BPTs for $N = 4$ spins corresponding to the compatible collective observable $M_{c} = \sum_{\mu=1}^{3} {\sigma}^{(\mu)}_{z}$. Allowed transitions by the Hamiltonian flip single bits in the $\{\ket{0},\ket{1} \}$ basis.}
\label{tab:N4_Z_selected_instance}
\end{center}

\end{table}

\begin{figure}[h!]
\includegraphics[width = 1\textwidth]{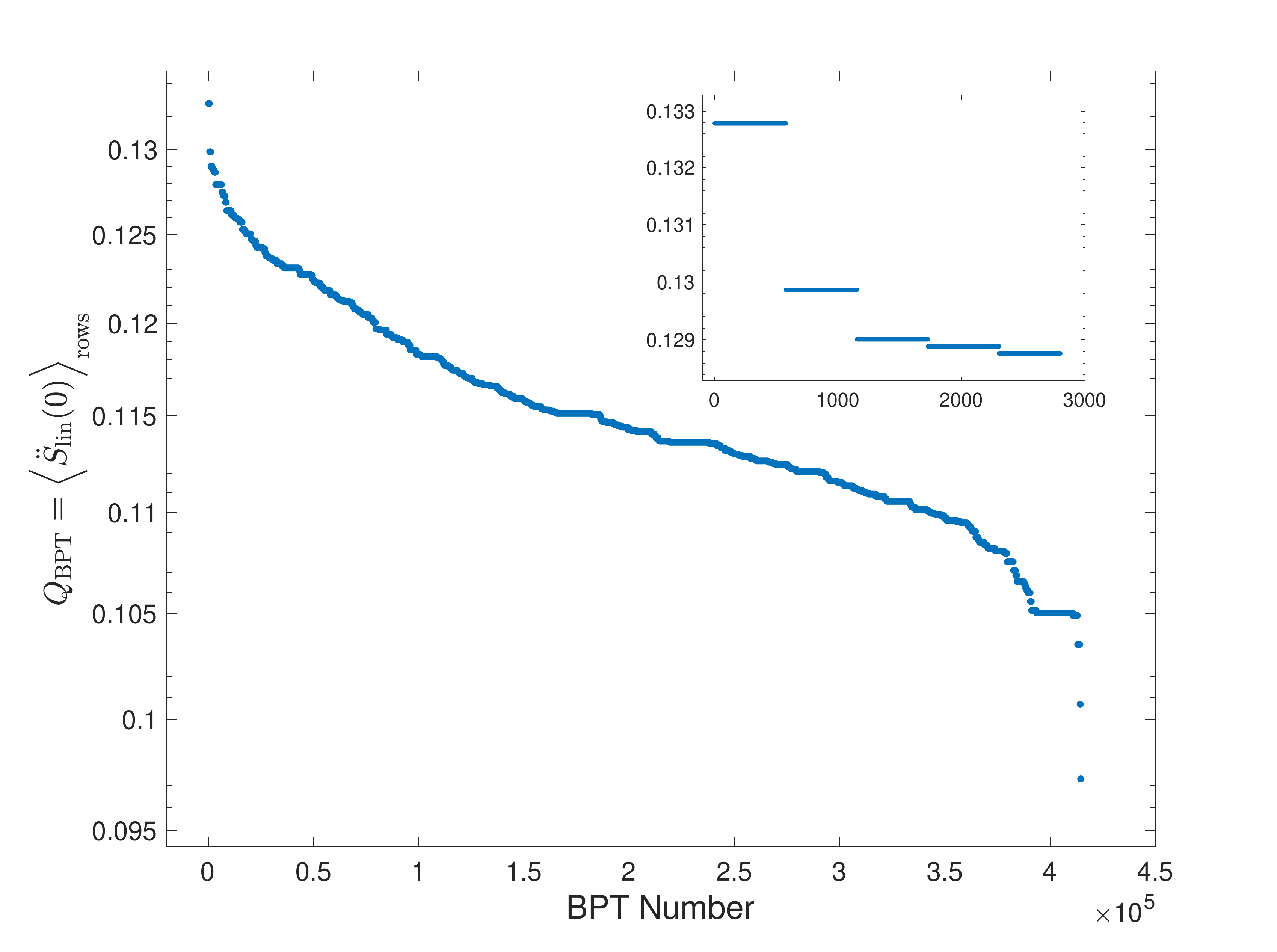}
\caption{Plot of average entanglement growth rate $Q_{\mrm{BPT}}$ over different BPTs (different row arrangements) for $N=4$ spins with the compatible collective observable $M_{c} = \sum_{\mu=1}^{3}{\sigma}^{(\mu)}_{z}$ corresponding to a value of $ g = 0.6 < g_{\mrm{crit}}$. The inset shows the first few classes of BPTs with lowest values of $Q_{\mrm{BPT}}$.}
\label{fig:N4_Z_entropy}
\end{figure}

\subsection{$g > g_{\mrm{crit}}$: Average spin-x as Collective Observable}

In the other case when $g > g_{\mrm{crit}}$, the external magnetic field term along the $x$-direction in the Ising Hamiltonian dominatesm which sets the most compatible collective observable to be $M_{c} \sim \sum_{\mu=1}^{N} {\sigma}^{(\mu)}_{x}$, the average spin-x of the Ising chain. As discussed before, this choice of $M_{c}$ offers us an orthonormal basis of $\{\ket{+},\ket{-} \}^{\otimes N}$ to work with in the BPT. The action of the Hamiltonian on these tensor product basis furnished from eigenstates of $M_{c}$ is to \emph{flip two adjacent bits} in the $\{\ket{+},\ket{-} \}$ basis due to the presence of the $z-z$ interaction term. Due to this transition structure of the Hamiltonian in the $M_c$ basis, bit flips of two adjacent spins \emph{cannot} successively connect each of the $N_C$ column subspaces and hence there will be superselection sectors in this case. Based on this transition structure of the Hamiltonian, we can split the $M_{c} \sim \sum_{\mu} {\sigma}^{(\mu)}_{x}$ basis into superselection sectors and iterate over row arrangements in each sector to maximize dynamical decoherence (by maximizing the average entanglement growth rate $Q_{\mrm{BPT}}$) to find the most compatible, quasi-classical coarse-graining. In Table \ref{tab:N3_X_selected}, we show the unique selected quasi-classical BPT. The selected BPTs again have the same feature that the transitions by the Hamiltonian are such that there is maximum dynamical decoherence under the state-reduction map induced by the BPT. A detailed analysis of this case of $g > g_{\mrm{crit}}$ can be done as was done for the $g < g_{\mrm{crit}}$ case by studying the variation of $Q_{\mrm{BPT}}$ for these BPTs and the Hamiltonian transition structure, but we keep the discussion here brief since the results follow the same physics as described in the previous subsection.  We see that depending on the nature of the Hamiltonian, different coarse-grained features can emerge as the ones which qualify as classical. Underlying symmetries of the Hamiltonian are reflected in the class of coarse-grainings which get picked out and reinforce the role played by dynamics in determining the set of quasi-classical variables of a system.

\begin{table}
\begin{center}
\begin{tabular}{lr}

\begin{tabular}{p{0.4\linewidth}}
\begin{tabular}{c|c|}
\hline 
\multicolumn{1}{|c|}{$---$}& $+-+$ \tabularnewline
\hline 
 & $++-$ \tabularnewline
\cline{2-2} 
& \multicolumn{1}{c|}{$-++$} \tabularnewline
\cline{2-2} 
\end{tabular}  \tabularnewline
 \hspace{82pt} %
\begin{tabular}{c|c}
\hline 
\multicolumn{1}{|c|}{$-+-$}&  \multicolumn{1}{c|}{$+++$} \tabularnewline
\hline 
 \multicolumn{1}{|c|}{$+--$} & \tabularnewline
\cline{1-1} 
 \multicolumn{1}{|c|}{$--+$} & \tabularnewline
\cline{1-1} 
\end{tabular} 
\end{tabular}

\end{tabular}

\caption{The selected, unique quasi-classical BPTs with minimum entanglement growth rate for $N=3$ spins corresponding to the compatible collective observable $M_{c} = \sum_{\mu=1}^{3} {\sigma}^{(\mu)}_{x}$. Allowed transitions by the Hamiltonian flip two adjacent bits in the $\{\ket{+},\ket{-} \}$ basis which induce the superselection sectors.}
\label{tab:N3_X_selected}
\end{center}
\end{table}
\vspace{0.7cm}

We would emphasize that a number of assumptions, albeit physically motivated, went into the formulation of this algorithm and for completeness, we enumerate them here to remind the reader of the context we are focusing on. First, we are working with a fixed microfactorization of Hilbert space into a collection of degrees of freedom which we wish to coarse-grain. We further take this access to the microfactorization to furnish us a tensor product basis for Hilbert space, in particular for the collective observable $M_{c}$. Once we have picked out $M_c$ by the compatibility condition with the Hamiltonian of Eq.\ (\ref{Mc_compatibility}), we focus only on compact BPTs since they minimize loss of coherence due to the action of the state-reduction map itself so we can study the dynamical decoherence which leads to classicality. One could work with more general bipartitions by allowing unitary change of basis which mixes between degrees of freedom and it would be interesting to develop an algorithm, akin to the one of Section \ref{sec:Irrep-Decompositions-of}, to construct partial BPTs based on access to a restricted set of observables which do not span an algebra. Such an algorithm would generalize considerations to non-compact form as well allow for superselection sectors governed by the specifying set of measurements.

To measure the dynamical decoherence induced by the internal subspace $\mathcal{S}_{\mrm{internal}}$, we used linear entanglement entropy for small times to measure how fast decoherence happens since non-classical states are expected to not decohere as fast on short timescales.  Our choice of initial states were uniform superposition states supported on a single row of the BPT, which offered a natural generalization of initial, unentangled states between the system and a ready state for the environment. One can imagine relaxing these assumptions to develop a more generic framework by studying a broader class of initial states, which would reflect more freedom in the ready state of $\mathcal{S}_{\mrm{internal}}$, or the type of superpositions in $\mathcal{S}_{\mrm{collective}}$ best suited to physical situations where decoherence is expected to be important. While our choice of linear entanglement entropy was for ease of mathematical manipulation, different measures of decoherence and entanglement such as von Neumann entropy could also be used. One can study the long time behavior where it is expected dynamical decoherence will have picked out the classical pointer basis where the reduced density matrix becomes diagonal and stays so. While more detailed, we expect the basic underlying physics to still be similar to the results described in this paper. 

 We would also briefly recall some features of the numerics which have gone into figures \ref{fig:N3_Z_entropy} and \ref{fig:N4_Z_entropy}. Recall that we are working with a normalized Hamiltonian so as to be able to tune the value of the interaction parameter $g$ which sets the strength of the external magnetic field in the $x$-direction to be able to toggle between a Hamiltonian with only the $z-z$ interaction ($g=0$) between neighboring spins to a Hamiltonian with only the external field ($g = \infty$). While this normalization is for us to get a better handle between the interplay of $\Ham$ and $M_{c}$, it affects the rate of entropy production and hence our results here are a proof of principle. Our measure of quasi-classicality of maximizing $Q_{\mrm{BPT}}$ is a suggestive quantifier that captures the qualitative idea that superpositions of classical macroscopic pointer states decohere effectively under evolution. We have focused on small time evolution for concreteness (by probing $\ddot{S}_{\mrm{lin}}(0)$) since we expect classical states under the quasi-classical BPT will start decohering rather quickly, unlike non-classical ones. The first few classes of BPTs with largest $Q_{\mrm{BPT}}$ are the most quasi-classical compared to generic BPTs. We emphasize that our measure of Eq.\ (\ref{eq:Q_BPT}) is a zeroth-order attempt to capture the broad idea of quasi-classicality and depending on the exact application one wishes to have, this quantifier can be more suitably chosen to yield more precise and richer quasi-classical coarse-grainings. The examples illustrated here were for a small number $N$ of spins for ease of tractability of results and we expect these results will become sharper as one goes to higher dimensions, since decoherence is typically aided by having large dimensions of the internal ``environment'' being traced over.
\section{Applications and Future Work}\label{sec:discussion}

In this section we discuss some of the potential applications of generalized and partial bipartitions to extant problems in the literature.

\subsection{Quantum Information Encoding}

It has long been recognized that the irreducible representation (irrep) structure of an operator algebra plays an important role in quantum information. In particular, in the theory of quantum error correction, the generalized bipartition structure is recognized as the fundamental structure behind all quantum error correcting codes \cite{blume2008characterizing,kribs2005unified,knill2000theory}. Noiseless subsystems, for example, are identified by the generalized bipartition  associated with the commutant algebra generated by the errors \cite{lidar2014review}. Subsystem codes \cite{kribs2005unified,kribs2005operator,bacon2006operator}, which generalize the idea of noiseless subsystems, are identified by a generalized bipartition usually associated with a non-abelian group (which also generalizes the construction of stabilizer codes that are associated with abelian groups \cite{poulin2005stabilizer,johnston2009operator}). Similarly, the idea of quantum state compression with respect to a preferred set of observables \cite{2018AnHP...19.1891B} relies on the generalized bipartition associated with the algebra of preferred observables; it is conceptually equivalent to the notion of quantum state reductions from a restricted algebra of observables that we discussed in Section \ref{subsec: State Reduction from Operational Constraints}. In such applications, the problem of identifying the generalized bipartition associated with the relevant algebra is fundamental. In the cases where the relevant algebra is given by a group with a well understood irrep structure, the generalized bipartition is clear. In all other cases, however, the algorithm presented in Section \ref{sec:Irrep-Decompositions-of} can be used as an analytical tool to identify the generalized bipartition.

We may also consider the more general problem of characterizing how the evolution given by a Hamiltonian or a channel acting on the physical system affects the logical degrees of freedom. Such problems are traditionally addressed by looking for symmetries of dynamics that identify the generalized bipartition (i.e. the irrep decomposition) with respect to which the dynamics are restricted to distinct irrep sectors \cite{2018PhRvA..97e2130K}. The main difficulty with this approach is of course in identifying ``useful'' symmetries. An alternative approach would be to identify an algebra that contains the operator(s) of dynamics directly, without appealing to symmetries. The action of quantum channels, for example, can be restricted to the irrep sectors of the algebra generated by their Kraus operators \cite{zanardi2000stabilizing}. Similarly, when dealing with Hamiltonians, even if we cannot find the irrep structure of the algebra generated by the Hamiltonian itself (a task that is equivalent to diagonalizing it), we can consider an irrep structure of some larger algebra that contains the Hamiltonian. This is, in fact, what we achieve by identifying a symmetry: the commutant algebra of the symmetry group is an algebra that contains the Hamiltonian which allows us to restrict its action to the irreps of the group. There are other ways, however, besides symmetries, to identify an algebra that contains the Hamiltonian. For example, if the Hamiltonian is a sum of multiple terms then it belongs to the algebra generated by those terms.  In particular, given a parameterized Hamiltonians as a sum of ``tunable'' terms whose strength is set by some natural or experimental constraints, the dynamics can be restricted to the irreps of the algebra generated by the tunable terms, independent of the parameters. A prime example of such scenario is the tunable exchange interaction in the Heisenberg spin-$\frac{1}{2}$ chain that implements qubit operations \cite{divincenzo2000universal}. It would be interesting to see if the algorithm of Section \ref{sec:Irrep-Decompositions-of} can address such problems, especially when the standard symmetry considerations fall short. 
\subsection{Bulk Reconstruction}

The AdS/CFT correspondence \cite{Maldacena:1997re,Gubser:1998bc,Witten:1998qj,Aharony:1999ti} equates the partition function, and thus the Hilbert space, of string theory or M-theory on negatively curved backgrounds and superconformal field theories. 
In the large $N$ limit, the relation describes a duality between classical (super)gravity in $D+1$ dimensions with fixed small and negative cosmological constant and a particular sector of (super)conformal field theories in $D$ dimensions with fixed large central charge.
So, in this limit, the correspondence becomes a holographic one in which we can use computations in a CFT living on the boundary of an appropriate spacetime to tell us about gravitational quantities in the bulk of the spacetime, and vice versa.
In many cases we would prefer to treat the bulk as a fixed solution to Einstein's equations sourced by quantum fields---that is, to consider only energy regimes and sets of observables which do not probe stringy or quantum-gravitational degrees of freedom in the bulk.
In the language we have used throughout the paper, it is thus natural to think of the classical states as living in a coarse-grained Hilbert space obtained by tracking only a restricted set of observables, namely (low-point) correlation functions of light bulk fields.
We can then apply the holographic duality and ask what the coarse-grained Hilbert space looks like from the perspective of the CFT.
In particular, we can ask what the holographic duals of classical bulk observables are, or how classical information about the bulk can be ``reconstructed'' from the CFT state.

In recent years a holographic error-correcting code approach to bulk reconstruction has been developed along these lines \cite{Almheiri:2014lwa,Harlow:2016vwg,Dong:2016eik,Cotler:2017erl}.
When the bulk dual of a CFT state is captured by a single bulk (Lorentzian) geometry, causality dictates that we should be able to recover all of the information inside a region by considering only its past domain of dependence. 
Hence we don't require knowledge of the entire coarse-grained CFT boundary state to reconstruct a local correlation function at a particular point in the bulk, but only some smaller region of the boundary at an earlier time.
(We can't directly associate a state to this region of the boundary, since the CFT does not factorize spatially, but we can instead consider the subalgebra of observables supported in the region.)
Because multiple possible boundary subregions can be used to redundantly reconstruct the same point in the bulk, the appropriate quantum-mechanical description of the bulk information contained in a given holographic CFT state is a \emph{complementary} error-correcting code, which can be divided into small code subspaces each of which can be used to reconstruct the appropriate bulk observables.

The methods of this paper apply directly to bulk reconstruction, at least when an appropriate UV cutoff or latticization is provided to render the system finite-dimensional.
It would be very interesting to directly construct the generalized bipartition for the classical observables in an explicit tensor network model (see e.g. \cite{Bao:2018pvs}). 
To probe the complementary nature of the resulting reduced state, we could consider, for example, first restricting to all classical observables, then further reducing to the state given only by the observables supported inside a particular lightcone. 
It would also be interesting to use our state-reduction methods to explicitly construct the set of holographic states by considering both a restricted set of classical bulk observables and a restricted set of boundary observables, which in general we expect to yield two different state-reduction maps, and enumerating the set of states for which their action is identical.

\subsection{Edge Modes and Gauge Symmetries}

In perturbative quantum field theory\footnote{In this motivational description of quantum field theory we are ignoring many subtleties such as normalization, renormalization, unitary inequivalence, convergence of the perturbative expansion, well-definedness of the theory, loop corrections, IR issues, etc., etc. We invite the reader to consult their favorite QFT textbook and/or keep in mind a lattice regularization which explicitly fixes the Hilbert space of the theory.}, we start from the free-field Hilbert space, which is constructed via a mode expansion in which the degrees of freedom are oscillators with given frequencies.
One basis for the Hilbert space is the field-value basis, in which each mode has a definite occupation number.
However, this picture runs into difficulties when the theory has (gauge or global) symmetries--that is, constraints, for example a Gauss law, on the allowed set of states in the ``gauged'' or ``physical'' Hilbert space. 
On the level of the mode expansion, these constraints prevent us from treating each mode as independent, meaning that the physical Hilbert space may not factorize into modes at all, and in particular that we might not be able to construct a reduced state by tracing out degrees of freedom in a gauge-invariant way.

As a toy model, for example, we can consider a lattice of 3 qubits with a $\mathbb{Z}_2$ symmetry, in which we identify a given state with the reversed state created by flipping the spin of each qubit simultaneously across some axis of the Bloch sphere. 
Without this global symmetry, the Hilbert space is isomorphic to $(\mathbb{C}^2)^{\otimes 3}$, an 8-dimensional Hilbert space which manifestly factorizes into three pieces.
However, imposing the symmetry reduces the Hilbert space to a 4-dimensional one in which we can no longer precisely identify individual qubits.
On the level of the abstract Hilbert space, to be sure, there was no need to talk about the larger 8-dimensional space at all---we could just have started directly with the 4-dimensional physical Hilbert space.

Although it is not justified from the physical Hilbert space alone, we nevertheless often have in mind a particular ``ungauged'' Hilbert space that does have nice factorization properties. 
Then we would like to be able to sensibly construct a reduced state even when the theory has an obstructing symmetry, such as the state of a gauge theory or conformal field theory on an interval, or the state of a spatial subregion in a diffeomorphism-invariant theory like general relativity.
Such a construction is provided by the edge modes program \cite{Donnelly:2014gva,Donnelly:2014fua,Donnelly:2015hta,Lin:2018bud} (see also \cite{Casini:2019kex}). 
On the quantum-mechanical level one looks for an embedding of the physical Hilbert space, which need not factorize, into a larger Hilbert space with some desired factorization properties, such as the existence of spatial intervals. 
Given this embedding we can map the original state to a state in the larger Hilbert space and then reduce in the usual way. 
The choice of embedding is not unique, but the edge modes program provides a particularly symmetric choice of embedding which corresponds to summing over all possible representations of matter charged under the symmetry, the eponymous edge modes.

From this description it should be clear that our approach is complementary.
The edge modes approach starts with a ``small,'' physical Hilbert space, chooses a ``large,'' auxilliary Hilbert space to embed into, and then constructs the reduced states by applying the appropriate partial-trace map on this large Hilbert space.
The generalized bipartition approach starts directly with a choice of operators specifying the allowed subregions, and provides a state-reduction map, not necessarily the partial-trace map, which produces the reduced states.
If we take the approach of Footnote \ref{fn:embedding} above and think of the generalized bipartition as a diagonal embedding into a larger bipartite Hilbert space, our approach naturally produces the desired auxilliary, ungauged Hilbert space as well.
It would be very interesting to directly compare the state-reduction maps from generalized bipartitions to the edge-modes description in discrete systems such as $\mathbb{Z}_n$ lattice gauge theories.
In the context of holography, we might, for example, compute the entropy of a subinterval of a CFT, and compare to the Cardy formula, the replica prescription, the edge modes prescription, and the Ryu-Takayanagi formula, some of which give definite answers and some of which should depend on the particular choice of embedding.

\subsection{Quantum Gravity}

Any realistic theory of quantum gravity must contain states, like our world, which look at low energies and large distances like field-theoretic excitations on top of a fixed spatial background.
That is, there should exist some sectors of the quantum gravity Hilbert space that look like QFTs on curved spacetime.
If quantum gravity is a \emph{bona fide} quantum-mechanical theory that describes more than a single fixed metric, it should contain many more states which look nothing like field theories on fixed backgrounds.
In ascending order of speculation, the theory should certainly include superpositions of geometries (which can be straightforwardly produced experimentally by placing test masses in superpositions, e.g.\ \cite{Page:1981aj}), if its UV completion still has a good notion of spatial backgrounds it should contain heavy or stringy states, and it might contain ``spacetime foam''-like states in which the notion of spacetime breaks down entirely.
Hence we should most likely not expect states with good spacetime descriptions to be simple factors of the full QG Hilbert space \cite{Bao:2017who}, especially if the UV description of gravity is holographic in the manner of AdS/CFT or de Sitter complementarity \cite{Bekenstein:1980jp,Susskind:1993if,Srednicki:1993im,Bao:2017rnv}.

A ``space from Hilbert space'' picture \cite{Cao:2016mst,Bao:2017rnv,Carroll:2018rhc,Pollack:2018yum} in which local spatial degrees of freedom are emergent rather than fundamental would require a detailed picture of exactly how these geometric and field-theoretic degrees of freedom in fact emerge.
In this paper we have attacked precisely this problem in a quantum-mechanical context.
Generalized bipartitions and partial bipartitions are tools for producing reduced states which provide information about degrees of freedom that are not manifest in the full Hilbert space (c.f.\ \cite{Ghosh:2017gtw}, which points out that the set of approximately-localized operators in a subregion of a gravitational theory may not comprise an algebra). 
Interactions between these degrees of freedom and the rest of the theory drive dynamics which may pick out a certain subset as classical observables along the lines of the decoherence program. 
Because quantum cosmology lacks a fixed separation between system and environment, a variational approach is required to find the ``most classical'' bipartitions, or to understand what dynamics lead these preferred observables to look like spacetime variables.

\begin{center} 
 {\bf Acknowledgments}
 \end{center}
 \noindent 
We thank Sean Carroll, Olivia Di Matteo, Ben Michel, Eric Minton, Philippe Sabella-Garnier, Mark Van Raamsdonk, Robert Raussendorf, and Benson Way for useful discussions. We would also like to acknowledge the help of an anonymous user on the \textit{Mathematics Stackexchange} website with the proof of Proposition \ref{prop:MA of single operator }. 
J.P. would like to thank the organizers of the workshop ``Gravity - New perspectives from strings and higher dimensions'' held at the Centro de Ciencias de Benasque Pedro Pascual, where part of this work was completed. A.S. would like to thank the High Energy Physics group at University of British Columbia at Vancouver for hosting and supporting his visit where part of this work was carried out. 
O.K. is supported by the Natural Sciences and Engineering Research Council of Canada.
J.P. is supported in part by the Simons Foundation and in part by the Natural Sciences and Engineering Research Council of Canada. 
A.S. is funded in part by the Walter Burke Institute for Theoretical Physics at Caltech, by DOE grant DE-SC0011632, and by the Foundational Questions Institute.

\bibliographystyle{utphys-modified}
\bibliography{gbp}

\end{document}